\let\oldvec\vec
\let\oldcaption\caption
\documentclass{elsarticle}
\let\vec\oldvec
\let\caption\oldcaption

\usepackage{float}
\usepackage{dsfont}
\usepackage{tikz}
\usetikzlibrary[positioning,decorations.pathmorphing,arrows]
\usepackage[normalem]{ulem}
\usepackage[all]{xy}

\usepackage[english]{babel}

\usepackage{amsmath}
\usepackage{amsthm}
\usepackage{amssymb}
\usepackage{color}
\usepackage{mathrsfs}
\usepackage{paralist}
\usepackage{prooftree}

\usepackage[belowskip=-10pt,aboveskip=10pt]{caption}
\newcommand{\comment}[1]{}
\newcommand{\parless}{\mathrel{\prec}}
\newcommand{\parleq}{\mathrel{\preceq}}
\newcommand{\pargreat}{\mathrel{\succ}}

\theoremstyle{definition}
\newtheorem{theorem}{Theorem}[section]

\newtheorem{definition}[theorem]{Definition}
\newtheorem{example}[theorem]{Example}
\newtheorem{lemma}[theorem]{Lemma}
\newtheorem{proposition}[theorem]{Proposition}
\newtheorem{remark}[theorem]{Remark}

\newcommand{\eg}{{\em e.g.~}}
\newcommand{\ie}{{\em i.e.~}}
\newcommand{\cf}{{\em cf.~}}

\newcommand{\Coloneq}{\ensuremath{::=}}
\newcommand{\dbckslash}{\ensuremath{\setminus\!\!\setminus}}

\newcommand{\eqdef}{\mathrel{\raisebox{-1pt}{\ensuremath{\stackrel{\textit{\tiny{def}}}{=}}}}}
\newcommand{\eqalpha}{\ensuremath{=_{\alpha}}}
\newcommand{\emphdef}[1]{\textbf{\emph{#1}}}

\newcommand{\SysBang}{\ensuremath{\mathcal{U}}}
\newcommand{\SysCBN}{\ensuremath{\mathcal{N}}}
\newcommand{\SysCBV}{\ensuremath{\mathcal{V}}}
\newcommand{\SysTight}{\ensuremath{\mathcal{E}}}

\newcommand{\ctxt}[1]{\ensuremath{\mathtt{#1}}}
\newcommand{\ctxtapp}[2]{\ensuremath{{#1}\langle{#2}\rangle}}

\newcommand{\ctxtleq}{\ensuremath{\subseteq}}

\newcommand{\TermExplicit}{\ensuremath{\mathcal{T}}}
\newcommand{\TermLambda}{\ensuremath{\mathcal{T_{\lambda}}}}

\newcommand{\TermVariable}{\ensuremath{\mathcal{X}}}

\newcommand{\TypeVariable}{\ensuremath{\mathcal{T}\mathcal{V}}}

\newcommand{\exsubs}[2]{\ensuremath{[{#1}\backslash {#2}]}}

\newcommand{\termapp}[2]{\ensuremath{{#1}\,{#2}}}
\newcommand{\termabs}[2]{\ensuremath{\lambda{#1}.{#2}}}
\newcommand{\termbang}[1]{\ensuremath{\mathop{!}{#1}}}
\newcommand{\termder}[1]{\ensuremath{\mathop{\textnormal{der}}{#1}}}
\newcommand{\termsubs}[3]{\ensuremath{{#3}\exsubs{#1}{#2}}}

\newcommand{\functtype}[2]{\ensuremath{{#1}\to{#2}}}
\newcommand{\intertype}[2]{\ensuremath{\multiset{#1}_{#2}}}
\newcommand{\M}{\ensuremath{\mathtt{M}}}

\newcommand{\sequ}[2]{\ensuremath{{#1}\vdash{#2}}}
\newcommand{\sequB}[2]{\ensuremath{{#1}\vdash{#2}}}

\newcommand{\sequN}[2]{\ensuremath{{#1}\vdash{#2}}}
\newcommand{\sequV}[2]{\ensuremath{{#1}\vdash{#2}}}
\newcommand{\assign}[2]{\ensuremath{{#1}:{#2}}}
\newcommand{\derivable}[3]{\ensuremath{{#1}\rhd_{#3}{#2}}}

\newcommand{\ctxtsum}[3]{\ensuremath{{#1}\mathrel{+_{#3}}{#2}}}
\newcommand{\ctxtres}[3]{\ensuremath{{#1}\mathrel{\dbckslash_{#3}}{#2}}}

\newcommand{\dom}[1]{\ensuremath{\mathtt{supp}}(#1)}

\newcommand{\bv}[1]{\ensuremath{\mathtt{bv}({#1})}}
\newcommand{\fv}[1]{\ensuremath{\mathtt{fv}({#1})}}

\newcommand{\size}[1]{\ensuremath{\mathtt{sz}\left({#1}\right)}}
\newcommand{\sizen}[1]{\ensuremath{\mathtt{sz}_{\callbyname}\left({#1}\right)}}
\newcommand{\sizev}[1]{\ensuremath{\mathtt{sz}_{\callbyvalue}\left({#1}\right)}}

\newcommand{\rGen}{\mathtt{a}}
\newcommand{\rrule}[1]{\ensuremath{\mathrel{\mapsto_{#1}}}}
\newcommand{\rewrite}[1]{\rightarrow_{#1}}
\newcommand{\rewriten}[1]{\twoheadrightarrow_{#1}}
\newcommand{\rewritenp}[1]{\rightarrow_{#1}^{\! +}}

\newcommand{\bangweak}{\ensuremath{\mathtt{w}}}
\newcommand{\callbyname}{\ensuremath{\mathtt{n}}}
\newcommand{\callbyvalue}{\ensuremath{\mathtt{v}}}
\newcommand{\weakstg}{\ensuremath{\mathtt{dw}}}

\newcommand{\lored}{\rewrite{\weakstg}}
\newcommand{\loredn}{\rewriten{\weakstg}}

\newcommand{\dB}{\ensuremath{\mathtt{dB}}}

\newcommand{\dBeta}{\ensuremath{\mathtt{dB}}}
\newcommand{\dBang}{\ensuremath{\mathtt{d!}}}
\newcommand{\sBang}{\ensuremath{\mathtt{s!}}}
\newcommand{\sTerm}{\ensuremath{\mathtt{s}}}
\newcommand{\sVal}{\ensuremath{\mathtt{sv}}}

\newcommand{\mStep}{\ensuremath{\mathtt{m}}}
\newcommand{\eStep}{\ensuremath{\mathtt{e}}}

\newcommand{\set}[1]{\ensuremath{\{{#1}\}}}
\newcommand{\multiset}[1]{\ensuremath{[#1]}}

\newcommand{\bsubstitute}[2]{\ensuremath{{\left\{{#1} \backslash{#2}\right\}}}}
\newcommand{\substitute}[3]{\ensuremath{{{#3}\left\{{#1}\backslash{#2}\right\}}}}

\newcommand{\many}[2]{\ensuremath{\left({#1}\right)_{#2}}}

\newcommand{\cbn}[1]{\ensuremath{{#1}^{\mathsf{cbn}}}}
\newcommand{\cbv}[1]{\ensuremath{{#1}^{\mathsf{cbv}}}}

\newcommand{\ruleName}[1]{\ensuremath{\mathtt{({#1})}}}

\newcommand{\ruleBAxiom}{\ruleName{ax}}
\newcommand{\ruleBArrowE}{\ruleName{app}}
\newcommand{\ruleBArrowI}{\ruleName{abs}}
\newcommand{\ruleBBang}{\ruleName{bg}}
\newcommand{\ruleBDer}{\ruleName{dr}}
\newcommand{\ruleBESubs}{\ruleName{es}}
\newcommand{\ruleDAxiom}{\ruleName{ax_{c}}}
\newcommand{\ruleDArrowE}{\ruleName{ae_{c1}}}
\newcommand{\ruleDArrowI}{\ruleName{ai_{c}}}
\newcommand{\ruleBTApp}{\ruleName{ae_{c2}}}
\newcommand{\ruleDBang}{\ruleName{bg_{c}}}
\newcommand{\ruleDDer}{\ruleName{dr_{c}}}
\newcommand{\ruleDESubs}{\ruleName{es_{c}}}
\newcommand{\ruleNAxiom}{\ruleName{ax_n}}
\newcommand{\ruleNArrowE}{\ruleName{app_n}}
\newcommand{\ruleNArrowI}{\ruleName{abs_n}}
\newcommand{\ruleNESubs}{\ruleName{es_n}}
\newcommand{\ruleTArrowE}{\ruleName{ae_{p}}}
\newcommand{\ruleTArrowI}{\ruleName{ai_{p}}}
\newcommand{\ruleTBang}{\ruleName{bg_{p}}}
\newcommand{\ruleTDer}{\ruleName{dr_{p}}}
\newcommand{\ruleTESubs}{\ruleName{es_{p}}}
\newcommand{\ruleVAxiom}{\ruleName{ax_v}}
\newcommand{\ruleVArrowE}{\ruleName{app_v}}
\newcommand{\ruleVArrowI}{\ruleName{abs_v}}
\newcommand{\ruleVESubs}{\ruleName{es_v}}
\newcommand{\Rule}[3]{
    \prooftree
         #1
    \justifies  
         #2
    \thickness=0.05em
    \using
         #3
    \endprooftree}

\newcommand{\ih}{{\em i.h.~}}

\newcommand{\wsize}[1]{\ensuremath{|{#1}|_\ensuremath{\mathtt{w}}}}
\newcommand{\simplesize}[1]{\ensuremath{|\!|{#1}|\!|}}
\newcommand{\nsize}[1]{\ensuremath{|{#1}|_\ensuremath{\mathtt{n}}}}   
\newcommand{\valsize}[1]{\ensuremath{|{#1}|_\callbyvalue}}

\newcommand{\sequT}[5]{\ensuremath{{#1}\vdash^{({#3},{#4},{#5})}{#2}}}

\newcommand{\typeabs}{\ensuremath{\mathtt{a}}}
\newcommand{\typebang}{\ensuremath{\mathtt{b}}}
\newcommand{\typeneutral}{\ensuremath{\mathtt{n}}}
\newcommand{\typetight}{\ensuremath{\mathtt{tt}}}

\newcommand{\iptight}{\ensuremath{\mathsf{tight}}}

\newcommand{\ptight}[1]{\ensuremath{\iptight(#1)}}

\newcommand{\ipabs}{\ensuremath{\mathsf{abs}}}
\newcommand{\ipvar}{\ensuremath{\mathsf{var}}}
\newcommand{\ipapp}{\ensuremath{\mathsf{app}}}
\newcommand{\ipbang}{\ensuremath{\mathsf{bang}}}
\newcommand{\ipnrml}{\ensuremath{\mathsf{no}_{\bangweak}}}
\newcommand{\ipntrl}{\ensuremath{\mathsf{ne}_{\bangweak}}}
\newcommand{\ipnbang}{\ensuremath{\mathsf{nb}_{\bangweak}}}
\newcommand{\ipnabs}{\ensuremath{\mathsf{na}_{\bangweak}}}

\newcommand{\pabs}[1]{\ensuremath{\ipabs(#1)}}
\newcommand{\pvar}[1]{\ensuremath{\ipvar(#1)}}
\newcommand{\papp}[1]{\ensuremath{\ipapp(#1)}}
\newcommand{\pbang}[1]{\ensuremath{\ipbang(#1)}}
\newcommand{\pnrml}[1]{\ensuremath{{#1}\in\ipnrml}}
\newcommand{\pntrl}[1]{\ensuremath{{#1}\in\ipntrl}}
\newcommand{\pnbang}[1]{\ensuremath{{#1}\in\ipnbang}}
\newcommand{\pnabs}[1]{\ensuremath{{#1}\in\ipnabs}}

\newcommand{\icfnrml}{\ensuremath{\mathsf{no}_{\cfz}}}
\newcommand{\icfntrl}{\ensuremath{\mathsf{ne}_{\cfz}}}
\newcommand{\icfnbang}{\ensuremath{\mathsf{nb}_{\cfz}}}
\newcommand{\icfnabs}{\ensuremath{\mathsf{na}_{\cfz}}}
\newcommand{\cfnrml}[1]{\ensuremath{{#1}\in\icfnrml}}
\newcommand{\cfntrl}[1]{\ensuremath{{#1}\in\icfntrl}}
\newcommand{\cfnbang}[1]{\ensuremath{{#1}\in\icfnbang}}
\newcommand{\cfnabs}[1]{\ensuremath{{#1}\in\icfnabs}}

\newcommand{\cbeta}{\ensuremath{\mathit{b}}}
\newcommand{\cexp}{\ensuremath{\mathit{e}}}
\newcommand{\csize}{\ensuremath{\mathit{s}}}

\newcommand{\sL}{\mathds{L}}
\newcommand{\iI}{i \in I}
\newcommand{\jJ}{j \in J}
\newcommand{\sig}{\sigma}
\newcommand{\cfz}{{\tt wcf}}
\newcommand{\emul}{\intertype{\, }{}}
\newcommand{\Gam}{\Gamma}

\newcommand{\CBNNF}{\ensuremath{\mathsf{no}_{\callbyname}}}
\newcommand{\CBVNF}{\ensuremath{\mathsf{no}_{\callbyvalue}}}
\newcommand{\HCBNNF}{\ensuremath{\mathsf{ne}_{\callbyname}}}
\newcommand{\HCBVNF}{\ensuremath{\mathsf{ne}_{\callbyvalue}}}
\newcommand{\VarHCBVNF}{\ensuremath{\mathsf{vr}_{\callbyvalue}}}
\newcommand{\BangRev}{\lambda !}

\newcommand{\Kterm}{{\tt K}}
\newcommand{\id}{{\tt I}}

\newcommand{\Kterml}{{\tt K}_\lambda}
\newcommand{\idl}{{\tt I}_\lambda}
\newcommand{\Deltal}{{\Delta}_\lambda}
\newcommand{\Omegal}{{\Omega}_\lambda}

\begin{document}

\begin{frontmatter}
\title{The Bang Calculus Revisited}

\author[1]{Antonio Bucciarelli}
\ead{buccia@irif.fr}

\author[1,2]{Delia Kesner\corref{cor1}}
\ead{kesner@irif.fr}

\author[3]{Alejandro R\'ios}
\ead{rios@dc.uba.ar}

\author[4]{Andr\'es Viso\corref{cor1}}
\ead{andres-ezequiel.viso@inria.fr}

\affiliation[1]{
  organization={Université Paris Cité, CNRS, IRIF},
  country={France}
}
\affiliation[2]{
  organization={Institut Universitaire de France},
  country={France}
}
\affiliation[3]{
  organization={Universidad de Buenos Aires},
  country={Argentina}
}
\affiliation[4]{
  organization={Inria},
  country={France}
}

\cortext[cor1]{Corresponding author}

\begin{abstract}
Call-by-Push-Value (CBPV) is a programming paradigm subsuming both Call-by-Name
(CBN) and Call-by-Value (CBV) semantics. The essence of this paradigm is
captured by the Bang Calculus, a (concise) term language connecting CBPV and
Linear Logic.

This paper presents a revisited version of the Bang Calculus, called
$\BangRev$, enjoying some important properties missing in the original
formulation. Indeed, the new calculus integrates permutative conversions to
unblock value redexes while being confluent at the same time. A second
contribution is related to non-idempotent types. We provide a quantitative type
system for our $\BangRev$-calculus, and we show that the length of the (weak)
reduction of a typed term to its normal form \emph{plus} the size of this
normal form is bounded by the size of its type derivation. We also explore the
properties of this type system with respect to CBN/CBV translations. We keep
the original CBN translation from $\lambda$-calculus to the Bang Calculus,
which preserves normal forms and is sound and complete with respect to the
(quantitative) type system for CBN. However, in the case of CBV, we reformulate
both the translation and the type system to restore two main properties:
preservation of normal forms and completeness.
Last but not least, the
quantitative system is refined to a \emph{tight} one, which transforms the
previous upper bound on the length of reduction to normal form plus its size
into two independent \emph{exact} measures for them.


\end{abstract}


\end{frontmatter}

\section{Introduction}

{\bf Call-by-Push-Value.}
The Call-by-Push-Value (CBPV) paradigm, introduced by
P.B.~Levy~\cite{Levy04,Levy06}, distinguishes between values and computations
under the slogan \emph{``a value is, a computation does''}. It subsumes the
$\lambda$-calculus by adding some primitives that allow to capture both the
Call-by-Name (CBN) and Call-by-Value (CBV) semantics. CBN is a lazy strategy
that consumes arguments without any preliminary evaluation, potentially
duplicating work, while CBV is greedy, always computing arguments disregarding
whether they are used or not, which may prevent a
normalising term from terminating, \eg $\termapp{(\termabs{x}{I})}{\Omega}$,
where $I = \termabs{x}{x}$ and $\Omega =
\termapp{(\termabs{x}{\termapp{x}{x}})}{(\termabs{x}{\termapp{x}{x}})}$.

Essentially, CBPV introduces unary primitives $\mathtt{thunk}$ and
$\mathtt{force}$. The former freezes the execution of a term (\ie it is not
allowed to compute under a $\mathtt{thunk}$) while the latter fires again a
frozen term. Informally,
$\termapp{\mathtt{force}}{(\termapp{\mathtt{thunk}}{t})}$ is semantically
equivalent to $t$. Resorting to the paradigm slogan, $\mathtt{thunk}$ turns a
computation into a value, while $\mathtt{force}$ does the opposite. Thus, CBN
and CBV are captured by conveniently labelling a $\lambda$-term using
$\mathtt{force}$ and $\mathtt{thunk}$ to pause/resume the evaluation of a
subterm depending on whether it is an argument (CBN) or a function (CBV). In
doing so, CBPV provides a unique formalism capturing two distinct
$\lambda$-calculi strategies, thus allowing to study operational and
denotational semantics of CBN and CBV in a unified framework.

{\bf Bang calculus.}  T.~Ehrhard~\cite{Ehrhard16} introduced a typed
calculus, that can be seen as a variation of CBPV, to establish a
relation between this paradigm and Linear Logic (LL).  A simplified
version of this formalism is later dubbed Bang calculus~\cite{EhrhardG16},
showing in particular how CBPV captures the CBN and CBV semantics of
$\lambda$-calculus via Girard's translations of intuitionistic logic into
LL. A further step in this direction~\cite{ChouquetT20} uses Taylor
expansion~\cite{EhrhardR08} in the Bang Calculus to approximate terms in
CBPV. The Bang calculus is essentially an extension of $\lambda$-calculus
with two new constructors, namely \emph{bang} ($\termbang{\!}$) and
\emph{dereliction} ($\termder{\!}$), together with the reduction rule
$\termder{(\termbang{t})} \rrule{} t$. There are two notions of
reduction for the Bang calculus, depending on whether it is allowed to
reduce under a bang constructor or not. They are called \emph{strong}
and \emph{weak reduction} respectively. Indeed, it is weak reduction
that makes bang/dereliction play the role of the primitives
$\mathtt{thunk}$/$\mathtt{force}$. Hence, these modalities are
essential to capture the essence behind the CBN--CBV duality. A
similar approach appears in~\cite{SantoPU19}, studying (simply typed)
CBN and CBV translations into a fragment of IS4, recast as a very
simple $\lambda$-calculus equipped with an indeterminate lax monoidal
comonad.

{\bf Non-Idempotent Types.}
Intersection types, pioneered by~\cite{CoppoD80,CoppoDV81}, can be seen as a
syntactical tool to denote programs. They are invariant under the equality
generated by the evaluation rules, and type all and only all normalising terms.
They were originally defined as \emph{idempotent} types, so that the equation
$\sigma \cap \sigma = \sigma$ holds, thus preventing any use of the
intersection constructor to count resources. On the other hand,
\emph{non-idempotent} types, pioneered by~\cite{Gardner94}, are inspired
by LL and can be seen as a syntactical formulation of its relational
model~\cite{Girard88,BucciarelliE01}. This connection suggests a
\emph{quantitative} typing tool, being able to specify properties related to
the consumption of resources, a remarkable
investigation pioneered by  de Carvalho's seminal PhD
thesis~\cite{Carvalho:thesis} (see also~\cite{Carvalho18}). Non-idempotent
types have also been used to provide characterisations of complexity
classes~\cite{BenedettiR16}. Several papers explore the qualitative and
quantitative aspects of non-idempotent types for different higher order
languages, as for example Call-by-Name, Call-by-Need and Call-by-Value
$\lambda$-calculi, as well as extensions to Classical Logic. Some references
are~\cite{BucciarelliKV17,Ehrhard12,AccattoliGL19,AccattoliG18,KesnerV17}.
Other relational models were directly defined in the more general context of
LL, rather than in the
$\lambda$-calculus~\cite{Carvalho16,GuerrieriPF16,CarvalhoPF11,CarvalhoF16}.

An interesting recent research topic concerns the use of non-idempotent types
to provide \emph{bounds} of reduction lengths. More precisely, the size of type
derivations has often been used as an \emph{upper bound} to the length of
different evaluation
strategies~\cite{PaganiR10,Ehrhard12,Kesner16,BucciarelliKV17,KesnerV14,KesnerV17}.
A key notion behind these works is that when $t$ evaluates to $t'$, then the
size of the type derivation of $t'$ is smaller than the one of $t$, thus the
size of type derivations provides an  upper bound for the \emph{length} of the
reduction to a normal form as well as for the \emph{size} of this normal form.

A crucial point to obtain \emph{exact bounds}, instead of upper bounds, is to
consider only \emph{minimal} type derivations, as the ones
in~\cite{Carvalho:thesis,BernadetL13,CarvalhoPF11}. Another approach was taken
in~\cite{AccattoliGK18}, which uses an appropriate notion of \emph{tightness}
to implement minimality, a technical tool adapted to
Call-by-Value~\cite{Guerrieri18,AccattoliG18,KesnerV22},
Call-by-Need~\cite{AccattoliGL19}, pattern-matching languages~\cite{AlvesKV19},
and control operators~\cite{KesnerV20}.

\subsection{Contributions and Related Works}

This article presents a reformulation of the untyped Bang calculus, and
proposes a quantitative study of it by means of non-idempotent types.

{\bf The Untyped Reduction.}
The Bang calculus in~\cite{Ehrhard16} suffers from the absence of
\emph{permutative conversions}~\cite{Regnier94,CarraroG14}, making some redexes
syntactically blocked when open terms are considered. A consequence of
this approach is that there are some normal forms that are semantically
equivalent to non-terminating programs, a situation which is clearly unsound.
This is repaired in~\cite{EhrhardG16} by adding permutative conversions
specified by means of $\sigma$-reduction rules, which are crucial to unveil
hidden (value) redexes. However, this approach presents a  major drawback since
the resulting combined reduction relation is not confluent (Page 6 in~\cite{EhrhardG16} or Example~\ref{ex:not-confluent} below).

Our revisited Bang calculus, called $\BangRev$, fixes these two problems at the
same time. Indeed, the syntax is enriched with explicit substitutions, and
$\sigma$-equivalence is integrated in the primary reduction system by using the
\emph{distance} paradigm~\cite{AccattoliK10}, without any need to unveil hidden
redexes by means of an independent relation. This approach restores
confluence.

{\bf The Untyped CBN and CBV Encodings. }
CBN and CBV (untyped) translations are extensively studied
in~\cite{GuerrieriM18,Danos90,MaraistOTW99}, where the authors establish two
encodings $cbn$ and $cbv$, from untyped $\lambda$-terms into untyped terms of
the Bang calculus, such that when $t$ reduces to $u$ in CBN (resp. CBV),
$cbn(t)$ reduces to $cbn(u)$ (resp. $cbv(t)$ reduces to $cbv(u)$) in the Bang
calculus. However, CBV normal forms in $\lambda$-calculus are not necessarily
translated to normal forms in the Bang calculus.

We extend to explicit substitutions the original CBN translation from
$\lambda$-calculus to the Bang calculus, which preserves normal forms,
and we reformulate the CBV one in such a way that, in contrast
to~\cite{GuerrieriM18}, our CBV translation does preserve normal
forms. In order to achieve the preservation of normal forms,
we use a non-compositional CBV translation based on
superdevelopments~\cite{Hyland78,vanoostrom21}, which performs
reduction of created redexes during the translation.  Our revisited notion of
reduction inside the Bang calculus naturally encodes \emph{head} CBN,
\ie reduction does not take place in arguments of applications, as
well as \emph{open} CBV, \ie reduction does not take place inside
abstractions. More precisely, the $\BangRev$-calculus encodes head CBN
and open CBV specified by means of explicit substitutions (see for
example~\cite{AccattoliP12}). These two notions are dual: head CBN
forbids reduction inside arguments, which are translated to bang
terms, while open CBV forbids reduction under $\lambda$-abstractions,
also translated to bang terms.

{\bf The Typed System.}
Starting from the relational model for the Bang calculus proposed
in~\cite{GuerrieriM18}, we propose a type system for the $\BangRev$-calculus,
called $\SysBang$, based on non-idempotent intersection types. System
$\SysBang$ is able to fully \emph{characterise} normalisation, in the sense
that a term $t$ is $\SysBang$-typable if and only if $t$ is normalising. More
interestingly, we show that system $\SysBang$ has also a quantitative flavour,
in the sense that the length of any reduction sequence from $t$ to normal form
\emph{plus} the size of this normal form is \emph{bounded} by the size of the
type derivation of $t$. We show that system $\SysBang$ also captures the
standard non-idempotent intersection type system $\SysCBN$ for CBN, in the
sense that a $\lambda$-term $t$ is typable in $\SysCBN$ if and only if its
translation $cbn(t)$ is typable in $\SysBang$ with the same type and context.
Concerning CBV, we define a new type system $\SysCBV$ and we show that our CBV
translation enjoys the same property. System $\SysCBV$ characterises
termination of open CBV, in the sense that $t$ is typable in $\SysCBV$ if and
only if $t$ is terminating in open CBV. This can be seen as another
(collateral) contribution of this article. Moreover, the CBV embedding
in~\cite{GuerrieriM18} is not complete with respect to their type system for
CBV. System $\SysCBV$ recovers completeness (left as an open question
in~\cite{GuerrieriM18}). Finally, an alternative CBV encoding of typed terms is
proposed. This encoding is not only sound and complete, but now enjoys
preservation of normal-forms.

{\bf A Refinement of the Type System Based on Tightness.}
A major observation concerning $\beta$-reduction in $\lambda$-calculus (and
therefore in the Bang calculus) is that the size of normal forms can be
exponentially bigger than the number of steps to these normal forms. This means
that bounding the sum of these two integers \emph{at the same} time is too
rough, not very relevant from a quantitative point of view. Following ideas
in~\cite{Carvalho:thesis,BernadetL13,AccattoliGK18}, we go beyond upper bounds.
Indeed, another major contribution of this article is the refinement of the
non-idempotent type system $\SysBang$ to another type system $\SysTight$,
equipped with constants and counters, together with an appropriate notion of
\emph{tightness} (\ie minimality). This new formulation fully exploits the
quantitative aspect of the system, in such a way that \emph{upper bounds}
provided by system $\SysBang$ are refined now into \emph{independent exact
  bounds} for time and space. More precisely, we show that a term $t$ admits 
  a tight type derivation  with counters $(\cbeta,\cexp,\csize)$  if and
  only if $t$ is normalisable in $(\cbeta+\cexp)$-steps and its normal form has
\emph{size} $\csize$. Therefore, exact
measures concerning the \emph{dynamic} behaviour of $t$, are extracted from a
\emph{static} (tight) typing property of $t$.

This is a revised and extended version of the authors' article
  \cite{BucciarelliKRV20}. \\

{\bf Road-map.}
Sec.~\ref{s:bang} introduces the $\BangRev$-calculus. Sec.~\ref{s:system}
presents the sound and complete type system $\SysBang$.
Sec.~\ref{s:cbname-cbvalue} discusses (untyped and typed) CBN/CBV translations.
In Sec.~\ref{s:tight} we refine system $\SysBang$ into system $\SysTight$, and
we prove soundness and completeness. Conclusions and future work are discussed
in Sec.~\ref{s:conclusion}.


\section{The Bang Calculus Revisited}
\label{s:bang}

This  section  presents a revisited (conservative) extension  of the original
Bang calculi~\cite{Ehrhard16,EhrhardG16}, called $\BangRev$. From a syntactical
point of view, we just add explicit substitution operators. From an
operational point of view, we use \emph{reduction at a
distance}~\cite{AccattoliK10}, thus integrating permutative conversions without
jeopardising confluence (see the discussion below).

Given a countably infinite set $\TermVariable$ of variables $x, y, z, \ldots$
we consider the following grammar for terms, denoted by $\TermExplicit$,  and
contexts:
\begin{center}
\begin{tabular}{rrcll}
\textbf{(Terms)}         & $t,u$      & $\Coloneq$ & $x \in \TermVariable \mid \termapp{t}{u} \mid \termabs{x}{t} \mid \termbang{t} \mid \termder{t} \mid \termsubs{x}{u}{t}$ \\
\textbf{(List Contexts)} & $\ctxt{L}$ & $\Coloneq$ & $\Box \mid \termsubs{x}{t}{\ctxt{L}}$ \\ 
\textbf{(Contexts)}      & $\ctxt{C}$ & $\Coloneq$ & $\Box \mid \termapp{\ctxt{C}}{t} \mid \termapp{t}{\ctxt{C}} \mid \termabs{x}{\ctxt{C}} \mid \termbang{\ctxt{C}} \mid \termder{\ctxt{C}} \mid \termsubs{x}{u}{\ctxt{C}} \mid \termsubs{x}{\ctxt{C}}{t}$ \\
\textbf{(Weak Contexts)} & $\ctxt{W}$ & $\Coloneq$ & $\Box \mid \termapp{\ctxt{W}}{t} \mid \termapp{t}{\ctxt{W}} \mid \termabs{x}{\ctxt{W}} \mid \termder{\ctxt{W}} \mid \termsubs{x}{u}{\ctxt{W}} \mid \termsubs{x}{\ctxt{W}}{t}$
\end{tabular}
\end{center}
Terms of the form $\termsubs{x}{u}{t}$ are \emphdef{closures}, and
$\exsubs{x}{u}$ is called an \emphdef{explicit substitution}
(ES). Special terms are $\id = \termabs{z}{z}$, $\Kterm =
\termabs{x}{\termabs{y}{x}}$, $\Delta =
\termabs{x}{\termapp{x}{\termbang{x}}}$, and $\Omega =
\termapp{\Delta}{\termbang{\Delta}}$. Weak contexts do not allow the
symbol $\Box$ to occur inside the bang construct. This is similar to
weak contexts in $\lambda$-calculus, where $\Box$ cannot occur inside
$\lambda$-abstractions.  We will see in Sec.~\ref{s:cbname-cbvalue}
that weak reduction in the $\BangRev$-calculus perfectly captures head
reduction in CBN, disallowing reduction inside arguments, as well as
open CBV, disallowing reduction inside abstractions. We use
$\ctxtapp{\ctxt{C}}{t}$ (resp. $\ctxtapp{\ctxt{W}}{t}$ and
$\ctxtapp{\ctxt{L}}{t}$) for the term obtained by replacing the hole
$\Box$ of $\ctxt{C}$ (resp. $\ctxt{W}$ and $\ctxt{L}$) by $t$.  In
order to increase readability we use the following notational
conventions: the application is left associative and has higher
priority than the $\lambda$-abstraction, so that for instance we may
write $\termabs{x}{\termapp{\termapp{t}{u}}{r}}$ for
$\termabs{x}{(\termapp{(\termapp{t}{u})}{r})}$. The unary operators
have higher priority than the binary ones, so that for instance
$\termapp{\termbang{t}}{u}$ reads $\termapp{(\termbang{t})}{u}$. Also,
the explicit substitution operator has higher priority than the other
binary operators. Nevertheless, we use these notations with parcimony
and we add parenthesis only whenever they could be misleading.  The
notions of \emphdef{free} and \emphdef{bound} variables are defined as
expected, in particular, $\fv{\termsubs{x}{u}{t}} \eqdef \fv{t}
\setminus \set{x} \cup \fv{u}$, $\fv{\termabs{x}{t}} \eqdef \fv{t}
\setminus \set{x}$, $\bv{\termsubs{x}{u}{t}} \eqdef \bv{t} \cup
\set{x} \cup \bv{u}$ and $\bv{\termabs{x}{t}} \eqdef \bv{t} \cup
\set{x}$. We extend the standard notion of
\emphdef{$\alpha$-conversion}~\cite{Barendregt84} to ES, as expected,
so that bound variables can always be renamed.  Thus
  \eg\ $\termabs{x}{x} \eqalpha \termabs{y}{y}$ and
  $\termsubs{x}{z}{x} \eqalpha \termsubs{y}{z}{y}$.  We use
$\substitute{x}{u}{t}$ to denote the \emphdef{meta-level} substitution
operation, \ie all the free occurrences of the variable $x$ in the
term $t$ are replaced by $u$. This operation is defined, as usual,
modulo $\alpha$-conversion. We use two special predicates to
distinguish abstractions and bang terms possibly affected by a list of
explicit substitutions. Indeed, $\pabs{t}$ holds iff $t =
\ctxtapp{\ctxt{L}}{\termabs{x}{t'}}$ for some $\ctxt{L}$ and
$\pbang{t}$ holds iff $t = \ctxtapp{\ctxt{L}}{\termbang{t'}}$
for some $\ctxt{L}$. Finally, we define the following
  notion of size for terms of the $\BangRev$-calculus:
  \begin{definition}
    The \emphdef{$\bangweak$-size} of $t\in \TermExplicit$ is inductively
    defined as follows:
    \[ \begin{array}{lll@{\qquad\qquad}lll}
         \wsize{x} & \eqdef & 0 &
        \wsize{\termbang{t}} & \eqdef & 0 \\
         \wsize{\termapp{t}{u}} & \eqdef & 1 + \wsize{t} + \wsize{u} &
         \wsize{\termder{t}}  & \eqdef & 1 + \wsize{t} \\
         \wsize{\termabs{x}{t}} & \eqdef & 1 + \wsize{t} &
         \wsize{\termsubs{x}{u}{t}} & \eqdef & 1 + \wsize{t} + \wsize{u}\\
       \end{array} \]
   \end{definition}

The \emphdef{$\BangRev$-calculus} is given by the set of terms $\TermExplicit$
and the \emphdef{(weak) reduction relation} $\rewrite{\bangweak}$, which is
defined as the \emph{union} of $\rewrite{\dBeta}$ (\texttt{d}istant
\texttt{B}eta), $\rewrite{\sBang}$ (\texttt{s}ubstitute bang) and
$\rewrite{\dBang}$ (\texttt{d}istant bang), defined  respectively as the
closure by weak contexts $\ctxt{W}$ of the following rewriting rules: \[
\begin{array}{rcl}
\termapp{\ctxtapp{\ctxt{L}}{\termabs{x}{t}}}{u}   & \rrule{\dBeta} & \ctxtapp{\ctxt{L}}{\termsubs{x}{u}{t}} \\
\termsubs{x}{\ctxtapp{\ctxt{L}}{\termbang{u}}}{t} & \rrule{\sBang} & \ctxtapp{\ctxt{L}}{\substitute{x}{u}{t}} \\
\termder{(\ctxtapp{\ctxt{L}}{\termbang{t}})}      & \rrule{\dBang} & \ctxtapp{\ctxt{L}}{t}
\end{array} \]
We assume that all these rules avoid capture of free variables.

  \begin{example}
\label{example:t0}
Let $t_0 =
\termapp{\termapp{\termder{(\termbang{\Kterm})}}{(\termbang{\id})}}{(\termbang{\Omega})}$.
Then, \[t_0 \rewrite{\dBang}
\termapp{\termapp{\Kterm}{(\termbang{\id})}}{(\termbang{\Omega})}
\rewrite{\dBeta}
\termapp{\termsubs{x}{\termbang{\id}}{(\termabs{y}{x})}}{(\termbang{\Omega})}
\rewrite{\dBeta}
\termsubs{x}{\termbang{\id}}{\termsubs{y}{\termbang{\Omega}}{x}}
\rewrite{\sBang}
\termsubs{x}{\termbang{\id}}{x} \rewrite{\sBang} \id \]
Observe that the second $\dBeta$-step uses action at a distance, where
$\ctxt{L}$ is $\termsubs{x}{\termbang{\id}}{\Box}$.
\end{example}

Given the translation of
the Bang Calculus into LL proof-nets~\cite{Ehrhard16}, we refer to
$\dBeta$-steps as $\mStep$-steps ($\mStep$ultiplicative) and
$\{\sBang,\dBang\}$-steps
as $\eStep$-steps ($\eStep$xponential).

Observe that reduction is \emph{at a distance}, in the sense that the
list context $\ctxt{L}$ in the rewriting rules allows the main
constructors involved in these rules to be separated by an arbitrary
finite list of substitutions.  This new formulation integrates
permutative conversions inside the main (logical) reduction rules of
the calculus, in contrast to~\cite{EhrhardG16} which treats these
conversions by means of a set of independent $\sigma$-rewriting rules,
thus inheriting many drawbacks. More precisely, in the first
formulation of the Bang calculus~\cite{Ehrhard16}, there are hidden
(value) redexes that block reduction, thus creating a mismatch between
normal terms that are semantically non-terminating. The second
formulation in~\cite{EhrhardG16} recovers soundness, by integrating a
notion of $\sigma$-equivalence which is crucial to unveil hidden
redexes and ill-formed terms (called \emph{clashes})\footnote{Indeed,
  there exist clash-free terms in normal form that are
  $\sigma$-reducible to normal terms with clashes, see the
    definition of clash term at the end of Section~\ref{s:bang}, \eg
  $R =
  \termder{(\termapp{(\termabs{y}{\termabs{x}{z}})}{(\termapp{\termder{(y)}}{y})})}
  \equiv_{\sigma}
  \termder{(\termabs{x}{\termapp{(\termabs{y}{z})}{(\termapp{\termder{(y)}}{y})}})}$.}.
However, adding $\sigma$-reduction to the logical reduction rules does
not preserve confluence. Our notion of reduction addresses
these two issues at the same time\footnote{In particular, the
  term $R$ is not in normal form in our framework, and it reduces to a
  clash term in normal form which is filtered by the type system,
  see Lemma \ref{l:clashes-do-not-type-Bang}.}: it
integrates permutative conversions and is confluent
(Theorem~\ref{t:confluence}). \comment{A consequence of our
  formulation is that redexes in the original
  calculus~\cite{EhrhardG16} are still redexes in our framework, but
  not necessarily the other way around.}

  \begin{example}
    \label{ex:not-confluent}
    The following example is a simplified version of the
    one in \cite{EhrhardG16} showing that their calculus is not
    confluent.  There are two permutative conversions involved in this example: 
    $$ \begin{array}{llll}
         \termapp{(\termapp{(\termabs{x}{s})}{r})}{u} & \mapsto_{\sigma_1} &
    \termapp{(\termabs{x}{\termapp{s}{u}})}{r} & x\not\in\fv{u} \\
    \termapp{(\termabs{y}{\termabs{x}{s}})}{r} & \mapsto_{\sigma_2}  &
    \termabs{x}{(\termapp{\termabs{y}{s})}{r}} & x\not\in\fv{r}\cup\{y\}
    \end{array} $$
    The term 
    $t=\termapp{\termapp{(\termabs{x}{\termabs{y}{z}})}{(\termder{x'})}}
    {(\termder{y'})}$ contains both a  $\sigma_1$ and a $\sigma_2$
    redex: $$t\rightarrow_{\sigma_1}\termapp{(\termabs{x}({\termapp{
    \termabs{y}{z})    }{(\termder{y'})}})}{(\termder{x'})}=t_1
    $$ 
    $$t\rightarrow_{\sigma_2}\termapp{(\termabs{y}({\termapp{
    \termabs{x}{z})    }{(\termder{x'}}}))}{(\termder{y'})}=t_2
    $$ 
    where $\rightarrow_{\sigma_i}$ is the closure by weak contexts of
    $\mapsto_{\sigma_i}$ for $i=1,2$.  The terms $t_1$ and $t_2$, which are different,  are normal forms
    in~\cite{EhrhardG16}, so that confluence is lost.
    In the $\BangRev$-calculus, however,  we
    have:
    $$t\rewrite{\dBeta}\termapp{\termsubs{x}{\termder{x'}}{(\termabs{y}{z})}}{(\termder{y'})}\rewrite{\dBeta}
    \termsubs{x}{\termder{x'}}{\termsubs{y}{\termder{y'}}{z}} = t_0
    $$
    where  $t_0$       is the (unique) normal form of $t$.
  \end{example}

We write $\rewriten{\bangweak}$ (resp. $\rewritenp{\bangweak}$) for the reflexive-transitive (resp. transitive) closure of
$\rewrite{\bangweak}$. We write $t \rewriten{\bangweak}^{(b,e)} u$ if $t
\rewriten{\bangweak} u$ using $b$ $\dBeta$ steps and $e$
$\{\sBang, \dBang\}$-steps.

The reduction relation $\rewrite{\bangweak}$ enjoys a kind of
(weak) diamond property,
\ie one-step divergence can be closed in one step  if the diverging  terms
are different, since $\rewrite{\bangweak}$ is not reflexive. Otherwise stated,
the reflexive closure of $\rewrite{\bangweak}$ enjoys the strong diamond
property.

\begin{lemma}
\label{l:diamond}
If $t \rewrite{p_1} t_1$ and $t \rewrite{p_2} t_2$ where $t_1 \neq
t_2$ and $p_1, p_2 \in \{\dBeta, \sBang, \dBang\}$, then there exists $t_3$
such that $t_1 \rewrite{p_2} t_3$ and $t_2 \rewrite{p_1} t_3$. 
\end{lemma}

\begin{proof}
The proof is by induction on $t$. To make the notations
easier to read, we write $u\sL$
instead of $\ctxtapp{\ctxt{L}}{u}$, where $\sL$ is the list of substitutions of
the context $\ctxt{L}$. 

We only show the two interesting cases (root reduction with superposition of redexes), all the other ones
being straightforward:
\begin{itemize}
  \item $t = ((\lambda x'. t')\sL_1 \exsubs{x}{(\termbang{u})\sL_2} \sL_3) u'
  \rewrite{\dBeta} \termsubs{x'}{u'}{t'}\sL_1 \exsubs{x}{(\termbang{u})\sL_2} \sL_3
  = t _1$ and $t \rewrite{\sBang} (((\lambda x'. t')\sL_1)\bsubstitute{x}{u} \sL_2 \sL_3) u'
  = ((\lambda x'. t'\bsubstitute{x}{u}) (\sL_1 \bsubstitute{x}{u}) \sL_2 \sL_3) u' = t_2$.
  By $\alpha$-conversion we can assume $x \notin \fv{u'}$ so that
  by defining \[
  t_3 = \termsubs{x'}{u'}{t\bsubstitute{x}{u}} (\sL_1 \bsubstitute{x}{u})\sL_2 \sL_3 = \termsubs{x'}{u'}{t}\bsubstitute{x}{u} (\sL_1 \bsubstitute{x}{u})\sL_2 \sL_3
  \] we can close the diagram as follows: $t_1
  \rewrite{\sBang} t_3$ and $t_2 \rewrite{\dBeta} t_3$.
  
  \item $t = \termder{((\termbang{u'})\sL_1\exsubs{x}{(\termbang{u})\sL_2} \sL_3)}
  \rewrite{\dBang} u' \sL_1\exsubs{x}{(\termbang{u})\sL_2} \sL_3 = t_1$ and
  also $t \rewrite{\sBang} \termder{((\termbang{u'})\sL_1\bsubstitute{x}{u}\sL_2 \sL_3)}
  = t_2$. We close the diagram with $t_1 \rewrite{\sBang} t_3$ and $t_2
  \rewrite{\dBang} t_3$, where $t_3 = (u' \sL_1)\bsubstitute{x}{u} \sL_2 \sL_3 =
  u' \bsubstitute{x}{u} \sL_1 \bsubstitute{x}{u} \sL_2 \sL_3$. 
\end{itemize}
\end{proof}

The result above does not hold if reductions are allowed inside arbitrary
contexts. Consider for instance the term $t =
\termsubs{x}{\termbang{(\termapp{\id}{\termbang{\id}}})}{(\termapp{x}{\termbang{x}})}$.
We have $t \rewrite{\sBang}
\termapp{(\termapp{\id}{\termbang{\id}})}{\termbang{(\termapp{\id}{\termbang{\id}})}}$
and, if we allow the reduction in $t$ of the $\dBeta$-redex
$\termapp{\id}{\termbang{\id}}$ appearing banged inside the explicit
substitution, we get $t \rewrite{\dBeta}
\termsubs{x}{\termbang{\termsubs{z}{\termbang \id}{z}}}{(\termapp{x}{\termbang{x}})}$.
Now, the term
$\termsubs{x}{\termbang{\termsubs{z}{\termbang \id}{z}}}{(\termapp{x}{\termbang{x}})}$
${\sBang}$-reduces to
$\termapp{\termsubs{z}{\termbang{\id}}{z}}{\termbang{(\termsubs{z}{\termbang{\id}}{z}})}$,
whereas two $\dBeta$-reductions are needed in order to close the diamond, \ie
to rewrite
$\termapp{(\termapp{\id}{\termbang{\id}})}{\termbang{(\termapp{\id}{\termbang{\id}})}}$
into
$\termapp{\termsubs{z}{\termbang \id}{z}}{\termbang{(\termsubs{z}{\termbang \id}{z}})}$.

It is possible to $\bangweak$-reduce two different redexes of a term in such a
way that the same reduct (modulo $\alpha$-conversion) is obtained. For
instance, if $t = \termsubs{z}{\termbang{u}}{\termsubs{y}{\termbang{u}}{x}}$
then $t \rewrite{\bangweak} \termsubs{y}{\termbang{u}}{x}$ for the
$\ctxt{W}$-context $\Box$, and $t \rewrite{\bangweak}
\termsubs{z}{\termbang{u}}{x}$ for the $\ctxt{W}$-context
$\termsubs{z}{\termbang{u}}{\Box}$. Nevertheless, in such a case the reduction
rules must be the same, as we shall establish in Lemma~\ref{l:pasos-iguales}. As a
preliminary result, we prove that no term can $\bangweak$-reduce to itself in
one $\bangweak$-step.

\begin{lemma}
\label{l:loop}
For all $ t\in \TermExplicit$,  $t\not\rewrite{\bangweak}t$.
\end{lemma}

\begin{proof}
    We  prove  by induction on $t$ that if $t \rewrite{\bangweak} t'$ then $t\neq t'$. 
\begin{itemize}
  \item If $t = x$ then the statement holds vacuously.
  \item If $t = \termabs{x}{r}$ or $t = \termbang{r}$, then the reduction
  $t \rewrite{\bangweak} t'$  takes place inside $r$ and we conclude by the \ih 
  \item If $t = \termder{r}$ and the reduction $ t \rewrite{\bangweak} t'$ takes place inside $r$ then we conclude by
the \ih  If the  reduction $t
  \rewrite{\bangweak} t'$ takes place at the root then it must be the  case that  $r = \ctxtapp{\ctxt{L}}{\termbang{s}}$ and $t'= \ctxtapp{\ctxt{L}}{s}$, so that $t \neq t'$.
  \item If $t = \termapp{r}{s}$ then we reason as in the previous case.
  \item If $t = \termsubs{x}{s}{u}$, then let us write $t =
  \termsubs{x_n}{s_n}{\termsubs{x_1}{s_1}{r}\ldots}$ in such a way that:
  \begin{itemize}
    \item $n\geq 1$
    \item $r$ is not an explicit substitution.
    \item for all $1 \leq i \leq n$, $x_i \notin \fv{s_i} \cup \bv{s_i}$.
  \end{itemize}
  As in the previous cases, the \ih settles the cases in which the reduction $t
  \rewrite{\bangweak} t'$ takes place inside $r$ or inside $s_i$, for $1 \leq
  i \leq n$. The only remaining cases to consider are the $\sBang$-steps of the
  form \[
\begin{array}{rcl}
t & \rewrite{\sBang}  & \termsubs{x_n}{s_n}{\termsubs{x_{j+1}}{s_{j+1}}{\ctxtapp{\ctxt{L}}{\substitute{x_j}{s'}{\termsubs{x_{j-1}}{s_{j-1}}{\termsubs{x_1}{s_1}{r}\ldots}}}}\ldots} \\
  & =                 & \termsubs{y_m}{u_m}{\termsubs{y_1}{u_1}{r'}\ldots} =                 t'
\end{array} \]
  where $s_j = \ctxtapp{\ctxt{L}}{\termbang{s'}}$, for some $\ctxt{L}$, $s'$
  and $1 \leq j \leq n$, and where $r'$ is not an explicit substitution. Let
  $l$ be the length of the list of explicit substitutions $\ctxt{L}$. If $l > 1$
  then $m > n$, so that $t' \neq t$.
  \begin{itemize}
    \item If $l = 0$ and $\substitute{x_j}{s'}{r}$ is not an explicit
    substitution, then $m = n-1$, so that $t' \neq t$. If $l = 0$ and
    $\substitute{x_j}{s'}{r}$ is an explicit substitution, then it must be the
    case that $r = x_j$. In this case, $t$ starts with a free occurrence of
    $x_j$ whereas $t'$ starts with $s'$, and by hypothesis $s'$ does not
    contain free occurrences of $x_j$, being $s'$ a subterm of $s_j$. Hence
    $t'\neq t$.
    \item If $l = 1$, then $s_j = \termsubs{z}{u}{(\termbang{s'})}$ for some
    $z, u$. If $m \neq n$ then $t' \neq t$ and we are done, otherwise $u_j =
    u$, and, in order to have $t=t'$ it should be the case that $u=s_j$. But
    $u$ is a proper subterm of $s_j$, so that $t\neq t'$.
  \end{itemize}
\end{itemize}
\end{proof}


As a matter of fact, a weaker version of Lemma~\ref{l:loop} stating that
one-step cycles are impossible using $\dBeta$ or $\dBang$ only, is sufficient to
prove Lemma~\ref{l:pasos-iguales}. However, the stronger version presented here
provides more insight on the $\BangRev$-calculus.
  
\begin{lemma}
\label{l:pasos-iguales}
If $t \rewrite{p_1} t'$ and $t \rewrite{p_2} t'$, then $p_1 = p_2$, where
$p_1, p_2 \in \{\dBeta, \sBang, \dBang\}$.
\end{lemma}

\begin{proof}
We prove by
induction on $t$ that if $t \rewrite{p_1} t_1$, $t \rewrite{p_2} t_2$
and  $p_1 \neq p_2$ then $t_1\neq t_2$. 
\begin{itemize}
  \item If $t = x$ then the statement holds vacuously.
  \item If $t = \termabs{x}{r}$ or $t = \termbang{r}$ then both reductions must
  take place inside $r$, and the \ih allows us to conclude.
  \item If $t = \termder{r}$ and both reductions take place in $r$, then the
  \ih allows us to conclude. Otherwise it must be the case that $r =
  \ctxtapp{\ctxt{L}}{\termbang{r'}}$, so that $t \rewrite{p_1} t_1$ is, say, $t
  \rewrite{\dBang} \ctxtapp{\ctxt{L}}{r'}$. If $\ctxtapp{\ctxt{L}}{\termbang{r'}}$
  is not an explicit substitution, then no reduction $t \rewrite{p_2} t_2$ with
  $p_1 \neq p_2$ is possible since $\bangweak$-reductions do not take place
  inside a bang. Otherwise $t_1$ is an explicit substitution whereas 
  $t_2 = \termder{\ctxtapp{\ctxt{L'}}{\termbang{r'}}}$, for some
  $\ctxtapp{\ctxt{L}}{\termbang{r'}}\rewrite{p2}\ctxtapp{\ctxt{L'}}{\termbang{r'}}$,
  is a dereliction so that $t_1\neq t_2$. 
  \item If $t = \termapp{r}{s}$ and the reductions take place either both in
  $r$ or both in $s$, then the \ih allows us to conclude. If $\termapp{r}{s}
  \rewrite{p_1} \termapp{r'}{s}$ and $\termapp{r}{s} \rewrite{p_2}
  \termapp{r}{s'}$, then Lemma~\ref{l:loop} allows us to conclude. Otherwise it
  must be the case that $r = \ctxtapp{\ctxt{L}}{\termabs{x}{r'}}$, so that $t
  \rewrite{p_1} t_1$ is, say, $t \rewrite{\dBeta}
  \ctxtapp{\ctxt{L}}{\termsubs{x}{s}{r'}}$, and $t_1$ is an explicit
  substitution. Now, the $p_2$-step must take place either in $r$ or in $s$,
  hence  $t_2$ is an application, so that $t_1\neq t_2$. 
  \item If $t = \termsubs{x}{s}{r}$ and the reductions take place either both
  in $r$ or both in $s$, then the \ih allows us to conclude. If
  $\termsubs{x}{s}{r} \rewrite{p_1} \termsubs{x}{s'}{r}$ and
  $\termsubs{x}{s}{r} \rewrite{p_2} \termsubs{x}{s}{r'}$, then
  Lemma~\ref{l:loop} allows us to conclude.
  Otherwise it must be the case that $s =
  \ctxtapp{\ctxt{L}}{\termbang{s'}}$, so that $t \rewrite{p_1} t_1$ is, say, $t
  \rewrite{\sBang} \ctxtapp{\ctxt{L}}{\substitute{x}{s'}{r}}$. Let us write
  $t = \termsubs{x}{s}{\termsubs{x_n}{s_n}{\termsubs{x_1}{s_1}{r'}\ldots}}$ in
  such a way that:
  \begin{itemize}
    \item $n \geq 0$
    \item $r'$ is not an explicit substitution.
    \item for all $1 \leq i \leq n$, $x_i \notin \fv{s_i} \cup \bv{s_i}$.
  \end{itemize}
  Thus, $t_1 =
  \ctxtapp{\ctxt{L}}{\substitute{x}{s'}{\termsubs{x_n}{s_n}{\termsubs{x_1}{s_1}{r'}\ldots}}}$.
  In the rest of the proof, we note $t \rewrite{p_2} t_2$ the second
  $\bangweak$-reduction of $t$. Since $p_1 = \sBang$
  and $p_1 \neq p_2$, $t \rewrite{p_2} t_2$ is a $\dBeta$-step or a
  $\dBang$-step taking place either in $r'$, or in one of the $s_i$, $1 \leq i
  \leq n$, or in $s$. We observe first that if  $r' = x$
  then  $t_2$ starts with a free occurrence of $x$ and $t_1$ starts
  with $s'$, where $x$ cannot occur free, so that $t_1\neq t_2$.
  It remains to consider the case $r'\neq x$. As a matter of terminology, let
  us say that a term $u$ has exactly $n$ explicit substitutions if $u =
  \termsubs{z_n}{u_n}{\termsubs{z_1}{u_1}{u'}\ldots}$ and $u'$ is not an
  explicit substitution. In the case $r' \neq x$ the term $t_1$ has exactly
  $n + m$ explicit substitutions, where $m$ is the length of $\ctxt{L}$, and
  the term $t_2$ has at least $n + 1$ explicit substitutions. Hence, if
  $\ctxt{L}$ is empty, then $t_1 \neq t_2$. It remains to consider the case
  $\ctxt{L} = {\ctxt{L'}}{\exsubs{w}{u}}$ for some list of explicit
  substitutions $\ctxt{L'}$, variable $w$ and term $u$. In this case the
  outermost explicit substitution of $t_1$ is ${\exsubs{w}{u}}$. We conclude
  the proof by showing that in this case the outermost explicit substitution
  ${\exsubs{w'}{u'}}$ of $t_2$ is such that $\simplesize{u}<\simplesize{u'}$,
  where $\simplesize{t}$ is the size  of the syntactic tree of $t$ (all nodes
  count one\footnote{Here is the definition: $\simplesize{x} = 1$,
  $\simplesize{\termbang{t}} = \simplesize{\termder{t}} =
  \simplesize{\termabs{x}{t}} = \simplesize{t} + 1$,
  $\simplesize{\termapp{t}[u]} = \simplesize{\termsubs{x}{u}{t}} =
  \simplesize{t} + \simplesize{u} + 1$.}). If the reduction $t \rewrite{p_2}
  t_2$ takes place outside the leftmost occurrence of $u$ then $u'$ is of the
  form $\ctxtapp{\ctxt{L_1}}{\termbang s'}$ with $\ctxt{L_1} =
  {\ctxt{L_2}}{\exsubs{w}{u}}$, hence $\simplesize{u} < \simplesize{u'}$. If
  the reduction $t \rewrite{p_2} t_2$ takes place in the leftmost occurrence of
  $u$, then $u'$ is of the form $\ctxtapp{\ctxt{L_1}}{\termbang s'}$ and
  $\ctxt{L_1} = {\ctxt{L'}}{\exsubs{w}{u''}}$, with $u \rewrite{\dBeta} u''$ or
  $u \rewrite{\dBang} u''$. By remarking that $\simplesize{u''} + 2 \geq
  \simplesize{u}$ (simple inspection of the rules), and that $\simplesize{u'}
  \geq \simplesize{u''} + 3$ we get $\simplesize{u} < \simplesize{u'}$, in this
  case too, and we are done.
\end{itemize}
\end{proof}

In the proof of the following lemma, pairs of natural numbers
  are used to decorate reduction sequences.  More precisely, by
  writing $t \rewriten{\bangweak}^{(b,e)} u$ we mean that $t$
  $\bangweak$-reduces to $u$ using $b$ multiplicative steps and $e$
  exponential steps.
We use the
following order on pairs: $(a,b) \parless (a',b')$ iff $a < a'$ and $b \leq b'$,
or $a \leq a'$ and $b < b'$. Moreover,  we use the operation $+$ on pairs  to denote the pairwise addition.

\begin{lemma}
\label{l:reduction-length}
If $t \rewriten{\bangweak}^{c_1} u_1$ and $t \rewriten{\bangweak}^{c_2} u_2$,
then there exists a term $t'$ and pairs $\overline{c_1}, \overline{c_2}$
such that $u_1 \rewriten{\bangweak}^{\overline{c_2}} t'$, and $u_2
\rewriten{\bangweak}^{\overline{c_1}} t'$, where $\overline{c_i} \parleq c_i$
with $i = 1,2$ and $c_1 + \overline{c_2} = c_2 + \overline{c_1}$.
\end{lemma}

\begin{proof} 
In this proof we use the notation $q_i$ for pairs of the form $(1,0)$ or
$(0,1)$. 

The proof of the lemma is by induction on $c_1+ c_2$. The cases $c_1 = (0,0)$
or $c_2 = (0,0)$ (including the base case) are all trivial. So let us suppose
$c_1 \pargreat (0,0)$ and $c_2 \pargreat  (0,0)$. Then $t \rewrite{\bangweak}^{q_1} t_1
\rewriten{\bangweak}^{c'_1} u_1$ and $t \rewrite{\bangweak}^{q_2} t_2
\rewriten{\bangweak}^{c'_2} u_2$, where $c'_i + q_i = c_i$ with $i = 1,2$.

If $t_1 = t_2$, then $q_1 = q_2$ by Lemma~\ref{l:pasos-iguales}. We have $t_1
\rewriten{\bangweak}^{c'_1} u_1$ and $t_1 \rewriten{\bangweak}^{c'_2} u_2$.
Since $c'_1+c'_2 \parless  c_1 + c_2$, then by the \ih there exists $t'$ such that
$u_1 \rewriten{\bangweak}^{\overline{c'_2}} t'$, $u_2
\rewriten{\bangweak}^{\overline{c'_1}} t'$ where $c'_1 + \overline{c'_2} = c'_2
+ \overline{c'_1}$, and $\overline{c'_i} \parleq  c'_i$ with $i = 1,2$. Let us set
$\overline{c_i} = \overline{c'_i}$ with $i = 1,2$.
We conclude since $c_1 + \overline{c_2} = c'_1 + q_1 + \overline{c'_2} = q_1 +
c'_2 + \overline{c'_1} = c_2 + \overline{c_1}$ and $\overline{c_i} =
\overline{c'_i} \parleq  c'_i \parless c_i$ with $i = 1,2$.

If $t_1 \neq t_2$, then by Lemma~\ref{l:diamond} there exists $t_3$ such that
$t_1 \rewrite{\bangweak}^{q_2} t_3$ and $t_2 \rewrite{\bangweak}^{q_1} t_3$. We
now have $t_2 \rewrite{\bangweak}^{q_1} t_3$ and $t_2
\rewriten{\bangweak}^{c'_2} u_2$. Since $c'_2+q_1 \parless  c_1+c_2$, then the \ih
gives $t_4$ such that $u_2 \rewriten{\bangweak}^{\overline{q_1}} t_4$,
$t_3 \rewriten{\bangweak}^{\overline{c'_2}} t_4$, $c'_2 + \overline{q_1} = q_1
+ \overline{c'_2}$, $\overline{q_1} \parleq q_1$ and $\overline{c'_2} \parleq c'_2$.
We now have $t_1 \rewriten{\bangweak}^{c'_1} u_1$ and $t_1
\rewriten{\bangweak}^{q_2 + \overline{c'_2}} t_4$.  In order to apply the \ih
we need $c'_1 + q_2 + \overline{c'_2} \parless  c_1 + c_2$. Indeed, $c'_1 + q_2 +
\overline{c'_2} \parless  c_1 + c_2 = q_1 + c'_1 + q_2 + c'_2$ iff $\overline{c'_2} \parless 
q_1 + c'_2$ iff $\overline{c'_2} \parleq c'_2$. Then the \ih gives $t'$ such that
$u_1 \rewriten{\bangweak}^{\overline{\overline{c'_2}+q_2}} t'$ and $t_4
\rewriten{\bangweak}^{\overline{c'_1}} t'$, where $c'_1 +
\overline{\overline{c'_2} + q_2} = q_2 + \overline{c'_2} + \overline{c'_1}$,
$\overline{c'_1} \parleq c'_1$ and $\overline{\overline{c'_2} + q_2} \parleq
\overline{c'_2} + q_2$. We can then conclude since $u_2
\rewriten{\bangweak}^{\overline{q_1} + \overline{c'_1}} t'$. We have
\begin{itemize}
  \item $\overline{\overline{c'_2} + q_2}\parleq  \overline{c'_2} + q_2 \parleq c'_2
  + q_2$.
  \item $\overline{q_1} + \overline{c'_1} \parleq  q_1 + c'_1$.
  \item $q_1 + c'_1 + \overline{\overline{c'_2} + q_2} = q_1 + q_2 +
  \overline{c'_2} + \overline{c'_1} = q_2 + \overline{c'_1} + c'_2 +
  \overline{q_1}$. 
\end{itemize}
\end{proof}

This gives the following two major results.

\begin{theorem}\mbox{}
\label{t:confluence}
\begin{itemize}
  \item The reduction relation $\rewrite{\bangweak}$ is confluent.
  \item Any two different reduction paths to $\bangweak$-normal form have the same length.
\end{itemize}
\end{theorem}

\begin{proof}
The two statements follow from Lemma~\ref{l:reduction-length}.
\end{proof}

\begin{example}
\label{ex:reduction}
Let us illustrate the previous result with $t =
\termsubs{x}{\termapp{\id'}{(\termapp{\id'}{\termbang{\id}})}}{(\termapp{x}{\termbang{x}})}$,
where $\id' = \termabs{x}{\termbang{x}}$. For that, let us first consider the
following reduction \[\termapp{\id'}{\termbang{\id}} \rewrite{\dBeta}
\termsubs{x}{\termbang{\id}}{(\termbang{x})} \rewrite{\sBang} \termbang{\id}\]
Observe that the second step is a weak step, since the reduction takes place at
the root of the term. On the contrary, notice that the step
$\termbang{(\termsubs{x}{\termbang{\id}}{x})} \rewrite{\sBang} \termbang{\id}$
is not a weak step since the reduction takes place under a bang.

Similarly, \[\termapp{\id}{\termbang{\id}} \rewrite{\dBeta}
\termsubs{x}{\termbang{\id}}{x} \rewrite{\sBang} \id\] Thus,
$\termapp{\id'}{\termbang{\id}} \rewriten{\bangweak}^{(1,1)} \termbang \id$ and
$\termapp{\id}{\termbang{\id}} \rewriten{\bangweak}^{(1,1)} \id$.

Coming back now to the original term $t$, the following essentially different
weak reductions from $t$ to its normal form $\id$ have both length $7$: \[
\begin{array}{lll}
t & \rewriten{\bangweak}^{(1,1)} &
\termsubs{x}{\termapp{\id'}{\termbang{\id}}}{(\termapp{x}{\termbang{x}})} \rewriten{\bangweak}^{(1,1)}
\termsubs{x}{\termbang{\id}}{(\termapp{x}{\termbang{x}})} \rewrite{\sBang}
\termapp{\id}{\termbang{\id}} \rewriten{\bangweak}^{(1,1)}
\id \\
\\
t & \rewrite{\dBeta} &
\termsubs{x}{\termsubs{z}{\termapp{\id'}{\termbang{\id}}}{(\termbang{z})}}{(\termapp{x}{\termbang{x}})} \rewrite{\sBang}
\termsubs{z}{\termapp{\id'}{\termbang{\id}}}{(\termapp{z}{\termbang{z}})} \\
  & \rewriten{\bangweak}^{(1,1)} &
\termsubs{z}{\termbang{\id}}{(\termapp{z}{\termbang{z}})} \rewrite{\sBang}
\termapp{\id}{\termbang \id} \rewriten{\bangweak}^{(1,1)}
\id
\end{array} \]
\end{example}

As explained above, the strong property expressed in the second item of
Theorem~\ref{t:confluence} and illustrated in Example \ref{ex:reduction} relies
essentially on the fact that reductions are disallowed under bangs. Observe
the important role of the $\mathtt L$-context
$\termsubs{z}{\termapp{\id'}{\termbang I}}{\Box}$ in the second step of the
last reduction sequence.

\paragraph{{\bf Normal forms and neutral terms}}
A term is said to be \emphdef{$\bangweak$-normal} if there is no $t'$ such that
$t \rewrite{\bangweak} t'$, in which case we write $t \not\rewrite{\bangweak}$.
This notion can be characterised by means of the following inductive grammars:
\begin{center}
\begin{tabular}{rrcll}
\textbf{(Neutral)}      & $\ipntrl$  & $\Coloneq$ & $x \in \TermVariable \mid \termapp{\ipnabs}{\ipnrml} \mid \termder{(\ipnbang)} \mid \termsubs{x}{\ipnbang}{\ipntrl}$ \\
\textbf{(Neutral-Abs)}  & $\ipnabs$  & $\Coloneq$ & $\termbang{t} \mid \ipntrl \mid \termsubs{x}{\ipnbang}{\ipnabs}$ \\ 
\textbf{(Neutral-Bang)} & $\ipnbang$ & $\Coloneq$ & $\ipntrl \mid \termabs{x}{\ipnrml} \mid \termsubs{x}{\ipnbang}{\ipnbang}$ \\
\textbf{(Normal)}       & $\ipnrml$  & $\Coloneq$ & $\ipnabs \mid \ipnbang$
\end{tabular}
\end{center}
As we shall see (\cf Proposition~\ref{prop:normal}), all these terms
are $\bangweak$-normal. Moreover, \emphdef{neutral} terms do not produce any
kind of redexes when inserted into a context, while \emphdef{neutral-abs}
terms (resp. \emphdef{neutral-bang})  may only produce $\sBang$ or $\dBang$
redexes (resp. $\dBeta$ redexes) when inserted into a context.

\begin{remark}
\label{r:normal}
Some immediate  properties of the sets of terms defined above are:
\begin{itemize}
  \item $\ipntrl\ =\ \ipnabs \cap \ipnbang$ 
  \item  $\ipnrml\ =\  \ipnabs \cup \ipnbang$
  \item for all terms $t$, $\pnabs{t}$ implies $t\in\ipntrl$ or $\pbang{t}$.
  \item for all terms $t$, $\pnbang{t}$ implies $t\in\ipntrl$ or $\pabs{t}$.
  \item for all terms $t$, $\pnabs{t}$ and $t \notin \ipnbang$ implies $\pbang{t}$.
  \item for all terms $t$, $\pnbang{t}$ and $t \notin \ipnabs$ implies $\pabs{t}$.
\end{itemize}
\end{remark}

\begin{proposition}[Normal Forms]
\label{prop:normal}
For all $t \in \TermExplicit$, $t\not\rewrite{\bangweak}$ iff
$\pnrml{t}$.
\end{proposition}

\begin{proof}
We prove simultaneously the following statements:
\begin{enumerate}
\item[(a)] $\pntrl{t}  \quad\Leftrightarrow\quad  t \not\rewrite{\bangweak}$
  and
  $\neg\pabs{t}$ and $\neg\pbang{t}$.
  \item[(b)] $\pnabs{t}  \quad\Leftrightarrow\quad t \not\rewrite{\bangweak}$ and
  $\neg\pabs{t}$.
  \item[(c)] $\pnbang{t}  \quad\Leftrightarrow\quad   t \not\rewrite{\bangweak}$ and
  $\neg\pbang{t}$.
  \item[(d)] $\pnrml{t}  \quad\Leftrightarrow\quad   t \not\rewrite{\bangweak}$.
\end{enumerate}
The implication $\Rightarrow$ is proved by induction on the predicate
$\pnrml{t}$.  In all cases, it is sufficient to reason by a simple
case analysis of the grammars defining the sets $\ipntrl$, $\ipnabs$,
$\ipnbang$ and $\ipnrml$, using the suitable induction hypothesis.

The implication $\Leftarrow$ is proved by induction on $t$.
As for $\Leftarrow$ (a): $\bangweak$-normal terms neither of the form
$\ctxtapp{\ctxt{L}}{\termbang{t'}}$ nor $\ctxtapp{\ctxt{L}}{\termabs{x}{t'}}$
have necessarily one of the following shapes:
\begin{itemize}
  \item $x$
  \item $\termder{t}$, where $t$ is normal and $\neg\pbang{t}$.
  \item $\termapp{t}{u}$, where $t$, $u$ are normal and $\neg\pabs{t}$.
  \item $\termsubs{x}{u}{t}$, where $t$, $u$ are normal, $\neg\pabs{t}$,
  $\neg\pbang{t}$ and $\neg\pbang{u}$.
\end{itemize} 
Each case is settled by using the corresponding case in the definition of
$\ipntrl$ and the suitable induction hypothesis. The same holds for
$\Leftarrow$ (b), $\Leftarrow$ (c) and  $\Leftarrow$ (d).
\end{proof}

\paragraph{{\bf Clashes}}\label{??}
Some ill-formed terms are not redexes but they don't represent a desired result
for a computation either. They are called \emphdef{clashes} (meta-variable
$c$), and defined as follows: \[
\termapp{\ctxtapp{\ctxt{L}}{\termbang{t}}}{u} \qquad 
\termsubs{y}{\ctxtapp{\ctxt{L}}{\termabs{x}{u}}}{t} \qquad 
\termder{(\ctxtapp{\ctxt{L}}{\termabs{x}{u}})} \qquad
\termapp{t}{({\ctxtapp{\ctxt{L}}{\termabs{x}{u}}})}
\] Observe that in the three first kind of clashes, replacing $\termabs{x}{\!}$
by $\termbang{\!}$, and inversely, creates a (root) redex, namely
$\termapp{(\ctxtapp{\ctxt{L}}{\termabs{x}{t}})}{u}$,
$\termsubs{x}{\ctxtapp{\ctxt{L}}{\termbang{t}}}{t}$ and 
$\termder{(\ctxtapp{\ctxt{L}}{\termbang{t}})}$, respectively. In the fourth
kind of clash, however, this is not the case since
$\termapp{t}{({\ctxtapp{\ctxt{L}}{\termbang{u}}})}$ is not a redex in general.

A term is \emphdef{clash-free} if it does not reduce to a term containing a
clash, it is \emphdef{weakly clash-free}, written $\cfz$, if it does not reduce
to a term containing a clash outside the scope of any constructor $!$. In other
words, $t$ is not $\cfz$ if and only if there exist a weak context $\ctxt{W}$
and a clash $c$ such that $t \rewriten{\bangweak} \ctxtapp{\ctxt{W}}{c}$.

Weakly clash-free normal terms can be characterised as follows:
\begin{center}
\begin{tabular}{rrcll}
\textbf{(Neutral $\cfz$)}      & $\icfntrl$  & $\Coloneq$ & $x \in \TermVariable \mid \termapp{\icfntrl}{\icfnabs} \mid \termder{(\icfntrl)} \mid \termsubs{x}{\icfntrl}{\icfntrl}$ \\
\textbf{(Neutral-Abs $\cfz$)}  & $\icfnabs$  & $\Coloneq$ & $\termbang{t} \mid \icfntrl \mid \termsubs{x}{\icfntrl}{\icfnabs}$ \\ 
\textbf{(Neutral-Bang $\cfz$)} & $\icfnbang$ & $\Coloneq$ & $\icfntrl \mid \termabs{x}{\icfnrml} \mid \termsubs{x}{\icfntrl}{\icfnbang}$ \\
\textbf{(Normal $\cfz$)}       & $\icfnrml$  & $\Coloneq$ & $\icfnabs \mid \icfnbang$
\end{tabular}
\end{center}
Intuitively, $\icfnrml$ denotes $\ipnrml \cap \cfz$ (respectively for
$\icfntrl$, $\icfnabs$ and $\icfnbang$).

\begin{proposition}[{\bf Clash-free normal forms}]
\label{p:clashfree}
Let $t \in \TermExplicit$. Then $t$ is a weakly clash-free normal form iff
$\cfnrml{t}$.
\end{proposition}

\begin{proof}
Similar to Proposition~\ref{prop:normal}.
\end{proof}


\section{The Type System $\SysBang$}
\label{s:system}

This section introduces a first type system $\SysBang$ for our revisited
version of the Bang calculus, which extends the one in~\cite{GuerrieriM18} to
explicit substitutions.  We show in this paper that $\SysBang$ does not only
qualitatively characterise normalisation, but is also \emph{quantitative}, in
the sense that the length of the (weak) reduction of a typed term to its normal
form plus the size of this normal form is bounded by the size of its type
derivation. We also explore in Sec.~\ref{s:cbname-cbvalue} the properties of
this type system with respect to the CBN and CBV translations.

Given a countable infinite set $\TypeVariable$ of base types $\alpha,
\beta, \gamma, \ldots$, we define the following sets of types:
\begin{center}
\begin{tabular}{rrcll}
\textbf{(Types)}          & $\sigma, \tau$ & $\Coloneq$ & $\alpha \in \TypeVariable \mid \M \mid \functtype{\M}{\sigma}$ \\
\textbf{(Multiset Types)} & $\M$           & $\Coloneq$ & $\intertype{\sigma_i}{i \in I}$  where $I$ is a finite set
\end{tabular}
\end{center}

Multiset types will be indistinctly written 
  as $\M$ or  $\intertype{\sigma_i}{i \in I}$, in both cases
  they denote a finite multiset of types. 
The empty multiset type is denoted by $\emul$.  
Also, $|\M|$ denotes the size of the multiset, thus
if $\M = \intertype{\sigma_i}{i \in I}$ then $|\M| = \#(I)$.

Two multiset types
  are identified if they are related by the relation $\equiv$,
  defined by induction on types as follows:
  \begin{itemize}
    \item $\alpha \equiv \alpha$,
     \item $\functtype{\M_1}{\sigma_1} \equiv \functtype{\M_2}{\sigma_2} $ if
  $\M_1 \equiv \M_2$ and $\sigma_1 \equiv \sigma_2$,
     \item $\intertype{\sigma_i}{\iI} \equiv \intertype{\tau_j}{\jJ}$ if
  $|I| = |J|$ and there is a bijection function $f$ between $I$ and $J$
       such that $\sigma_i\equiv \tau_{f(i)}$ for all $\iI$.
\end{itemize} 

\emphdef{Typing contexts} (or just \emphdef{contexts}), written $\Gamma,
\Delta$, are functions from variables to multiset types, assigning the empty
multiset to all but a finite set of variables. The support of $\Gamma$ is given
by $\dom{\Gamma} \eqdef \set{x \mid \Gamma(x) \neq \intertype{}{}}$. The
\emphdef{empty context} is the context with an empty support. The
\emphdef{union of contexts}, written $\ctxtsum{\Gamma}{\Delta}{}$, is defined
by $(\ctxtsum{\Gamma}{\Delta}{})(x) \eqdef \Gamma(x) \sqcup \Delta(x)$, where
$\sqcup$ denotes multiset union. An example is
$\ctxtsum{(\assign{x}{\intertype{\sigma}{}},
\assign{y}{\intertype{\tau}{}})}{(\assign{x}{\intertype{\sigma}{}},
\assign{z}{\intertype{\tau}{}})}{} = (\assign{x}{\intertype{\sigma, \sigma}{}},
\assign{y}{\intertype{\tau}{}}, \assign{z}{\intertype{\tau}{}})$. This notion
is  extended to several contexts as expected, so that $\ctxtsum{}{}{i \in I}
\Gamma_i$ denotes a finite union of contexts (when $I = \emptyset$ the notation
is to be understood as the empty context). We write $\ctxtres{\Gamma}{x}{}$ for
the context $(\ctxtres{\Gamma}{x}{})(x) = \emul$ and
$(\ctxtres{\Gamma}{x}{})(y) = \Gamma(y)$ if $y \neq x$. Contexts can be
compared as follows $\Gamma \ctxtleq \Delta$ iff $\Gamma(x) \sqsubseteq
\Delta(x)$ for every variable $x$, where $\sqsubseteq$ is multiset inclusion.

\emphdef{Type judgements} have the form $\sequ{\Gamma}{\assign{t}{\sigma}}$,
where $\Gamma$ is a typing context, $t$ is a term and $\sigma$ is a type.
The type system $\SysBang$ for the $\BangRev$-calculus is given in
Figure~\ref{fig:typingSchemesBang}.

\begin{figure}
\centering $
\begin{array}{c}
\Rule{\vphantom{\Gamma}}
     {\sequ{\assign{x}{\intertype{\sigma}{}}}{\assign{x}{\sigma}}}
     {\ruleBAxiom}
\qquad
\Rule{\sequ{\Gamma}{\assign{t}{\sigma}}
      \quad
      \sequ{\Delta}{\assign{u}{\Gamma(x)}}
     }
     {\sequ{\ctxtsum{(\ctxtres{\Gamma}{x}{})}{\Delta}{}}{\assign{\termsubs{x}{u}{t}}{\sigma}}}
     {\ruleBESubs}
     \qquad
\Rule{\sequ{\Gamma}{\assign{t}{\tau}}}
     {\sequ{\ctxtres{\Gamma}{x}{}}{\assign{\termabs{x}{t}}{\functtype{\Gamma(x)}{\tau}}}}
     {\ruleBArrowI} \\ \\
\Rule{\sequ{\Gamma}{\assign{t}{\functtype{\M}{\tau}}}
      \quad
      \sequ{\Delta}{\assign{u}{\M}}
     }
     {\sequ{\ctxtsum{\Gamma}{\Delta}{}}{\assign{\termapp{t}{u}}{\tau}}}
     {\ruleBArrowE}
\qquad 
\Rule{\many{\sequ{\Gamma_i}{\assign{t}{\sigma_i}}}{i \in I}}
     {\sequ{\ctxtsum{}{\Gamma_i}{i \in I}}{\assign{\termbang{t}}{\intertype{\sigma_i}{i \in I}}}}
     {\ruleBBang}
\qquad
\Rule{\sequ{\Gamma}{\assign{t}{\intertype{\sigma}{}}}}
     {\sequ{\Gamma}{\assign{\termder{t}}{\sigma}}}
     {\ruleBDer}
\end{array} $
\caption{System $\SysBang$ for the $\BangRev$-calculus.}
\label{fig:typingSchemesBang}
\end{figure}

The axiom $\ruleBAxiom$ is relevant (there is no weakening) and the rules
$\ruleBArrowE$ and $\ruleBESubs$ are  multiplicative. Note that the argument of
a bang is typed $\#(I)$ times by the premises of rule $\ruleBBang$. A
particular case is when $I = \emptyset$: the subterm $t$ occurring in the typed
term $\termbang{t}$ turns out not to be typed (we often say that $t$ is \emph{untyped} in this case). 

A \emphdef{(type) derivation} is a tree obtained by applying the (inductive)
typing rules of system $\SysBang$. The notation
$\derivable{}{\sequ{\Gamma}{\assign{t}{\sigma}}}{\SysBang}$ means there is a
derivation of the judgement $\sequ{\Gamma}{\assign{t}{\sigma}}$ in system
$\SysBang$. The term $t$ is \emphdef{typable} in system $\SysBang$, or
$\SysBang$-typable, iff there are $\Gamma$ and $\sigma$ such that
$\derivable{}{\sequ{\Gamma}{\assign{t}{\sigma}}}{\SysBang}$. We use the capital
Greek letters $\Phi, \Psi, \ldots$ to name type derivations, by writing for
example $\derivable{\Phi}{\sequ{\Gamma}{\assign{t}{\sigma}}}{\SysBang}$.
The \emphdef{size of the derivation} $\Phi$, denoted by $\size{\Phi}$, is
defined as  the number of rules  in the type derivation $\Phi$ except rule $\ruleBBang$, which
does not count. Note in particular that,  given a derivation $\Phi_t$ for a term
$t$,  we always have $\size{\Phi_t} \geq \wsize{t}$, as $\wsize{\_}$
does not count in turn subterms prefixed by a bang constructor.

\begin{example}
\label{example:t0-typed}
The following tree $\Phi_0$ is a type derivation for term $t_0$ of
Example~\ref{example:t0}.
\begin{center}
{\small$
\Rule{
  \Rule{
    \Rule{
      \Rule{
        \Rule{
          \Rule{
            \Rule{}{
              \sequB{\assign{x}{\multiset{\functtype{\multiset{\tau}}{\tau}}}}{\assign{x}{\functtype{\multiset{\tau}}{\tau}}}
            }{\ruleBAxiom}
          }{
            \sequB{\assign{x}{\multiset{\functtype{\multiset{\tau}}{\tau}}}}{\assign{\termabs{y}{x}}{\functtype{\emul}{\functtype{\multiset{\tau}}{\tau}}}}
          }{\ruleBArrowI}
        }{
          \sequB{}{\assign{\termabs{x}{\termabs{y}{x}}}{\functtype{\multiset{\functtype{\multiset{\tau}}{\tau}}}{\functtype{\emul}{\functtype{\multiset{\tau}}{\tau}}}}}
        }{\ruleBArrowI}
      }{
        \sequB{}{\assign{\termbang{\Kterm}}{\multiset{\functtype{\multiset{\functtype{\multiset{\tau}}{\tau}}}{\functtype{\emul}{\functtype{\multiset{\tau}}{\tau}}}}}}
      }{\ruleBBang}
    }{
      \sequB{}{\assign{\termder{(\termbang{\Kterm})}}{\functtype{\multiset{\functtype{\multiset{\tau}}{\tau}}}{\functtype{\emul}{\functtype{\multiset{\tau}}{\tau}}}}}
    }{\ruleBDer}
    \ 
    \Rule{
      \Rule{
        \Rule{}{
          \sequB{\assign{x}{\multiset{\tau}}}{\assign{x}{\tau}}
        }{\ruleBAxiom}
      }{
        \sequB{}{\assign{\termabs{x}{x}}{\functtype{\multiset{\tau}}{\tau}}}
      }{\ruleBArrowI}
    }{
      \sequB{}{\assign{\termbang{\id}}{\multiset{\functtype{\multiset{\tau}}{\tau}}}}
    }{\ruleBBang}
  }{
    \sequB{}{\assign{\termapp{\termder{(\termbang{\Kterm})}}{(\termbang{\id})}}{\functtype{\emul}{\functtype{\multiset{\tau}}{\tau}}}}
  }{\ruleBArrowE}
  \Rule{}{
    \sequB{}{\assign{\termbang{\Omega}}{\emul}}
  }{\ruleBBang}
}{
  \sequB{}{\assign{\termapp{\termapp{\termder{(\termbang{\Kterm})}}{(\termbang{\id})}}{(\termbang{\Omega})}}{\functtype{\multiset{\tau}}{\tau}}}
}{\ruleBArrowE}
$}
\end{center}
Note that $\size{\Phi_0} = 8 \geq 3 = \wsize{t_0}$, the normal form of $t_0$ is $\id$ with size
$1$, and $t_0$ reduces to $\id$ in $5$ steps. We will see in
Theorem~\ref{t:correctness-bang} that the size of a derivation
$\derivable{\Phi}{\sequ{\Gamma}{\assign{t}{\sigma}}}{\SysBang}$
is always an upper bound of the $\bangweak$-size of the $\bangweak$-normal form
of $t$ plus the length of the reduction of $t$ to its $\bangweak$-normal form.
\end{example}

\bigskip
The typability of a term may provide additional information about the
neutrality/normality of its subterms:

\begin{lemma}\mbox{}
  \label{l:clashfree-bang}
  Let $u\in \TermExplicit$:
\begin{enumerate}
  \item If $\pnabs{t}$ and $\termapp{t}{u}$ is $\SysBang$-typable, then $\pntrl{t}$.
  \item If $\pnbang{t}$ and $\termsubs{x}{t}{u}$ is $\SysBang$-typable, then $\pntrl{t}$.
  \item If $\pnbang{t}$ and $\termder{t}$ is $\SysBang$-typable, then $\pntrl{t}$.
  \item If $\pnbang{t}$ and $\termapp{u}{t}$ is $\SysBang$-typable, then $\pntrl{t}$.
  \item If $\pnrml{t}$ and $\termapp{u}{t}$ is $\SysBang$-typable, then $\pnabs{t}$.
\end{enumerate}
\end{lemma}

\begin{proof}
Straightforward case analysis using the characterisation of $\bangweak$-normal forms in
Proposition~\ref{prop:normal}.
\end{proof}

The quantitative aspect of system $\SysBang$ is materialised in the following
weighted subject reduction (WSR) and expansion (WSE) properties. As usual, a
substitution lemma must be proved.

\begin{lemma}[Substitution]
\label{l:substitution-bang}
If
$\derivable{\Phi_{t}}{\sequB{\Gamma; \assign{x}{\intertype{\sigma_i}{i \in I}}}{\assign{t}{\tau}}}{\SysBang}$
and
$\many{\derivable{\Phi^{i}_{u}}{\sequB{\Delta_i}{\assign{u}{\sigma_i}}}{\SysBang}}{i \in I}$,
then there exists
$\derivable{\Phi_{\substitute{x}{u}{t}}}{\sequB{\ctxtsum{\Gamma}{\Delta_i}{i \in I}}{\assign{\substitute{x}{u}{t}}{\tau}}}{\SysBang}$
such that $\size{\Phi_{\substitute{x}{u}{t}}} = \size{\Phi_{t}}  +_{i \in I}{\size{\Phi^i_{u}}} - |I|$.
\end{lemma}

\begin{proof}
By induction on $\Phi_t$. If $\Phi_t$ is $\ruleBAxiom$ and $t = x$, then
$\substitute{x}{u}{t} = u$ and $\Phi_t$ is of the form
$\sequ{\assign{x}{\intertype{\sigma}{}}}{\assign{x}{\sigma}}$, so that $\Gamma = \emptyset$ and $I = \{ i_0\} $ and $\tau = \sigma$. We let
$\Phi_{\substitute{x}{u}{t}} = \Phi^{i_0}_u$. We conclude since $\size{\Phi_t} =
1$ and $|I|=1$. If $\Phi_t$ is $\ruleBAxiom$ and $t = y \neq x$, then
$\substitute{x}{u}{t} = y$, $\Gamma = \assign{y}{\intertype{\tau}{}}$,
and $I= \emptyset$. We let
$\Phi_{\substitute{x}{u}{t}} = \Phi_y$. We conclude since $|I|=0$.

If $\Phi_t$ ends with $\ruleBArrowE$, then $t=t_1t_2$, $\Gamma=\Gamma_1+\Gamma_2$
  and there exist a type $\M$ and two derivations $\derivable{\Phi_{t_1}}{\sequ{\Gamma_1; \assign{x}{\intertype{\sigma_i}{\iI_1}}}{\assign{t_1}{\M\rightarrow \tau}}}{\SysBang}$ and
  $\derivable{\Phi_{t_2}}{\sequ{\Gamma_2;  \assign{x}{\intertype{\sigma_i}{\iI_2}}}{\assign{t_2}{\M}}}{\SysBang}$ such that
  $I = I_1 \uplus I_2$.
  Using the \ih on $\Phi_{t_i}$ and $(\Phi^j_{u})_{j \in I_i}$, for $i=1,2$, we get two derivations
  $\derivable{\Phi_{\substitute{x}{u}{t_1}}}{\sequ{\ctxtsum{\Gamma_1}{\Delta_i}{\iI_1}}{\assign{\substitute{x}{u}{t_1}}{\M\rightarrow \tau}}}{\SysBang}$ and
  $\derivable{\Phi_{\substitute{x}{u}{t_2}}}{\sequ{\ctxtsum{\Gamma_2}{\Delta_i}{\iI_2}}{\assign{\substitute{x}{u}{t_2}}{\M}}}{\SysBang}$ such that
  $\size{\Phi_{\substitute{x}{u}{t_i}}} =\size{\Phi_{t_i}} +_{j \in I_i} \size{\Phi^j_{u}} - |I_i|$, for  $i=1,2$. By observing that $\substitute x u t=\termapp{\substitute x u {t_1}}{\substitute x u {t_2}}$
  and by using $\ruleBArrowE$, we get a derivation $\derivable{\Phi_{\substitute{x}{u}{t}}}{\sequV{\ctxtsum{\Gamma}{\Delta_i}{\iI}}{\assign{\substitute{x}{u}{t}}{\tau}}}{\SysBang}$ such that
  $\size{\Phi_{\substitute{x}{u}{t}}} = (\size{\Phi_{t_1}} +_{j \in I_1} \size{\Phi^j_{u}} - |I_1|)+(\size{\Phi_{t_2}} +_{j \in I_2} \size{\Phi^j_{u}} - |I_2|)+1=
  (\size{\Phi_{t_1}} + \size{\Phi_{t_2}} +1)+_{\iI} \size{\Phi^i_{u}} -(|I_1|)
  +|I_2|)=  
  \size{\Phi_{t}} +_{\iI} \size{\Phi^i_{u}} - |I|$.

  All the other cases proceed similarly by the \ih 
\end{proof}

\begin{lemma}[Weighted Subject Reduction]
\label{l:wsr-bang}
Let $\derivable{\Phi}{\sequB{\Gamma}{\assign{t}{\tau}}}{\SysBang}$. If
$t \rewrite{\bangweak} t'$, then there is
$\derivable{\Phi'}{\sequB{\Gamma}{\assign{t'}{\tau}}}{\SysBang}$
such that $\size{\Phi} > \size{\Phi'}$.
\end{lemma}

\begin{proof}
By induction on $t \rewrite{\bangweak} t'$. 
\begin{itemize}
\item For the base cases we
    have to consider three rules:
  \begin{itemize}
    \item Rule $\dBeta$. Then $t = \termapp{\ctxtapp{\ctxt{L}}{\termabs{x}{s}}}{u}$
    and $t' = \ctxtapp{\ctxt{L}}{\termsubs{x}{u}{s}}$. We proceed by induction
    on $\ctxt{L}$.
    \begin{itemize}
      \item $\ctxt{L} = \Box$. Then $\Gamma = \ctxtsum{(\ctxtres{\Gamma'}{x}{})}{\Delta}{}$
      s.t. $\Gamma'(x) = \M$ and $$
\Rule{
  \Rule{
    \derivable{\Phi_{s}}{\sequB{\Gamma'}{\assign{s}{\tau}}}{\SysBang}
  }{
    \sequB{\ctxtres{\Gamma'}{x}{}}{\assign{\termabs{x}{s}}{\functtype{\M}{\tau}}}
  }{\ruleBArrowI}
  \quad
  \derivable{\Phi_{u}}{\sequB{\Delta}{\assign{u}{\M}}}{\SysBang}
}{
  \derivable{\Phi}{\sequB{\Gamma}{\assign{\termapp{(\termabs{x}{s})}{u}}{\tau}}}{\SysBang}
}{\ruleBArrowE}
      $$ and we conclude with $$
\Rule{
  \derivable{\Phi_{s}}{\sequB{\Gamma'}{\assign{s}{\tau}}}{\SysBang}
  \quad
  \derivable{\Phi_{u}}{\sequB{\Delta}{\assign{u}{\M}}}{\SysBang}
}{
  \derivable{\Phi'}{\sequB{\Gamma}{\assign{\termsubs{x}{u}{s}}{\tau}}}{\SysBang}
}{\ruleBESubs} $$
     Note that $\size{\Phi} = \size{\Phi'} + 1$.
 
      \item $\ctxt{L} = \termsubs{y}{r}{\ctxt{L'}}$. Then $\Gamma =
      \ctxtsum{(\ctxtsum{(\ctxtres{\Gamma'}{y}{})}{\Delta'}{})}{\Delta}{}$
      s.t. $(\ctxtsum{(\ctxtres{\Gamma'}{y}{})}{\Delta'}{})(x) = \M$ and
      $\Gamma'(y) = \M'$ and
      $$\kern-5em
\prooftree
  \Rule{
    \derivable{\Phi_{\ctxt{L'}}}{\sequB{\Gamma'}{\assign{\ctxtapp{\ctxt{L'}}{\termabs{x}{s}}}{\functtype{\M}{\tau}}}}{\SysBang}
    \quad
    \derivable{\Phi_{r}}{\sequB{\Delta'}{\assign{r}{\M'}}}{\SysBang}
  }{
    \sequB{\ctxtsum{(\ctxtres{\Gamma'}{y}{})}{\Delta'}{}}{\assign{\ctxtapp{\ctxt{L}}{\termabs{x}{s}}}{\functtype{\M}{\tau}}}
  }{\ruleBESubs}
  \quad
  \derivable{\Phi_{u}}{\sequB{\Delta}{\assign{u}{\M}}}{\SysBang}
\justifies
  \derivable{\Phi}{\sequB{\Gamma}{\assign{\termapp{\ctxtapp{\ctxt{L}}{\termabs{x}{s}}}{u}}{\tau}}}{\SysBang}
\using
  \ruleBArrowE
\endprooftree
      $$ Thus, we build $$
\Rule{
  \derivable{\Phi_{\ctxt{L'}}}{\sequB{\Gamma'}{\assign{\ctxtapp{\ctxt{L'}}{\termabs{x}{s}}}{\functtype{\M}{\tau}}}}{\SysBang}
  \quad
  \derivable{\Phi_{u}}{\sequB{\Delta}{\assign{u}{\M}}}{\SysBang}
}{
  \derivable{\Psi}{\sequB{\ctxtsum{\Gamma'}{\Delta}{}}{\assign{\termapp{\ctxtapp{\ctxt{L'}}{\termabs{x}{s}}}{u}}{\tau}}}{\SysBang}
}{\ruleBArrowE}
      $$ and by \ih there exists
$\derivable{\Psi'}{\sequB{\ctxtsum{\Gamma'}{\Delta}{}}{\assign{\ctxtapp{\ctxt{L'}}{\termsubs{x}{u}{s}}}{\tau}}}{\SysBang}$
such that $\size{\Psi} > \size{\Psi'}$. Then, we conclude with $$
\Rule{
  \derivable{\Psi'}{\sequB{\ctxtsum{\Gamma'}{\Delta}{}}{\assign{\ctxtapp{\ctxt{L'}}{\termsubs{x}{u}{s}}}{\tau}}}{\SysBang}
  \quad
  \derivable{\Phi_{r}}{\sequB{\Delta'}{\assign{r}{\M'}}}{\SysBang} }{
  \derivable{\Phi'}{\sequB{\ctxtsum{(\ctxtsum{(\ctxtres{\Gamma'}{y}{})}{\Delta'}{})}{\Delta}{}}{\assign{\ctxtapp{\ctxt{L}}{\termsubs{x}{u}{s}}}{\tau}}}{\SysBang}
}{\ruleBESubs} $$ since we may assume that $y \notin
\dom{\Delta}$. Notice that $\size{\Phi} = \size{\Psi} +
\size{\Phi_{r}} +1 >\size{\Psi'} + \size{\Phi_{r}}+1
= \size{\Phi'}$.
    \end{itemize}
    
    \item Rule $\sBang$. Then $t = \termsubs{x}{\ctxtapp{\ctxt{L}}{\termbang{u}}}{s}$
    and $t' = \ctxtapp{\ctxt{L}}{\substitute{x}{u}{s}}$. We proceed by induction
    on $\ctxt{L}$.
    \begin{itemize}
      \item $\ctxt{L} = \Box$. Then $\Gamma = \ctxtsum{(\ctxtres{\Gamma'}{x}{})}{\Delta_i}{i \in I}$
      s.t. $\Gamma'(x) = \intertype{\sigma}{i \in I}$ and $$
\prooftree
  \derivable{\Phi_{s}}{\sequB{\Gamma'}{\assign{s}{\tau}}}{\SysBang}
  \quad
  \Rule{
    \many{\derivable{\Phi^{i}_{u}}{\sequB{\Delta_i}{\assign{u}{\sigma_i}}}{\SysBang}}{i \in I}
  }{
    \sequB{\ctxtsum{}{\Delta_i}{i \in I}}{\assign{\termbang{u}}{\intertype{\sigma}{i \in I}}}
  }{\ruleBBang}
\justifies
  \derivable{\Phi}{\sequB{\Gamma}{\assign{\termsubs{x}{\termbang{u}}{s}}{\tau}}}{\SysBang}
\using
  \ruleBESubs
\endprooftree $$
      Thus, we conclude directly by Lemma~\ref{l:substitution-bang} with $\Phi_{s}$ and
      $\many{\Phi^{i}_{u}}{i \in I}$. Notice that $\size{\Phi} = 1+ \size{\Phi_s} 
      +_{i \in I}{\size{\Phi^i_u}}$, while $\size{\Phi'} = \size{\Phi_s}
      +_{i \in I}{\size{\Phi^i_u}} - |I|$.
      
      \item $\ctxt{L} = \termsubs{y}{r}{\ctxt{L'}}$. Then $\Gamma =
      \ctxtsum{(\ctxtres{\Gamma'}{x}{})}{\ctxtsum{(\ctxtres{\Delta}{y}{})}{\Delta'}{}}{}$
      with $\Gamma'(x) = \M$, $\Delta(y) = \M'$ and $$
\prooftree
  \derivable{\Phi_{s}}{\sequB{\Gamma'}{\assign{s}{\tau}}}{\SysBang}
  \quad
  \Rule{
    \derivable{\Phi_{\ctxt{L'}}}{\sequB{\Delta}{\assign{\ctxtapp{\ctxt{L'}}{\termbang{u}}}{\M}}}{\SysBang}
    \quad
    \derivable{\Phi_{r}}{\sequB{\Delta'}{\assign{r}{\M'}}}{\SysBang}
  }{
  \sequB{\ctxtsum{(\ctxtres{\Delta}{y}{})}{\Delta'}{}}{\assign{\ctxtapp{\termsubs{y}{r}{\ctxt{L'}}}{\termbang{u}}}{\M}}
  }{\ruleBESubs}
\justifies
  \derivable{\Phi}{\sequB{\Gamma}{\assign{\termsubs{x}{
  \ctxtapp{\ctxt{L}}{\termbang{u}}
  }{s}}{\tau}}}{\SysBang}
\using
  \ruleBESubs
\endprooftree
      $$ Thus, we build $$
\Rule{
  \derivable{\Phi_{s}}{\sequB{\Gamma'}{\assign{s}{\tau}}}{\SysBang}
  \quad
  \derivable{\Phi_{\ctxt{L'}}}{\sequB{\Delta}{\assign{\ctxtapp{\ctxt{L'}}{\termbang{u}}}{\M}}}{\SysBang}
}{
  \derivable{\Psi}{\sequB{\ctxtsum{(\ctxtres{\Gamma'}{x}{})}{\Delta}{}}{\assign{\termsubs{x}{
  \ctxtapp{\ctxt{L'}}{\termbang{u}}
  }{s}}{\tau}}}{\SysBang}
}{\ruleBESubs}
      $$ and by the \ih there exists
      $\derivable{\Psi'}{\sequB{\ctxtsum{(\ctxtres{\Gamma'}{x}{})}{\Delta}{}}{\assign{\ctxtapp{\ctxt{L'}}{\substitute{x}{u}{s}}}{\tau}}}{\SysBang}$
      such that $\size{\Psi} > \size{\Psi'}$. Then, we conclude with $$
\Rule{
  \derivable{\Psi'}{\sequB{\ctxtsum{(\ctxtres{\Gamma'}{x}{})}{\Delta}{}}{\assign{\ctxtapp{\ctxt{L'}}{\substitute{x}{u}{s}}}{\tau}}}{\SysBang}
  \quad
  \derivable{\Phi_{r}}{\sequB{\Delta'}{\assign{r}{\M'}}}{\SysBang}
}{
  \derivable{\Phi'}{\sequB{\ctxtsum{(\ctxtres{\Gamma'}{x}{})}{\ctxtsum{(\ctxtres{\Delta}{y}{})}{\Delta'}{}}{}}{\assign{\ctxtapp{\ctxt{L}}{\substitute{x}{u}{s}}}{\tau}}}{\SysBang}
}{\ruleBESubs} $$
      since we may assume that $y \notin \dom{\Gamma'}$. Notice that
      $\size{\Phi} = \size{\Psi} + \size{\Phi_{r}} +1 > \size{\Psi'} +
      \size{\Phi_{r}} +1  = \size{\Phi'}$.
    \end{itemize}
    
    \item Rule $\dBang$. Then $t = \termder{\ctxtapp{\ctxt{L}}{\termbang{s}}}$
    and $t' = \ctxtapp{\ctxt{L}}{s}$. We proceed by induction on $\ctxt{L}$.
    \begin{itemize}
      \item $\ctxt{L} = \Box$. This case is immediate since $$
\prooftree
  \Rule{
    \derivable{\Phi'}{\sequB{\Gamma}{\assign{s}{\tau}}}{\SysBang}
  }{
    \sequB{\Gamma}{\assign{\termbang{s}}{\intertype{\tau}{}}}
  }{\ruleBBang}
\justifies
  \derivable{\Phi}{\sequB{\Gamma}{\assign{\termder{\termbang{s}}}{\tau}}}{\SysBang}
\using
  \ruleBDer
\endprooftree $$
      \item $\ctxt{L} = \termsubs{y}{r}{\ctxt{L'}}$. Then $\Gamma =
      \ctxtsum{(\ctxtres{\Gamma'}{x}{})}{\Delta}{}$ with $\Gamma'(x) = \M$ and $$
\prooftree
  \Rule{
    \derivable{\Phi_{\ctxt{L'}}}{\sequB{\Gamma'}{\assign{\ctxtapp{\ctxt{L'}}{\termbang{s}}}{\intertype{\tau}{}}}}{\SysBang}
    \quad
    \derivable{\Phi_{r}}{\sequB{\Delta}{\assign{r}{\M}}}{\SysBang}
  }{
    \sequB{\Gamma}{\assign{\ctxtapp{\ctxt{L}}{\termbang{s}}}{\intertype{\tau}{}}}
  }{\ruleBESubs}
\justifies
  \derivable{\Phi}{\sequB{\Gamma}{\assign{\termder{\ctxtapp{\ctxt{L}}{\termbang{s}}}}{\tau}}}{\SysBang}
\using
  \ruleBDer
\endprooftree
      $$ Thus, we build $$
\Rule{
  \derivable{\Phi_{\ctxt{L'}}}{\sequB{\Gamma'}{\assign{\ctxtapp{\ctxt{L'}}{\termbang{s}}}{\intertype{\tau}{}}}}{\SysBang}
}{
  \derivable{\Psi}{\sequB{\Gamma'}{\assign{\termder{\ctxtapp{\ctxt{L'}}{\termbang{s}}}}{\tau}}}{\SysBang}
}{\ruleBDer}
      $$ and by the \ih there exists
      $\derivable{\Psi'}{\sequB{\Gamma'}{\assign{\ctxtapp{\ctxt{L'}}{s}}{\tau}}}{\SysBang}$.
      Hence we conclude $$
\Rule{
  \derivable{\Psi'}{\sequB{\Gamma'}{\assign{\ctxtapp{\ctxt{L'}}{s}}{\tau}}}{\SysBang}
  \quad
  \derivable{\Phi_{r}}{\sequB{\Delta}{\assign{r}{\M}}}{\SysBang}
}{
  \derivable{\Phi'}{\sequB{\Gamma}{\assign{\ctxtapp{\ctxt{L}}{s}}{\tau}}}{\SysBang}  
}{\ruleBESubs} $$
      Notice that $\size{\Phi} = \size{\Psi} + \size{\Phi_{r}} +1 > \size{\Psi'} +
      \size{\Phi_{r}} +1 = \size{\Phi'}$.
    \end{itemize}
  \end{itemize}
  
\item All the inductive cases for
  $t \rewrite{\bangweak} t'$ are straightforward by the \ih
\end{itemize}
\end{proof}

In order to prove subject expansion, an anti-substitution lemma is needed:

\begin{lemma}[Anti-Substitution]
\label{l:antisubstitution-bang}
If
$\derivable{\Phi_{\substitute{x}{u}{t}}}{\sequB{\Gamma'}{\assign{\substitute{x}{u}{t}}{\tau}}}{\SysBang}$,
then there exists
$\derivable{\Phi_t}{\sequB{\Gamma; \assign{x}{\intertype{\sigma_i}{i \in I}}}{\assign{t}{\tau}}}{\SysBang}$
and
$\many{\derivable{\Phi^i_u}{\sequB{\Delta_i}{\assign{u}{\sigma_i}}}{\SysBang}}{i \in I}$
such that $\Gamma' = \ctxtsum{\Gamma}{\Delta_i}{i \in I}$ and
$\size{\Phi_{\substitute{x}{u}{t}}} = \size{\Phi_{t}}  +_{i \in
I}{\size{\Phi^i_{u}}} - |I|$.
\end{lemma}

\begin{proof}
By induction on $t$.
\end{proof}

\begin{lemma}[Weighted Subject Expansion]
\label{l:se-bang}
Let $\derivable{\Phi'}{\sequB{\Gamma}{\assign{t'}{\tau}}}{\SysBang}$. If
$t \rewrite{\bangweak} t'$, then there is
$\derivable{\Phi}{\sequB{\Gamma}{\assign{t}{\tau}}}{\SysBang}$
such that $\size{\Phi} > \size{\Phi'}$.
\end{lemma}

\begin{proof}

By induction on $t \rewrite{\bangweak} t'$.
\begin{itemize}
  \item For the base cases we
    have to consider three rules:
  \begin{itemize}
    \item Rule $\dBeta$. Then, $t =
    \termapp{\ctxtapp{\ctxt{L}}{\termabs{x}{s}}}{u}$ and $t' =
    \ctxtapp{\ctxt{L}}{\termsubs{x}{u}{s}}$. We proceed by induction on
    $\ctxt{L}$.
    \begin{itemize}
      \item $\ctxt{L} = \Box$. Then, $\Gamma = \ctxtsum{(\ctxtres{\Gamma'}{x}{})}{\Delta}{}$ such
      that \[
      \Rule{\derivable{\Phi'_s}{\sequB{\Gamma'}{\assign{s}{\tau}}}{\SysBang}
            \quad
            \derivable{\Phi'_u}{\sequB{\Delta}{\assign{u}{\Gamma'(x)}}}{\SysBang}
           }
           {\derivable{\Phi'}{\sequB{\ctxtsum{(\ctxtres{\Gamma'}{x}{})}{\Delta}{}}{\assign{\termsubs{x}{u}{s}}{\tau}}}{\SysBang}}
           {\ruleBESubs}
      \] then, we construct \[
      \Rule{\Rule{\derivable{\Phi'_s}{\sequB{\Gamma'}{\assign{s}{\tau}}}{\SysBang}}
                 {\sequB{\ctxtres{\Gamma'}{x}{}}{\assign{\termabs{x}{s}}{\functtype{\Gamma'(x)}{\tau}}}}
                 {\ruleBArrowI}
            \quad
            \derivable{\Phi'_u}{\sequB{\Delta}{\assign{u}{\Gamma'(x)}}}{\SysBang}
           }
           {\derivable{\Phi}{\sequB{\ctxtsum{(\ctxtres{\Gamma'}{x}{})}{\Delta}{}}{\assign{\termapp{(\termabs{x}{s})}{u}}{\tau}}}{\SysBang}}
           {\ruleBArrowE} \]
      Note that $\size{\Phi} = \size{\Phi'} + 1$.
      
      \item $\ctxt{L} = \termsubs{y}{r}{\ctxt{L'}}$. Then, $\Gamma =
      \ctxtsum{(\ctxtres{\Gamma'}{y}{})}{\Delta}{}$ such that \[
      \Rule{\derivable{\Psi'}{\sequB{\Gamma'}{\assign{\ctxtapp{\ctxt{L'}}{\termsubs{x}{u}{s}}}{\tau}}}{\SysBang}
            \quad
            \derivable{\Phi'_r}{\sequB{\Delta}{\assign{r}{\Gamma'(y)}}}{\SysBang}
           }
           {\derivable{\Phi'}{\sequB{\ctxtsum{(\ctxtres{\Gamma'}{y}{})}{\Delta}{}}{\assign{\termsubs{y}{r}{\ctxtapp{\ctxt{L'}}{\termsubs{x}{u}{s}}}}{\tau}}}{\SysBang}}
           {\ruleBESubs}
      \] By \ih on $\Psi'$ we have a derivation
      $\derivable{\Psi}{\sequB{\Gamma'}{\assign{\termapp{\ctxtapp{\ctxt{L'}}{\termabs{x}{s}}}{u}}{\tau}}}{\SysBang}$
      with $\size{\Psi} > \size{\Psi'}$.
      Moreover, by rule $\ruleBArrowE$, $\Gamma' =
      \ctxtsum{\Gamma'_1}{\Gamma'_2}{}$ and \[
      \Rule{\derivable{\Phi_{\ctxt{L'}}}{\sequB{\Gamma'_1}{\assign{\ctxtapp{\ctxt{L'}}{\termabs{x}{s}}}{\functtype{\M}\tau}}}{\SysBang}
            \quad
            \derivable{\Phi_u}{\sequB{\Gamma'_2}{\assign{u}{\M}}}{\SysBang}
           }
           {\derivable{\Psi}{\sequB{\Gamma'}{\assign{\termapp{\ctxtapp{\ctxt{L'}}{\termabs{x}{s}}}{u}}{\tau}}}{\SysBang}}
           {\ruleBArrowE}
      \] Moreover, by hypothesis of rule $\dB$, $y \notin \fv{u}$. Thus, in
      particular, $y \notin \dom{\Gamma'_2}$ and $\Gamma'(y) = \Gamma'_1(y)$.
      Then, we construct 
      \[\kern-5em
      \Rule{\Rule{\derivable{\Phi_{\ctxt{L'}}}{\sequB{\Gamma'_1}{\assign{\ctxtapp{\ctxt{L'}}{\termabs{x}{s}}}{\functtype{\M}\tau}}}{\SysBang}
                  \quad
                  \derivable{\Phi'_r}{\sequB{\Delta}{\assign{r}{\Gamma'_1(y)}}}{\SysBang}
                 }
                 {\sequB{\ctxtsum{(\ctxtres{\Gamma'_1}{y}{})}{\Delta}{}}{\assign{\ctxtapp{\ctxt{L}}{\termabs{x}{s}}}{\functtype{\M}{\tau}}}}
                 {\ruleBESubs}
            \quad
            \derivable{\Phi_u}{\sequB{\Gamma'_2}{\assign{u}{\M}}}{\SysBang}
           }
           {\derivable{\Phi}{\sequB{\ctxtsum{(\ctxtres{\Gamma'}{y}{})}{\Delta}{}}{\assign{\termapp{\ctxtapp{\ctxt{L}}{\termabs{x}{s}}}{u}}{\tau}}}{\SysBang}}
           {\ruleBArrowE}
      \] 
      and conclude with \[
      \begin{array}{rcl}
      \size{\Phi} & = & \size{\Phi_{\ctxt{L'}}} + \size{\Phi'_r} + \size{\Phi_u} +2 \\
                  & = & \size{\Psi} + \size{\Phi'_r} + 2 \\
                  & > & \size{\Psi'} + \size{\Phi'_r} + 1 \\
                  & = & \size{\Phi'}
      \end{array} \]
    \end{itemize}
    
    \item Rule $\sBang$. Then, $t =
    \termsubs{x}{\ctxtapp{\ctxt{L}}{\termbang{u}}}{s}$ and $t' =
    \ctxtapp{\ctxt{L}}{\substitute{x}{u}{s}}$. We proceed by induction on
    $\ctxt{L}$.
    \begin{itemize}
      \item $\ctxt{L} = \Box$. By Lemma~\ref{l:antisubstitution-bang} with
      $\Phi'$, there exist 
      $\derivable{\Phi_s}{\sequB{\Gamma'; \assign{x}{\intertype{\sigma_i}{i \in I}}}{\assign{s}{\tau}}}{\SysBang}$
      and
      $\many{\derivable{\Phi^i_u}{\sequB{\Delta_i}{\assign{u}{\sigma_i}}}{\SysBang}}{i \in I}$
      such that $\Gamma = \ctxtsum{\Gamma'}{\Delta_i}{i \in I}$ and
      $\size{\Phi'} = \size{\Phi_{s}} +_{\iI} \size{\Phi^i_{u}} - |I|$. Then, we construct \[
      \Rule{\derivable{\Phi_s}{\sequB{\Gamma'; \assign{x}{\intertype{\sigma_i}{i \in I}}}{\assign{s}{\tau}}}{\SysBang}
            \quad
            \Rule{\many{\derivable{\Phi^i_u}{\sequB{\Delta_i}{\assign{u}{\sigma_i}}}{\SysBang}}{i \in I}}
                 {\sequB{\ctxtsum{}{\Delta_i}{i \in I}}{\assign{\termbang{u}}{\intertype{\sigma_i}{i \in I}}}}
                 {\ruleBBang}
           }
           {\derivable{\Phi}{\sequB{\Gamma}{\assign{t}{\tau}}}{\SysBang}}
           {\ruleBESubs}
      \] and conclude since $\size{\Phi} = \size{\Phi_s} +_{\iI} \size{\Phi^i_{u}} + 1 > \size{\Phi'}$.
      
      \item $\ctxt{L} = \termsubs{y}{r}{\ctxt{L'}}$. Then, $\Gamma =
      \ctxtsum{(\ctxtres{\Gamma'}{y}{})}{\Delta}{}$ such that \[
      \Rule{\derivable{\Psi'}{\sequB{\Gamma'}{\assign{\ctxtapp{\ctxt{L'}}{\substitute{x}{u}{s}}}{\tau}}}{\SysBang}
            \quad
            \derivable{\Phi'_r}{\sequB{\Delta}{\assign{r}{\Gamma'(y)}}}{\SysBang}
           }
           {\derivable{\Phi'}{\sequB{\ctxtsum{(\ctxtres{\Gamma'}{y}{})}{\Delta}{}}{\assign{\termsubs{y}{r}{\ctxtapp{\ctxt{L'}}{\substitute{x}{u}{s}}}}{\tau}}}{\SysBang}}
           {\ruleBESubs}
      \] By \ih on $\Psi'$ we have a derivation
      $\derivable{\Psi}{\sequB{\Gamma'}{\assign{\termsubs{x}{\ctxtapp{\ctxt{L'}}{\termbang{u}}}{s}}{\tau}}}{\SysBang}$
      with $\size{\Psi} > \size{\Psi'}$. Moreover, by rule $\ruleBESubs$,
      $\Gamma' = \ctxtsum{\Gamma'_1}{\Gamma'_2}{}$ and \[
      \Rule{\derivable{\Phi_s}{\sequB{\Gamma'_1}{\assign{s}{\tau}}}{\SysBang}
            \quad
            \derivable{\Phi_{\ctxt{L'}}}{\sequB{\Gamma'_2}{\assign{\ctxtapp{\ctxt{L'}}{\termbang{u}}}{\Gamma'_1(x)}}}{\SysBang}
           }
           {\derivable{\Psi}{\sequB{\Gamma'}{\assign{\termsubs{x}{\ctxtapp{\ctxt{L'}}{\termbang{u}}}{s}}{\tau}}}{\SysBang}}
           {\ruleBESubs}
      \] Moreover, by hypothesis of rule $\sBang$, $y \notin \fv{s}$. Thus, in
      particular, $y \notin \dom{\Gamma'_1}$ and $\Gamma'(y) = \Gamma'_2(y)$.
      Then, we construct \[\kern-1em
      \Rule{\derivable{\Phi_s}{\sequB{\Gamma'_1}{\assign{s}{\tau}}}{\SysBang}
            \quad
            \Rule{\derivable{\Phi_{\ctxt{L'}}}{\sequB{\Gamma'_2}{\assign{\ctxtapp{\ctxt{L'}}{\termbang{u}}}{\Gamma'_1(x)}}}{\SysBang}
                  \quad
                  \derivable{\Phi'_r}{\sequB{\Delta}{\assign{r}{\Gamma'_2(y)}}}{\SysBang}
                 }
                 {\sequB{\ctxtsum{(\ctxtres{\Gamma'_2}{y}{})}{\Delta}{}}{\assign{\ctxtapp{\ctxt{L}}{\termbang{u}}}{\Gamma'_1(x)}}}
                 {\ruleBESubs}
           }
           {\derivable{\Phi}{\sequB{\Gamma}{\assign{t}{\tau}}}{\SysBang}}
           {\ruleBESubs}
      \] and conclude with \[
      \begin{array}{rcl}
      \size{\Phi} & = & \size{\Phi_s} + \size{\Phi_{\ctxt{L'}}} + \size{\Phi'_r} +2\\
                  & = & \size{\Psi} + \size{\Phi'_r} +2 \\
                  & > & \size{\Psi'} + \size{\Phi'_r} +1 \\
                  & = & \size{\Phi'}
      \end{array} \]
    \end{itemize}
    
    \item Rule $\dBang$. Then, $t =
    \termder{(\ctxtapp{\ctxt{L}}{\termbang{s}})}$ and $t' =
    \ctxtapp{\ctxt{L}}{s}$. We proceed by induction on $\ctxt{L}$.
    \begin{itemize}
      \item $\ctxt{L} = \Box$. We have a derivation
      $\derivable{\Phi'}{\sequB{\Gamma}{\assign{s}{\tau}}}{\SysBang}$ and we
      construct \[
      \Rule{\Rule{\derivable{\Phi'}{\sequB{\Gamma}{\assign{s}{\tau}}}{\SysBang}}
                 {\sequB{\Gamma}{\assign{\termbang{s}}{\intertype{\tau}{}}}}
                 {\ruleBBang}
           }
           {\derivable{\Phi}{\sequB{\Gamma}{\assign{\termder{(\termbang{s})}}{\tau}}}{\SysBang}}
           {\ruleBDer}
      \] to conclude since $\size{\Phi} = \size{\Phi'} +1$.
      
      \item $\ctxt{L} = \termsubs{y}{r}{\ctxt{L'}}$. Then, $\Gamma =
      \ctxtsum{(\ctxtres{\Gamma'}{y}{})}{\Delta}{}$ such that \[
      \Rule{\derivable{\Psi'}{\sequB{\Gamma'}{\assign{\ctxtapp{\ctxt{L'}}{s}}{\tau}}}{\SysBang}
            \quad
            \derivable{\Phi'_r}{\sequB{\Delta}{\assign{r}{\Gamma'(y)}}}{\SysBang}
           }
           {\derivable{\Phi'}{\sequB{\ctxtsum{(\ctxtres{\Gamma'}{y}{})}{\Delta}{}}{\assign{\termsubs{y}{r}{\ctxtapp{\ctxt{L'}}{s}}}{\tau}}}{\SysBang}}
           {\ruleBESubs}
      \] By \ih on $\Psi'$ we have a derivation
      $\derivable{\Psi}{\sequB{\Gamma'}{\assign{\termder{(\ctxtapp{\ctxt{L'}}{\termbang{s}})}}{\tau}}}{\SysBang}$
      with $\size{\Psi} > \size{\Psi'}$. Moreover, by rule $\ruleBDer$, \[
      \Rule{\derivable{\Phi_{\ctxt{L'}}}{\sequB{\Gamma'}{\assign{\ctxtapp{\ctxt{L'}}{\termbang{s}}}{\intertype{\tau}{}}}}{\SysBang}}
           {\derivable{\Psi}{\sequB{\Gamma'}{\assign{\termder{(\ctxtapp{\ctxt{L'}}{\termbang{s}})}}{\tau}}}{\SysBang}}
           {\ruleBDer}
      \] Then, we construct \[
      \Rule{\Rule{\derivable{\Phi_{\ctxt{L'}}}{\sequB{\Gamma'}{\assign{\ctxtapp{\ctxt{L'}}{\termbang{s}}}{\intertype{\tau}{}}}}{\SysBang}
                  \quad
                  \derivable{\Phi'_r}{\sequB{\Delta}{\assign{r}{\Gamma'(y)}}}{\SysBang}
                 }
                 {\sequB{\ctxtsum{(\ctxtres{\Gamma'}{y}{})}{\Delta}{}}{\assign{\ctxtapp{\ctxt{L}}{\termbang{s}}}{\intertype{\tau}{}}}}
                 {\ruleBESubs}
           }
           {\derivable{\Phi}{\sequB{\Gamma}{\assign{t}{\tau}}}{\SysBang}}
           {\ruleBDer}
      \] and conclude with \[
      \begin{array}{rcl}
      \size{\Phi} & = & \size{\Phi_{\ctxt{L'}}} + \size{\Phi'_r} +2 \\
                  & = & \size{\Psi} + \size{\Phi'_r} +2  \\
                  & > & \size{\Psi'} + \size{\Phi'_r} +1 \\
                  & = & \size{\Phi'}
      \end{array} \]
    \end{itemize}
  \end{itemize}
  
  \item All the inductive cases for $t \rewrite{\bangweak} t'$ are straightforward by the \ih
\end{itemize}
\end{proof}

Erasing steps like $\termsubs{x}{\termbang{z}}{y} \rewrite{\sBang} y$
may seem problematic for subject reduction and expansion, but they are
not: the variable $x$ is necessarily assigned a type $\emul$
  in the corresponding typing context, and the term $\termbang{z}$ is then
  necessarily typed with $\emul$, so there is no loss of information
since the contexts allowing to type the redex and the reduced term are
the same.

Typable terms are necessarily weak clash-free:
\begin{lemma}
\label{l:clashes-do-not-type-Bang}
If $\derivable{\Phi}{\sequB{\Gamma}{\assign{t}{\sigma}}}{\SysBang}$, then $t$
is $\cfz$.
\end{lemma}

\begin{proof}
Assume towards a contradiction that $t$ is not $\cfz$, \ie there exists
a weak context $\ctxt{W}$ and a clash  $c$ such that $t \rewriten{\bangweak}
\ctxtapp{\ctxt{W}}{c}$. Then, Lemma~\ref{l:wsr-bang} gives
$\derivable{\Phi'}{\sequB{\Gamma}{\assign{\ctxtapp{\ctxt{W}}{c}}{\sigma}}}{\SysBang}$.
If we show that a term of the form $\ctxtapp{\ctxt{W}}{c}$ cannot be typed in
system $\SysBang$, we are done. This follows by straightforward induction on
$\ctxt{W}$. The base case is when $\ctxt{W} = \Box$. For every possible $c$, it
is immediate to see that there is a mismatch between its syntactical form and
the typing rules of system $\SysBang$. For instance, if $c =
\termapp{\ctxtapp{\ctxt{L}}{\termbang{t}}}{u}$, then
$\ctxtapp{\ctxt{L}}{\termbang{t}}$ should  have a functional type by rule
$\ruleBArrowE$ but it can only be assigned a multiset type by rules
$\ruleBESubs$ and $\ruleBBang$. As for the inductive case, an easy inspection
of the typing rules  shows for all terms $t$ and weak contexts $\ctxt{W}$, $t$
must be typed in order to type $\ctxtapp{\ctxt{W}}{t}$.
\end{proof}

However, normal terms are not necessarily clash-free, but the type
system captures weak clash-freeness of normal terms. Said
differently, when restricted to $\ipnrml$, typability exactly
corresponds to weak clash-freeness.

\begin{theorem}
\label{t:clashfree-bang}
Let $t \in \TermExplicit$. Then, $\cfnrml{t}$ iff $\pnrml{t}$ and $t$ is
$\SysBang$-typable.
\end{theorem}

\begin{proof}
By simultaneous induction on the following claims: 
\begin{enumerate}
  \item\label{t:clashfree-bang:ne} $\cfntrl{t}$ iff $\pntrl{t}$ and for every $\tau$ there exists $\Gamma$ such that $\derivable{}{\sequB{\Gamma}{\assign{t}{\tau}}}{\SysBang}$.
  \item\label{t:clashfree-bang:na} $\cfnabs{t}$ iff $\pnabs{t}$ and
    there exist $\Gamma$ and $\M$ such that
    $\derivable{}{\sequB{\Gamma}{\assign{t}{\M}}}{\SysBang}$.
  \item\label{t:clashfree-bang:nb} $\cfnbang{t}$ iff $\pnbang{t}$ and
    there exist $\Gamma$ and $\tau$ such that $\derivable{}{\sequB{\Gamma}{\assign{t}{\tau}}}{\SysBang}$.
  \item\label{t:clashfree-bang:no} $\cfnrml{t}$ iff $\pnrml{t}$ and there exist $\Gamma$ and $\tau$ such that $\derivable{}{\sequB{\Gamma}{\assign{t}{\tau}}}{\SysBang}$.
\end{enumerate}
We first show the left-to-right implications by only analysing the key cases.

If $t =
\cfntrl{x}$, then $\pntrl{t}$ and for every type $\tau$ we conclude by
$\ruleBAxiom$.

If $t = \cfntrl{\termapp{s}{u}}$, by definition
$\cfntrl{s}$ and $\cfnabs{u}$. Let $\tau$ be any type. By \ih
(\ref{t:clashfree-bang:na}) $\pnabs{u}$ and
$\derivable{}{\sequB{\Delta}{\assign{u}{\M}}}{\SysBang}$. Then, by \ih
(\ref{t:clashfree-bang:ne}) we get $\pntrl{s}$ and
$\derivable{}{\sequB{\Gamma}{\assign{s}{\functtype{\M}{\tau}}}}{\SysBang}$.
Moreover, $\pntrl{s}$ and $\pnabs{u}$ imply $\pnabs{s}$ and
$\pnrml{u}$ resp., hence $\pntrl{t}$. Thus, we conclude by
$\ruleBArrowE$,
$\derivable{}{\sequB{\ctxtsum{\Gamma}{\Delta}{}}{\assign{\termapp{s}{u}}{\tau}}}{\SysBang}$.

If $t = \cfnbang{\termabs{x}{s}}$, then $\cfnrml{s}$ by definition. By
\ih (\ref{t:clashfree-bang:no}) $\pnrml{s}$ and
$\derivable{}{\sequB{\Gamma}{\assign{s}{\tau}}}{\SysBang}$ for some
type $\tau$. Then, $\pnbang{t}$ and we conclude by $\ruleBArrowI$,
$\derivable{}{\sequB{\ctxtres{\Gamma}{x}{}}{\assign{\termabs{x}{s}}{\functtype{\Gamma(x)}{\tau}}}}{\SysBang}$.

If $t = \cfnabs{\termbang{s}}$, then $\pnabs{t}$ by definition and we
conclude by $\ruleBBang$,
$\derivable{}{\sequB{}{\assign{\termbang{s}}{\emul}}}{\SysBang}$.

If
$t = \cfntrl{\termder{s}}$, then $\cfntrl{s}$ by definition. Let
$\tau$ be any type, by \ih (\ref{t:clashfree-bang:ne}) we get $\pntrl{s}$ and
$\derivable{}{\sequB{\Gamma}{\assign{s}{\multiset{\tau}}}}{\SysBang}$.
Moreover, $\pntrl{s}$ implies $\pnbang{s}$ and hence
$\pntrl{t}$. Thus, we conclude by $\ruleBDer$,
$\derivable{}{\sequB{\Gamma}{\assign{\termder{s}}{\tau}}}{\SysBang}$.

If
$t = \termsubs{x}{u}{s}$, then $\cfntrl{u}$ and there are three
possible cases: $\cfntrl{s}$, $\cfnabs{s}$ or $\cfnbang{s}$. In either
case, by the proper \ih we get
$\derivable{}{\sequB{\Gamma}{\assign{s}{\tau}}}{\SysBang}$
(resp. $\M$) and conclude by $\ruleBESubs$, given that \ih (\ref{t:clashfree-bang:ne}) on $u$
implies $\pntrl{u}$ and
$\derivable{}{\sequB{\Delta}{\assign{u}{\Gamma(x)}}}{\SysBang}$. Note
that $\pntrl{u}$ in turn implies $\pnbang{u}$, thus $t$ remains
  in the same set as $s$ ($\ipntrl$, $\ipnabs$ or $\ipnbang$ resp.).

The right-to-left implications uses Lemma~\ref{l:clashfree-bang}.  
\end{proof}

Typability can be shown to (qualitatively and quantitatively) characterise
normalisation. The type system $\SysBang$ is  {\em sound} (all the typable terms are normalising) and {\em complete} (all the normalising terms are typable).
  
\begin{theorem}[Soundness and Completeness for System $\SysBang$]
\label{t:correctness-bang}
The term $t$ is $\SysBang$-typable iff $t$ $\bangweak$-normalises to a term
$p \in \icfnrml$. Moreover, if
$\derivable{\Phi}{\sequB{\Gamma}{\assign{t}{\tau}}}{\SysBang}$, then
$t \rewriten{\bangweak}^{(\cbeta,\cexp)} p$ and $\size{\Phi} \geq \cbeta +
\cexp + \wsize{p}$.
\end{theorem}

\begin{proof}
The soundness proof is straightforward by Lemma~\ref{l:wsr-bang} and
Theorem~\ref{t:clashfree-bang}. Observe that the argument is simply combinatorial,
no reducibility argument is needed.
For the completeness proof, we
reason by induction on the length of the $\bangweak$-normalising sequence. For
the base case, we use Theorem~\ref{t:clashfree-bang} which states that
$\cfnrml{p}$ implies $p$ is $\SysBang$-typable. For the inductive case we use
Lemma~\ref{l:se-bang}. The \emph{moreover} statement holds by
Lemma~\ref{l:wsr-bang} and~\ref{l:se-bang}, and the fact that the size of the
type derivation of $p$ is greater than or equal to $\wsize{p}$.
\end{proof}

The previous theorem can be illustrated by the term $t_0 =
\termapp{\termapp{\termder{(\termbang{\Kterm})}}{(\termbang{\id})}}{(\termbang{\Omega})}$
defined in Example~\ref{example:t0}, which normalises in $5$ steps to a normal
form of $\bangweak$-size $1$, the sum of the two being bounded by the size $11$
of its type derivation $\Phi_0$ given in Example~\ref{example:t0-typed}.


\section{Capturing Call-by-Name and Call-by-Value}
\label{s:cbname-cbvalue}

This section explores the CBN/CBV embeddings into the $\BangRev$-calculus. For
CBN, we slightly adapt Girard's translation into LL~\cite{Girard87}, which
preserves normal forms and is sound and complete with respect to the standard
(quantitative) type system~\cite{Gardner94}. For CBV, however, we reformulate
both the translation and the type system, so that preservation of normal forms
and completeness are restored. In both cases, we specify the operational
semantics of CBN and CBV by means of a very simple notion of explicit
substitution, see for example~\cite{AccattoliP12}. 

Terms ($\TermLambda$), values and contexts are defined as follows:
\begin{center}
\begin{tabular}{rrcll}
\textbf{(Terms)}                  & $t,u$      & $\Coloneq$ & $v \mid \termapp{t}{u} \mid \termsubs{x}{u}{t}$ \\
\textbf{(Values)}                 & $v$        & $\Coloneq$ & $x \in \TermVariable \mid \termabs{x}{t}$ \\
\textbf{(List Contexts)}          & $\ctxt{L}$ & $\Coloneq$ & $\Box \mid \termsubs{x}{t}{\ctxt{L}}$ \\
\textbf{(Call-by-Name Contexts)}  & $\ctxt{N}$ & $\Coloneq$ & $\Box \mid \termapp{\ctxt{N}}{t} \mid \termabs{x}{\ctxt{N}} \mid \termsubs{x}{u}{\ctxt{N}}$ \\
\textbf{(Call-by-Value Contexts)} & $\ctxt{V}$ & $\Coloneq$ & $\Box \mid \termapp{\ctxt{V}}{t} \mid \termapp{t}{\ctxt{V}} \mid \termsubs{x}{u}{\ctxt{V}} \mid \termsubs{x}{\ctxt{V}}{t}$
\end{tabular}
\end{center}

As in Sec.~\ref{s:bang} we use the predicate $\pabs{t}$ iff $t =
\ctxtapp{\ctxt{L}}{\termabs{x}{t'}}$.  We also use the predicates $\papp{t}$
iff $t = \ctxtapp{\ctxt{L}}{\termapp{t'}{t''}}$ and $\pvar{t}$ iff $t =
\ctxtapp{\ctxt{L}}{x}$.

The \emphdef{Call-by-Name} reduction relation $\rewrite{\callbyname}$ is
defined as the closure of the rules $\dBeta$ and
$\sTerm$ presented below under contexts $\ctxt{N}$, while the \emphdef{Call-by-Value} reduction relation
$\rewrite{\callbyvalue}$ is defined as the closure of 
the rules $\dBeta$ and  $\sVal$ below under contexts $\ctxt{V}$. Equivalently,
$\mathbin{\rewrite{\callbyname}} \eqdef \ctxt{N}(\rrule{\dBeta} \cup
\rrule{\sTerm})$ and $\mathbin{\rewrite{\callbyvalue}} \eqdef
\ctxt{V}(\rrule{\dBeta} \cup \rrule{\sVal})$ and \[
\begin{array}{rcl}
\termapp{\ctxtapp{\ctxt{L}}{\termabs{x}{t}}}{u} & \rrule{\dBeta} & \ctxtapp{\ctxt{L}}{\termsubs{x}{u}{t}} \\   
\termsubs{x}{u}{t}                              & \rrule{\sTerm} & \substitute{x}{u}{t} \\
\termsubs{x}{\ctxtapp{\ctxt{L}}{v}}{t}          & \rrule{\sVal}  & \ctxtapp{\ctxt{L}}{\substitute{x}{v}{t}}
\end{array} \] 
We write $t \not\rewrite{\callbyname}$ (resp. $t \not \rewrite{\callbyvalue}$),
and call $t$ an \emphdef{$\callbyname$-normal form} (resp.
\emphdef{$\callbyvalue$-normal form}), if $t$ cannot be reduced by means of
$\rewrite{\callbyname}$ (resp. $\rewrite{\callbyvalue}$).

Observe that we use CBN and CBV formulations based on distinguished
multiplicative (\cf $\dBeta$) and exponential (\cf $\sTerm$ and $\sVal$) rules,
inheriting the nature of cut elimination rules in LL. Moreover, CBN is to be
understood as \emph{head} CBN reduction~\cite{Barendregt84}, \ie reduction does
not take place in arguments of applications (and in the present case in
arguments of substitutions as well) while CBV corresponds to \emph{open} CBV
reduction~\cite{AccattoliP12,AccattoliG16}, \ie reduction does not take place
inside abstractions.

\begin{theorem}
\label{t:confluence-cbname-cbvalue}
The reduction relations $\rewrite{\callbyname}$ 
and $\rewrite{\callbyvalue}$ are both confluent.
\end{theorem}

\begin{proof} The proofs of these properties are folklore.
    They use  an abstract result in rewriting theory stating that
    the union of two reduction relations is confluent if
    both are 
    confluent and commute with each other~\cite{Terese03}. More precisely,
    two arbitrary reduction relations $\rewrite{1}$ and $\rewrite{2}$
    \emphdef{commute}
    iff for all $t_0,t_1,t_2$
    such that $t_0 \rewriten{1} t_1$
    and $t_0 \rewriten{2} t_2$, there exists $t_3$ such that
    $t_1 \rewriten{2} t_3$ and
     $t_2 \rewriten{1} t_3$.
The reasoning in our case is as follows:
\begin{itemize}
\item Each rule $\rGen \in \{\dBeta, \sTerm, \sVal\}$ induces a \emphdef{complete} reduction relation $\rewrite{\rGen}$, \ie confluent and terminating. 
\item The relation $\rewrite{\dBeta}$ commutes with $\rewrite{\rGen}$,
  for $\rGen \in  \{\sTerm, \sVal\}$.
\end{itemize}
The two previous points are straightforward.
\end{proof}

\comment{
The set of $\callbyname$-normal forms can be alternatively characterised by the
following definition. \[
\Rule{\phantom{x \in \HCBNNF}}{x \in \HCBNNF}{}
\qquad
\Rule{t \in \HCBNNF}{tu \in \HCBNNF}{}
\qquad
\Rule{t \in \HCBNNF}{ t \in \CBNNF}{}\qquad
\Rule{t \in \CBNNF}{ \lambda x. t \in \CBNNF}{}
\]
The set of $\callbyvalue$-normal forms can be alternatively characterised by
the following definition. \[
\begin{array}{c}
\Rule{\phantom{x \in \VarHCBVNF}}{x \in \VarHCBVNF}{} \qquad
\Rule{t \in \VarHCBVNF, u \in \HCBVNF}{t\exsubs{x}{u} \in \VarHCBVNF}{}
\\
\\
\qquad
\Rule{t \in \VarHCBVNF, u \in \CBVNF}{tu \in \HCBVNF}{} \qquad
\Rule{t \in \HCBVNF, u \in \CBVNF}{tu \in \HCBVNF}{} \qquad
\Rule{t \in \HCBVNF, u \in \HCBVNF}{t\exsubs{x}{u} \in \HCBVNF}{}
\\
\\
\Rule{\phantom{\lambda x. t \in \CBVNF}}{\lambda x. t \in \CBVNF}{} \qquad
\Rule{t \in \VarHCBVNF}{t \in \CBVNF}{}\qquad
\Rule{t \in \HCBVNF}{t \in \CBVNF}{}\qquad
\Rule{t \in \CBVNF, u \in \HCBVNF}{t\exsubs{x}{u} \in \CBVNF}{}
\end{array} \]
}

\paragraph{{\bf Embeddings}}
The CBN and CBV embeddings into the $\BangRev$-calculus, written  $\cbn{\_}$ and
$\cbv{\_}$ resp., are inductively defined as: \[
\begin{array}{c@{\qquad}c}
\begin{array}{rcl}
\cbn{x}                    & \eqdef & x \\
\cbn{(\termabs{x}{t})}     & \eqdef & \termabs{x}{\cbn{t}} \\
\cbn{(\termapp{t}{u})}     & \eqdef & \termapp{\cbn{t}}{\termbang{\cbn{u}}} \\
\cbn{(\termsubs{x}{u}{t})} & \eqdef & \termsubs{x}{\termbang{\cbn{u}}}{\cbn{t}}
\end{array}
&
\begin{array}{rcl}
\cbv{x} & \eqdef & \termbang{x} \\
\cbv{(\termabs{x}{t})} & \eqdef & \termbang{\termabs{x}{\cbv{t}}} \\
\cbv{(\termapp{t}{u})} & \eqdef & \left \{
\begin{array}{ll}
  \termapp{\ctxtapp{\ctxt{L}}{s}}{\cbv{u}} & \text{if $\cbv{t} = \ctxtapp{\ctxt{L}}{\termbang{s}}$} \\
\termapp{\termder{(\cbv{t})}}{\cbv{u}}   & \text{otherwise}
\end{array}
\right. \\
\cbv{(\termsubs{x}{u}{t})} & \eqdef & \termsubs{x}{\cbv{u}}{\cbv{t}}
\end{array}
\end{array} \]

Both translations extend to list contexts $\ctxt{L}$ as expected. Observe that
the terms of the $\BangRev$-calculus that are in the image of the embeddings
never contain two consecutive $\termbang{\!}$ constructors. The CBN embedding
extends Girard's translation to explicit substitutions, while the CBV one is
different. Indeed, the translation of an application $\termapp{t}{u}$ is
usually defined as $\termapp{\termder{(\cbv{t})}}{\cbv{u}}$ (see for
example~\cite{EhrhardG16}). This definition does not preserve normal forms,
\ie $\termapp{x}{y}$ is a $\callbyvalue$-normal form but its translated version
$\termapp{\termder{(\termbang{x})}}{\termbang{y}}$ is not a $\bangweak$-normal
form. We restore this fundamental property by using the notion of
superdevelopment~\cite{Hyland78,vanoostrom21}, so that 
$\dBang$-reductions created by the translation are
executed on the fly.

\begin{example}
\label{ex:embeddings1}
Special $\lambda$-terms are $\idl = \termabs{z}{z}$, $\Kterml =
\termabs{x}{\termabs{y}{x}}$, $\Deltal = \termabs{x}{\termapp{x}{x}}$, and
$\Omegal = \termapp{\Deltal}{\Deltal}$. Observe that  $\cbn{\idl} = \idl = \id$
and $\cbn{\Kterml} = \Kterml = \Kterm$. In this example, we compute $\cbn{t}$
and $\cbv{t}$, where $t = \Kterml\idl\Omegal$. \[
\begin{array}{rll}
\cbn{t}
& =              & \cbn{(\Kterml\idl)}\termbang{\cbn{\Omegal}} \\
& =              & \cbn{\Kterml}\termbang{\cbn{\idl}}\termbang{\cbn{\Omegal}}  \\
& =              & \Kterm\termbang{\id}\termbang{(\cbn{\Deltal}\termbang{\cbn{\Deltal}})}  \\
& =              & \Kterm\termbang{\id}\termbang{(\termapp{(\termabs{x}{\termapp x {\termbang x}})}{\termbang{(\termabs{x}{\termapp x {\termbang x}}}}}))\\
& =              & \Kterm\termbang{\id}\termbang{(\termapp{\Delta}{\termbang{\Delta}})}\\
& =              & \Kterm\termbang{\id}\termbang{\Omega}\\
\end{array} \]
The $\BangRev$-term \cbn{t} is the same as the one obtained by performing a
$\dBang$-step starting from the term $t_0$ of Example~\ref{example:t0}. The CBV
embedding is a bit more involved, due to the superdevelopments. \[
\begin{array}{rll}
\cbv{t}
& =              & \termapp{\termder{\cbv{(\termapp\Kterml\idl)}}}{\cbv{\Omegal}} \\
& =              & \termapp{\termder{(\termapp{(\termabs{x}{\termbang{\termabs{y}{\termbang x}}})}{(\termbang{\termabs{x}{}\termbang x})})}}{\cbv{\Omegal}} \\
& =              & \termapp{\termder{(\termapp{(\termabs{x}{\termbang{\termabs{y}{\termbang x}}})}{(\termbang{\termabs{x}{}\termbang x})})}}      (\termapp{(\termabs{x}{\termapp x {\termbang x}})}{\termbang{(\termabs{x}{\termapp x {\termbang x}}}))              } \\
& =              & \termapp{\termder{(\termapp{(\termabs{x}{\termbang{\termabs{y}{\termbang x}}})}{(\termbang{\termabs{x}{}\termbang x})})}}  {\Omega    } \\
\end{array} \]
Notice that in this particular case we have $\cbn{\Omegal} = \cbv{\Omegal} =
\Omega$.
\end{example}

The example above will be continued after the proof of the fact that the
embeddings preserve the CBN and CBV reductions.

We first show (Lemma~\ref{l:normal-forms-cbn-cbv}) that the set of
$\callbyname$-normal forms and $\callbyvalue$-normal forms can be respectively
characterised by the following grammars:
\begin{center}
\begin{tabular}{rrcll}
\textbf{(CBN Neutral)} & $\HCBNNF$ & $\Coloneq$ & $x \in \TermVariable \mid \termapp{\HCBNNF}{t}$ \\
\textbf{(CBN Normal)}  & $\CBNNF$  & $\Coloneq$ & $\termabs{x}{\CBNNF} \mid \HCBNNF$ \\
\\
\textbf{(Variable)}    & $\VarHCBVNF$ & $\Coloneq$ & $x \in \TermVariable \mid \termsubs{x}{\HCBVNF}{\VarHCBVNF}$ \\
\textbf{(CBV Neutral)} & $\HCBVNF$    & $\Coloneq$ & $\termapp{\VarHCBVNF}{\CBVNF} \mid \termapp{\HCBVNF}{\CBVNF} \mid \termsubs{x}{\HCBVNF}{\HCBVNF}$ \\
\textbf{(CBV Normal)}  & $\CBVNF$     & $\Coloneq$ & $\termabs{x}{t} \mid \VarHCBVNF \mid \HCBVNF \mid \termsubs{x}{\HCBVNF}{\CBVNF}$
\end{tabular}
\end{center}

Note that the grammar for neutral and normal forms for CBN does not use the ES operator. The reason is that   the $\sTerm$-rule executing ES can always be fired in a CBN context, so that
  ES only appear in the right-hand sides of applications of neutral terms. 
Note that variables are left out of the definition of neutral terms for the CBV
case, since they are now considered values. Moreover, in this way both CBN and
CBV neutral terms translate to neutral terms of the $\BangRev$-calculus.

\begin{lemma}\mbox{}
\label{l:normal-forms-cbn-cbv}
Let $t \in \TermLambda$. 
\begin{enumerate}
  \item\label{l:normal-forms-cbn} $t \in \CBNNF$ iff $t \not
  \rewrite{\callbyname}$.
  \item\label{l:normal-forms-cbv} $t \in \CBVNF$ iff $t \not
  \rewrite{\callbyvalue}$.
\end{enumerate}
\end{lemma}

\begin{proof} 
  \begin{enumerate}
  \item We prove simultaneously the following statements:
    \begin{enumerate}
    \item \label{ldix:uno} $t \in \HCBNNF$ iff $t \not \rewrite{\callbyname}$ and $\neg \pabs{t}$.
    \item \label{ldix:dos} $t \in \CBNNF$ iff $t \not \rewrite{\callbyname}$.
    \end{enumerate}
    
    \begin{enumerate}
      \item \label{case-non-lambda} The left-to-right implication is
        straightforward. For the right-to-left implication we reason
        by induction on $t$.  Suppose $t \not
        \rewrite{\callbyname}$ and $\neg \pabs{t}$. If $t$ is a
        variable, then $t\in \HCBNNF$.  If $t$ is a substitution, then
        rule $\sTerm$ is applicable, so this
          case does not apply. If $t = \termapp{u}{u'}$, then $\neg
        \pabs{u}$ (otherwise $\dBeta$ would be applicable) and $u \not
        \rewrite{\callbyname}$ (otherwise $t$ would be
        $\callbyname$-reducible).  The \ih (\ref{ldix:uno}) gives $u
        \in \HCBNNF$ and thus we conclude $t=u u' \in \HCBNNF$.
      
      \item The left-to-right implication is straightforward. For the
        right-to-left implication we reason by induction on
        $t$. If $t = \termabs{x}{u}$, then $u \not
        \rewrite{\callbyname}$, and the \ih (\ref{ldix:dos}) gives $u
        \in \CBNNF$, which implies in turn $\lambda x. u \in \CBNNF$.
        Otherwise, we apply the previous case.
    \end{enumerate}
    
    \item We prove simultaneously the following statements:
    \begin{enumerate}
      \item \label{ldix:dos-uno} $t \in \VarHCBVNF$ iff $t \not \rewrite{\callbyvalue}$ and $\neg \pabs{t}$ and $\neg \papp{t}$.
      
      \item \label{ldix:dos-dos} $t \in \HCBVNF$ iff $t \not \rewrite{\callbyvalue}$ and $\neg
      \pabs{t}$ and $\neg \pvar{t}$.
      
      \item \label{ldix:dos-tres} $t \in \CBVNF$ iff $t \not \rewrite{\callbyvalue}$.
    \end{enumerate}
    \begin{enumerate}
      \item The left-to-right implication is straightforward. For the
        right-to-left implication we reason by induction on
        $t$. Suppose $t \not \rewrite{\callbyvalue}$
        and $\neg \pabs{t}$ and $\neg \papp{t}$.  Then necessarily
        $\pvar{t}$.  If $t = u\exsubs{x}{u'}$, then $u, u'\not
        \rewrite{\callbyvalue}$ (otherwise $t$ would be
        $\callbyvalue$-reducible), $\neg \pabs{u'}$ and $\neg
        \pvar{u'}$ (otherwise $t$ would be $\sVal$-reducible), $\neg
        \pabs{u}$ and $\neg \papp{u}$.  The \ih (\ref{ldix:dos-uno})
        gives $u \in \VarHCBVNF$ and the \ih (\ref{ldix:dos-dos})
        gives $u' \in \HCBVNF$, so that we conclude $t \in
        \VarHCBVNF$.  If $t=x$ we trivially conclude $t \in
        \VarHCBVNF$.
      
      \item The left-to-right implication is straightforward. For the
      right-to-left implication we reason by induction on $t$.
      Suppose $t \not \rewrite{\callbyvalue}$ and $\neg \pabs{t}$
      and $\neg \pvar{t}$. Then necessarily $\papp{t}$.
      If $t = u\exsubs{x}{u'}$, then $u, u'\not \rewrite{\callbyvalue}$
      (otherwise $t$ would be $\callbyvalue$-reducible),
      $\neg \pabs{u'}$ and $\neg \pvar{u'}$ (otherwise $t$ would be $\sVal$-reducible),
      $\neg \pabs{u}$ and $\neg \pvar{u}$.
      The \ih (\ref{ldix:dos-dos}) gives $u \in \HCBVNF$ and
      $u' \in \HCBVNF$, so that we conclude $t \in \HCBVNF$.
      If $t = uu'$, then $u, u'\not \rewrite{\callbyvalue}$
      (otherwise $t$ would be $\callbyvalue$-reducible),
      $\neg \pabs{u}$ (otherwise $t$ would be $\dBeta$-reducible).
      The \ih (\ref{ldix:dos-tres}) gives $u' \in \CBVNF$.
      Moreover, if $\neg \papp{u}$ (resp.  $\neg \pvar{u}$), then
      the \ih (\ref{ldix:dos-uno}) gives $u \in \VarHCBVNF$
(resp. the \ih (\ref{ldix:dos-dos}) gives $u \in \HCBVNF$)
      In both cases we conclude $t=uu' \in  \HCBVNF$.
Since $\neg \papp{u}$ or  $\neg \pvar{u}$) holds, we are done.

      \item The left-to-right implication is straightforward. For the
      right-to-left implication we reason by induction on $t$.
      Suppose $t \not \rewrite{\callbyvalue}$. If $t=\lambda x. u$,
      then $t \in \CBVNF$ is straightforward. If $t =
      u\exsubs{x}{u'}$, then $u, u'\not \rewrite{\callbyvalue}$
      (otherwise $t$ would be $\callbyvalue$-reducible), $\neg
      \pabs{u'}$ and $\neg \pvar{u'}$ (otherwise $t$ would be $\sVal$-reducible).
      The \ih (\ref{ldix:dos-tres}) and (\ref{ldix:dos-dos}) give
      $u \in \CBVNF$ and $u' \in \HCBVNF$, so that we conclude $t \in \CBVNF$.
      If $t=\termapp{u}{u'}$, then $u \not \rewrite{\callbyvalue}$ and $u' \not
      \rewrite{\callbyvalue}$.  The \ih (\ref{ldix:dos-tres}) on $u'$ gives $u' \in
      \CBVNF$.  Moreover, $\neg \pabs{u}$, otherwise $t$ would be
      $\callbyvalue$-reducible. Since either $\neg \pvar{u}$ or
        $\neg \papp{u}$ must hold, then the
        \ih (\ref{ldix:dos-uno}) or (\ref{ldix:dos-dos}) gives
      $u \in \VarHCBVNF$ or $u \in \HCBVNF$.
      Thus $\termapp{u}{u'} \in \HCBVNF\subseteq \CBVNF$ as required.
    \end{enumerate}
  \end{enumerate}
\end{proof}

The following lemma shows that the CBN and CBV embeddings into the $\BangRev$-calculus preserve the property of being in normal form.
\begin{lemma}
\label{l:preservation-cbn-cbv}
Let $t \in \TermLambda$.
\begin{enumerate}
  \item\label{l:preservation-cbn} If $t \not\rewrite{\callbyname}$, then $\cbn{t}\not\rewrite{\bangweak}$.
  \item\label{l:preservation-cbv} If $t \not\rewrite{\callbyvalue}$, then $\cbv{t}\not\rewrite{\bangweak}$.
\end{enumerate}
\end{lemma}

\begin{proof}\mbox{} 
\begin{enumerate}
  \item By Lemma~\ref{l:normal-forms-cbn-cbv} and Proposition~\ref{prop:normal} it
  is sufficient to prove that $t \in \CBNNF$ implies $\pnrml{\cbn{t}}$. This is
  done by simultaneously showing the following points:
  \begin{enumerate}
    \item\label{l:preservation-cbn-ne} If $t \in \HCBNNF$, then $\pntrl{\cbn{t}}$.
    \item\label{l:preservation-cbn-no} If $t \in \CBNNF$, then $\pnrml{\cbn{t}}$.
  \end{enumerate}
  
  \begin{itemize}
    \item If $t = x \in \HCBNNF$, then $\cbn{x} = x$ and thus $\pntrl{x}$
    holds.
    
    \item If $t = \termapp{s}{u} \in \HCBNNF$ comes from $s \in \HCBNNF$, then
    the \ih (\ref{l:preservation-cbn-ne}) gives $\pntrl{\cbn{s}}$ and hence
    $\pnabs{\cbn{s}}$. Moreover
    $\pnabs{\termbang{\cbn{u}}}$ also holds, which implies
    $\pnrml{\termbang{\cbn{u}}}$. We can then conclude
    $\pntrl{\termapp{\cbn{s}}{\termbang{\cbn{u}}}}$, \ie
    $\pntrl{\cbn{(\termapp{s}{u})}}$.
    
    \item If $t \in \CBNNF$ comes from $t \in \HCBNNF$, then $\pntrl{\cbn{t}}$
    holds by the previous point, which gives $\pnrml{\cbn{t}}$.
    
    \item If $t = \termabs{x}{u} \in \CBNNF$ comes from $u \in \CBNNF$, then
    the \ih (\ref{l:preservation-cbn-no}) gives $\pnrml{\cbn{u}}$, which implies
    $\pnbang{\termabs{x}{\cbn{u}}}$, thus giving
    $\pnrml{\termabs{x}{\cbn{u}}}$. We conclude since $\cbn{(\termabs{x}{u})} =
    \termabs{x}{\cbn{u}}$.
  \end{itemize}
  
  \item By Lemma~\ref{l:normal-forms-cbn-cbv} and Proposition~\ref{prop:normal} it
  is sufficient to prove that $t \in \CBVNF$ implies $\pnrml{\cbv{t}}$. This is
  done by simultaneously showing the following points:
  \begin{enumerate}
    \item\label{l:preservation-cbv-vr} If $t \in \VarHCBVNF$, then $\pnabs{\cbv{t}}$.
    \item\label{l:preservation-cbv-ne} If $t \in \HCBVNF$, then $\pntrl{\cbv{t}}$.
    \item\label{l:preservation-cbv-no} If $t \in \CBVNF$, then $\pnrml{\cbv{t}}$.
  \end{enumerate}
  
  \begin{itemize}
    \item If $t = x \in \VarHCBVNF$, then $\cbv{x} = \termbang{x}$ and
    $\pnabs{\termbang{x}}$.
    
    \item If $t = \termsubs{x}{u}{s} \in \VarHCBVNF$ comes from $s \in
    \VarHCBVNF$ and $u \in \HCBVNF$, then the \ih (\ref{l:preservation-cbv-vr}) gives $\pnabs{\cbv{s}}$ and
    $\pntrl{\cbv{u}}$, which in turn implies $\pnbang{\cbv{u}}$. Hence, it
    allows us to conclude $\pnabs{\termsubs{x}{\cbv{u}}{\cbv{s}}}$.
    
    \item If $t = \termapp{s}{u} \in \HCBVNF$ comes from $s \in \VarHCBVNF$ and
    $u \in \CBVNF$. It is immediate to see that $\cbv{s} =
    \ctxtapp{\ctxt{L}}{\termbang{x}}$ for some variable $x$. Thus, $\cbv{t} =
    \termapp{\ctxtapp{\ctxt{L}}{x}}{\cbv{u}}$. Moreover, by \ih (\ref{l:preservation-cbv-vr})
    $\pnabs{\cbv{s}}$ holds, which implies $\pnbang{r}$ for all $\exsubs{y}{r}$
    in $\ctxt{L}$. Thus, $\pntrl{\ctxtapp{\ctxt{L}}{x}}$ holds.
    This implies $\pnabs{\ctxtapp{\ctxt{L}}{x}}$, by definition. Also, the \ih
    (\ref{l:preservation-cbv-no}) on $u \in \CBVNF$ gives $\pnrml{\cbv{u}}$. Hence, we conclude
    $\pntrl{\cbv{t}}$.

    \item If $t = \termapp{s}{u} \in \HCBVNF$ comes from $s \in \HCBVNF$ and
    $u \in \CBVNF$. By \ih (\ref{l:preservation-cbv-ne}) $\pntrl{\cbv{s}}$ holds. In particular, it implies
    $\pnbang{\cbv{s}}$ (\ie $\neg\pbang{\cbv{s}}$). Hence, $\cbv{t} =
    \termapp{\termder{(\cbv{s})}}{\cbv{u}}$ and $\pntrl{\termder{(\cbv{s})}}$.
    Moreover, the latter implies $\pnabs{\termder{(\cbv{s})}}$, by definition.
    Also, the \ih (\ref{l:preservation-cbv-no}) on $u \in \CBVNF$ gives $\pnrml{\cbv{u}}$. Thus, we conclude
    $\pntrl{\cbv{t}}$.
    
    \item If $t = \termsubs{x}{u}{s} \in \HCBVNF$ comes from $s \in \HCBVNF$
    and $u \in \HCBVNF$, then the \ih (\ref{l:preservation-cbv-ne}) gives $\pntrl{\cbv{s}}$ and
    $\pntrl{\cbv{u}}$, which in turn implies $\pnbang{\cbv{u}}$. Hence, it
    allows us to conclude $\pntrl{\termsubs{x}{\cbv{u}}{\cbv{s}}}$.
    
    \item If $t = \termabs{x}{s} \in \CBVNF$, then $\cbv{t} =
    \termbang{\termabs{x}{\cbv{s}}}$. Thus, $\pnabs{\cbv{t}}$ holds and we
    conclude $\pnrml{\cbv{t}}$.
    
    \item If $t \in \CBVNF$ comes from $t \in \VarHCBVNF$. Then,
    item (\ref{l:preservation-cbv-vr}) gives
    $\pnabs{\cbv{t}}$ which implies $\pnrml{\cbv{t}}$ by definition.

    \item If $t \in \CBVNF$ comes from $t \in \HCBVNF$. Then, item
    (\ref{l:preservation-cbv-ne}) gives $\pntrl{\cbv{t}}$ which
    implies $\pnrml{\cbv{t}}$.
    
    \item If $t = \termsubs{x}{u}{s} \in \CBVNF$ comes from $s \in \CBVNF$
    and $u \in \HCBVNF$, then the \ih (\ref{l:preservation-cbv-no}) and
    (\ref{l:preservation-cbv-ne}) gives $\pnrml{\cbv{s}}$ and
    $\pntrl{\cbv{u}}$, which in turn implies $\pnbang{\cbv{u}}$. Hence, it
    allows us to conclude $\pnrml{\termsubs{x}{\cbv{u}}{\cbv{s}}}$.
  \end{itemize}
\end{enumerate}
\end{proof}

Simulation of CBN and CBV reductions in the $\BangRev$-calculus can be shown by
induction on the reduction relations. We start with a technical lemma
showing that the substitution is compatible with the CBN translation, whereas in the case of the CBV translation the argument of the substitution must be adjusted.

\begin{lemma}\mbox{}
\label{l:substitution-cbn-cbv}
Let $t, u, v \in \TermLambda$ with $v$ a value.
\begin{enumerate}
  \item\label{l:substitution-cbn} $\cbn{\substitute{x}{u}{t}} =
  \substitute{x}{\cbn{u}}{\cbn{t}}$.
  \item\label{l:substitution-cbv} $\cbv{\substitute{x}{v}{t}} =
  \substitute{x}{u}{\cbv{t}}$, where $\cbv{v} = \termbang{u}$.
\end{enumerate}
\end{lemma}

\begin{proof} \mbox{}
\begin{enumerate}
  \item By induction on $t$.
  \begin{itemize}
    \item $t = x$. Then, $\cbn{\substitute{x}{u}{x}} = \cbn{u} =
    \substitute{x}{\cbn{u}}{x} = \substitute{x}{\cbn{u}}{\cbn{x}}$.
    
    \item $t = y$. Then, $\cbn{\substitute{x}{u}{y}} = y =
    \substitute{x}{\cbn{u}}{y} = \substitute{x}{\cbn{u}}{\cbn{y}}$.
    
    \item $t = \termapp{t_0}{t_1}$. Then, \[
\begin{array}{rll}
\cbn{\substitute{x}{u}{(t_0 t_1)}}
 & =              & \cbn{(\termapp{\substitute{x}{u}{t_0}}{\substitute{x}{u}{t_1}})} \\
 & =              & \termapp{\cbn{\substitute{x}{u}{t_0}}}{\termbang{(\cbn{\substitute{x}{u}{t_1}})}} \\
 & =_{\text{\ih}} & \termapp{\substitute{x}{\cbn{u}}{\cbn{t_0}}}{\termbang{\substitute{x}{\cbn{u}}{\cbn{t_1}}}} \\
 & =              & \substitute{x}{\cbn{u}}{(\termapp{\cbn{t_0}}{\termbang{\cbn{t_1}}})} \\
 & =              & \substitute{x}{\cbn{u}}{\cbn{(\termapp{t_0}{t_1})}}
\end{array} \]
    
    \item $t = \termabs{y}{t_0}$. Straightforward by the \ih
    
    \item $t = \termsubs{y}{t_1}{t_0}$. Similar to the previous case. 
  \end{itemize}
  
  \item By induction on $t$.
  \begin{itemize}
    \item $t = x$. Then, $\cbv{\substitute{x}{v}{x}} = \cbv{v} = \termbang{u} =
    \substitute{x}{u}{(\termbang{x})} = \substitute{x}{u}{\cbv{x}}$.
    
    \item $t = y$. Then, $\cbv{\substitute{x}{v}{y}} = \cbv{y} =
    \substitute{x}{u}{\cbv{y}}$.
    
    \item $t = \termapp{t_0}{t_1}$. There are two cases:
    \begin{itemize}
      \item If $\cbv{t_0} = \ctxtapp{\ctxt{L}}{\termbang{s}}$,
      then the \ih gives $\cbv{\substitute{x}{v}{t_0}} =
      \substitute{x}{u}{\cbv{t_0}} =
      \substitute{x}{u}{\ctxtapp{\ctxt{L}}{\termbang{s}}} =
      \ctxtapp{\substitute{x}{u}{\ctxt{L}}}{\termbang{\substitute{x}{u}{s}}}$.
      Then, \[
\begin{array}{rll}
\cbv{\substitute{x}{v}{(\termapp{t_0}{t_1})}}
 & =              & \cbv{(\termapp{\substitute{x}{v}{t_0}}{\substitute{x}{v}{t_1}})} \\
 & =              & \termapp{\ctxtapp{\substitute{x}{u}{\ctxt{L}}}{\substitute{x}{u}{s}}}{\cbv{(\substitute{x}{v}{t_1})}} \\
 & =_{\text{\ih}} & \termapp{\ctxtapp{\substitute{x}{u}{\ctxt{L}}}{\substitute{x}{u}{s}}}{(\substitute{x}{u}{\cbv{t_1}})} \\
 & =              & \substitute{x}{u}{\termapp{\ctxtapp{\ctxt{L}}{s}}{\cbv{t_1}}} \\
 & =              & \substitute{x}{u}{\cbv{(\termapp{t_0}{t_1})}}
\end{array} \]
      
      \item Otherwise, $\cbv{\substitute{x}{v}{(\termapp{t_0}{t_1})}} =
      \cbv{(\termapp{\substitute{x}{v}{t_0}}{\substitute{x}{v}{t_1}})}$. The
      \ih gives $\cbv{\substitute{x}{v}{t_0}} = \substitute{x}{u}{\cbv{t_0}}$
      and $\cbv{\substitute{x}{v}{t_1}} = \substitute{x}{u}{\cbv{t_1}}$.
      Moreover, $\cbv{t_0} \neq \ctxtapp{\ctxt{L}}{\termbang{s}}$ implies the
      same for $\substitute{x}{u}{\cbv{t_0}}$, simply because $u$ is a variable
      or an abstraction. Then, \[
\begin{array}{rll}
\cbv{(\termapp{\substitute{x}{v}{t_0}}{\substitute{x}{v}{t_1}})}
 & = & \termapp{\termder{(\substitute{x}{u}{\cbv{t_0}})}}{\substitute{x}{u}{\cbv{t_1}}} \\
 & = & \substitute{x}{u}{(\termapp{\termder{(\cbv{t_0})}}{\cbv{t_1}})} \\
 & = & \substitute{x}{u}{\cbv{(\termapp{t_0}{t_1})}}
\end{array} \]
    \end{itemize}
    
    \item $t = \termabs{y}{t_0}$. Then,\[
\begin{array}{rll}
\cbv{\substitute{x}{v}{(\termabs{y}{t_0})}}
 & =              & \cbv{(\termabs{y}{\substitute{x}{v}{t_0}})} \\
 & =              & \termbang{\termabs{y}{\cbv{(\substitute{x}{v}{t_0})}}} \\
 & =_{\text{\ih}} & \termbang{\termabs{y}{(\substitute{x}{u}{\cbv{t_0}})}} \\
 & =              & \substitute{x}{u}{\termbang{(\termabs{y}{\cbv{t_0}})}} \\
 & =              & \substitute{x}{u}{\cbv{(\termabs{y}{t_0})}}
\end{array} \]
    
    \item $t = \termsubs{y}{t_1}{t_0}$. This case is straightforward by the \ih
  \end{itemize}
\end{enumerate}
\end{proof}

The following lemma shows that CBN and CBV reductions are simulated
  via the embeddings. One CBN reduction step is simulated by exactly one step  of weak reduction, whereas
  a single CBV reduction step may give rise to several steps of weak reduction
  An example is given after the proof of the lemma.

\begin{lemma}
\label{l:reduction-cbn-cbv}
Let $t, s \in \TermLambda$.
\begin{enumerate}
  \item\label{l:reduction-cbn} If $t \rewrite{\callbyname} s$, then $\cbn{t} \rewrite{\bangweak} \cbn{s}$.
  \item\label{l:reduction-cbv} If $t \rewrite{\callbyvalue} s$, then $\cbv{t}
   \rewritenp{\bangweak} \cbv{s}$.
\end{enumerate}
Moreover, exponential steps in CBN/CBV
are always simulated by exponential steps in $\BangRev$,
while multiplicative  steps in CBN/CBV
are  simulated by at least one multiplicative step in $\BangRev$.
\end{lemma}

\begin{proof}
Both proofs are by induction on the reduction relations.
\begin{enumerate}
  \item Case $t  \rewrite{\callbyname} s$.
  \begin{itemize}
    \item $t = \termapp{\ctxtapp{\ctxt{L}}{\termabs{x}{r}}}{u} \rrule{\dBeta}
    \ctxtapp{\ctxt{L}}{\termsubs{x}{u}{r}} = s$. Then, $$\cbn{t} =
    \termapp{\ctxtapp{\cbn{\ctxt{L}}}{\termabs{x}{\cbn{r}}}}{\termbang{\cbn{u}}}
    \rewrite{\dBeta}
    \ctxtapp{\cbn{\ctxt{L}}}{\termsubs{x}{\termbang{\cbn{u}}}{\cbn{r}}} =
    \cbn{s}$$
    
    \item $t = \termsubs{x}{u}{r} \rrule{\sTerm} \substitute{x}{u}{r} = s$.
    Then, $$\cbn{t} = \termsubs{x}{\termbang{\cbn{u}}}{\cbn{r}}
    \rewrite{\sBang} \substitute{x}{\cbn{u}}{\cbn{r}}
    =_{L.\ref{l:substitution-cbn-cbv}} \cbn{s}$$
    
    \item $t = \termapp{r}{u} \rewrite{\callbyname} \termapp{r'}{u} = s$, or
    $t = \termabs{x}{r} \rewrite{\callbyname} \termabs{x}{r'} = s$, or $t =
    \termsubs{x}{u}{r} \rewrite{\callbyname} \termsubs{x}{u}{r'} = s$, where
    $r \rewrite{\callbyname} r'$. Then, we easily conclude by the \ih
  \end{itemize}
  
  \item Case $t \rewrite{\callbyvalue} s$.
  \begin{itemize}
    \item $t = \termapp{\ctxtapp{\ctxt{L}}{\termabs{x}{r}}}{u} \rrule{\dBeta}
    \ctxtapp{\ctxt{L}}{\termsubs{x}{u}{r}} = s$. Then, $$\cbv{t} =
    \termapp{\ctxtapp{\cbv{\ctxt{L}}}{\termabs{x}{\cbv{r}}}}{\cbv{u}}
    \rewrite{\dBeta} \ctxtapp{\cbv{\ctxt{L}}}{\cbv{r}\exsubs{x}{\cbv{u}}} =
    \cbv{s}$$
    
    \item $t = \termsubs{x}{\ctxtapp{\ctxt{L}}{v}}{r} \rrule{\sVal}
    \ctxtapp{\ctxt{L}}{\substitute{x}{v}{r}} = s$. Then, $$\cbv{t} =
    \termsubs{x}{\ctxtapp{\cbv{\ctxt{L}}}{\cbv{v}}}{\cbv{r}} \rewrite{\sBang}
    \ctxtapp{\cbv{\ctxt{L}}}{\substitute{x}{u}{\cbv{r}}}
    =_{L.\ref{l:substitution-cbn-cbv}} \cbv{s}$$ where $\cbv{v} =
    \termbang{u}$.
    
    \item $t = \termapp{r_0}{u} \rewrite{\callbyvalue} \termapp{r_1}{u} =
    s$ comes from $r_0 \rewrite{\callbyvalue} r_1$. The \ih gives $\cbv{r_0}
    \rewritenp{\bangweak} \cbv{r_1}$. There are two cases:
    \begin{itemize}
      \item If $\cbv{r_0} = \ctxtapp{\ctxt{L}}{\termbang{r}}$, then $\cbv{t}
      = \termapp{\ctxtapp{\ctxt{L}}{r}}{\cbv{u}}$. The \ih implies in
      particular that $\cbv{r_1} = \ctxtapp{\ctxt{L}'}{\termbang{r'}}$, which
      implies in turn $\ctxtapp{\ctxt{L}}{r} \rewritenp{\bangweak}
      \ctxtapp{\ctxt{L}'}{r'}$. We then have $\cbv{t} =
      \termapp{\ctxtapp{\ctxt{L}}{r}}{\cbv{u}} \rewritenp{\bangweak}
      \termapp{\ctxtapp{\ctxt{L}'}{r'}}{\cbv{u}} = \cbv{s}$.
      
      \item Otherwise, $\cbv{t} = \termapp{\termder{(\cbv{r_0})}}{\cbv{u}}$.
      We have again two cases. 
      \begin{itemize}
        \item If $\cbv{r_1} = \ctxtapp{\ctxt{L}}{\termbang{r}}$, then $$\kern-5em
        \cbv{t} = \termapp{\termder{(\cbv{r_0})}}{\cbv{u}} \rewritenp{\bangweak} 
        \termapp{\termder{(\cbv{r_1})}}{\cbv{u}} =
        \termapp{\termder{(\ctxtapp{\ctxt{L}}{\termbang{r}})}}{\cbv{u}}
        \rewrite{\dBang} \termapp{\ctxtapp{\ctxt{L}}{r}}{\cbv{u}} = \cbv{s}$$
        
        \item Otherwise, the \ih allows us to conclude $$\cbv{t} =
        \termapp{\termder{(\cbv{r_0})}}{\cbv{u}} \rewritenp{\bangweak} 
        \termapp{\termder{(\cbv{r_1})}}{\cbv{u}} = \cbv{s}$$
      \end{itemize}
    \end{itemize}
    
    \item $t = \termapp{r}{u} \rewrite{\callbyvalue} \termapp{r}{u'} =
    s$ comes from $u \rewrite{\callbyvalue} u'$. Then, there are two cases
    according to $\cbv{r}$ but the proof is easy using the \ih
    
    \item $t = \termsubs{x}{u}{r} \rewrite{\callbyvalue}
    \termsubs{x}{u}{r'} = s$ or $t = \termsubs{x}{r}{u}
    \rewrite{\callbyvalue} \termsubs{x}{r'}{u} = s$, where $r
    \rewrite{\callbyvalue} r'$. Then, we easily conclude by the \ih
  \end{itemize}
\end{enumerate}
\end{proof}

As mentioned above, the CBV case may require several reduction steps between
$\cbv{t}$ and $\cbv{s}$.
For instance, if $t = \termapp{\termapp{\id}{y}}{z} \rewrite{\callbyvalue} \termapp{\termsubs{w}{y}{w}}{z}=s$,
then $\cbv{t} = \termapp{\termder{(\termapp{(\termabs{w}{\termbang{w}})}{\termbang{y}})}}{\termbang{z}}
\rewrite{\bangweak} \termapp{\termder{(\termsubs{w}{\termbang{y}}{\termbang{w}})}}{\termbang{z}}
\rewrite{\bangweak} \termapp{\termsubs{w}{\termbang{y}}{w}}{\termbang{z}} = \cbv{s}$.

One may as well wonder if the converse of Lemma~\ref{l:reduction-cbn-cbv} also holds. Indeed, the property holds for the CBN case, \ie $\cbn{t} \rewrite{\bangweak} \cbn{s}$ implies $t \rewrite{\callbyname} s$.
However, the  following example~\cite{Arrial-exemple}  shows that the property does not hold for CBV. Let $t =  \termapp{(\termabs{x}{\termapp{(\termabs{y}{y})}{z}})}{z}$ and $s= \termapp{(\termabs{x}{\termsubs{y}{z}{y}})}{z}$.
Then $\cbv{t}= \termapp{(\termabs{x}{\termapp{(\termabs{y}{\termbang{y}})}{\termbang{z}}})}{\termbang{z}} \rewrite{\bangweak} \termapp{(\termabs{x}{\termsubs{y}{\termbang{z}}{(\termbang{y})}})}{\termbang{z}} = \cbv{s}$,
but $t$ does not reduce to $s$ in CBV.

\begin{example}\label{ex:embeddings2}
Let us consider again the term $t=\Kterml\idl\Omegal$ of the Example~\ref{ex:embeddings1}.
We have seen that 
$\cbn{t}= \Kterm\termbang{\id}\termbang{\Omega}$
and $\cbv{t}= \termapp{\termder{(\termapp{(\termabs{x}{\termbang{\termabs{y}{\termbang x}}})}{(\termbang{\termabs{x}{\termbang x}})})}}  {\Omega}$.
The CBN reduction of $t$ is the following:
$$\Kterml\idl\Omegal\rewrite{\dBeta}\termapp{(\termsubs x{\idl}{\termabs y x})}\Omegal\rewrite{\sTerm} \termapp{(\termabs y \idl)}{\Omegal}\rewrite{\dBeta} \termsubs y{\Omegal}{\idl}\rewrite{\sTerm}\idl $$
The corresponding reduction of $\cbn{t}$ in the $\BangRev$-calculus is:
$$\Kterm\termbang{\id}\termbang{\Omega}\rewrite{\dBeta}\termapp{(\termsubs x{\termbang\id}{\termabs y x})}{\termbang\Omega}\rewrite{\sBang} \termapp{(\termabs{y}{\id})}{\termbang\Omega}
\rewrite{\dBeta} \termsubs y{\Omega}{\id}\rewrite{\sBang}\id
$$
The CBV reduction of $t$ is the following:
$$\Kterml\idl\Omegal\rewrite{\dBeta}\termapp{(\termsubs x{\idl}{\termabs y x})}\Omegal
\rewrite{\sVal}\termapp{(\termabs y \idl)}{\Omegal}\rewrite{\dBeta}$$
$$ \rewrite{\dBeta} \termsubs y{\Omegal}{\idl}
\rewrite{\dBeta} \termsubs y {\termsubs x{\Deltal}{(\termapp x x )}}{\idl}
$$
Since $\termsubs y{\termsubs x {\Deltal}{(\termapp x x )}}{\idl}\rewrite{\sVal} \termsubs y {\Omegal}{\idl}$, the reduction enters a loop.
Accordingly, the reduction of $\cbv{t}$ is:
$$
\termapp{\termder{(\termapp{(\termabs{x}{\termbang{\termabs{y}{\termbang x}}})}{(\termbang{\termabs{x}{\termbang x}})})}}{\Omega}\rewrite{\dBeta}
\termapp{\termder{(\termsubs {x} {\termbang{\termabs{x}{\termbang x}}}{(\termbang{\termabs{y}{\termbang x}})})}}{\Omega} \rewrite{\dBang}$$
$$\rewrite{\dBang} \termapp{\termsubs{x}{\termbang{\termabs{x}{\termbang x}}}{(\termabs{y}{\termbang x})}}{\Omega}
\rewrite{\dBeta} \termsubs {y}{\Omega}{\termsubs{x}{\termbang{\termabs{x}{\termbang x}}}{(\termbang x)}}\rewrite{\sBang} \termsubs{y}{\Omega}{(\termbang{\termabs{x}{\termbang x}})}
\rewrite{\dBeta}$$
$$
\rewrite{\dBeta} \termsubs{y}{\termsubs{x}{\termbang\Delta}{(\termapp{x}{\termbang{x}})}}{(\termbang{\termabs{x}{\termbang x}})}
$$
Since $ \termsubs{y}{\termsubs{x}{\termbang\Delta}{(\termapp{x}{\termbang{x}})}}{(\termbang{\termabs{x}{\termbang x}})}\rewrite{\sBang}  \termsubs{y}{\Omega}{(\termbang{\termabs{x}{\termbang x}})} $, the reduction enters a loop.
\end{example}

\paragraph{{\bf Non-Idempotent Types for Call-by-Name and Call-by-Value}}
For CBN we use the non-idempotent type system defined
in~\cite{KesnerV14} for explicit substitutions, that we present in
Figure~\ref{fig:typingSchemesNV} (top), and which is an extension of
that in~\cite{Gardner94}. For CBV,  we slightly  reformulate the
non-idempotent system in~\cite{GuerrieriM18}, tha we  present in
Figure~\ref{fig:typingSchemesNV} (bottom), in order  to recover
completeness of the (typed) CBV translation. In both cases,
  the types of our systems are the same as the ones in system  $\SysBang$.

\begin{figure}
\centering $
\begin{array}{c}
\textbf{System $\SysCBN$ for Call-by-Name}
\\[1em]
\Rule{\vphantom{\Gamma}}
     {\sequ{\assign{x}{\intertype{\sigma}{}}}{\assign{x}{\sigma}}}
     {\ruleNAxiom}
\qquad
\Rule{\sequ{\Gamma; \assign{x}{\intertype{\sigma_i}{i \in I}}}{\assign{t}{\tau}}
      \quad
      \many{\sequ{\Delta_i}{\assign{u}{\sigma_i}}}{i \in I}
     }
     {\sequ{\ctxtsum{\Gamma}{\Delta_i}{i \in I}}{\assign{\termsubs{x}{u}{t}}{\tau}}}
     {\ruleNESubs}
\\
\\
\Rule{\sequ{\Gamma}{\assign{t}{\tau}}}
     {\sequ{\ctxtres{\Gamma}{x}{}}{\assign{\termabs{x}{t}}{\functtype{\Gamma(x)}{\tau}}}}
     {\ruleNArrowI}
\qquad
\Rule{\sequ{\Gamma}{\assign{t}{\functtype{\intertype{\sigma_i}{i \in I}}{\tau}}}
      \quad
      \many{\sequ{\Delta_i}{\assign{u}{\sigma_i}}}{i \in I}
     }
     {\sequ{\ctxtsum{\Gamma}{\Delta_i}{i \in I}}{\assign{\termapp{t}{u}}{\tau}}}
     {\ruleNArrowE}
\\[2em]
\textbf{System $\SysCBV$ for Call-by-Value}
\\[1em]
\Rule{\vphantom{\Gamma}}
     {\sequ{\assign{x}{\M}}{\assign{x}{\M}}}
     {\ruleVAxiom}
\qquad
\Rule{\sequ{\Gamma}{\assign{t}{\sigma}}
      \quad
      \sequ{\Delta}{\assign{u}{\Gamma(x)}}
     }
     {\sequ{\ctxtsum{(\ctxtres{\Gamma}{x}{})}{\Delta}{}}{\assign{\termsubs{x}{u}{t}}{\sigma}}}
     {\ruleVESubs}
\\
\\
\Rule{\many{\sequ{\Gamma_i}{\assign{t}{\tau_i}}}{i \in I}}
     {\sequ{\ctxtsum{}{\ctxtres{\Gamma_i}{x}{}}{i \in I}}{\assign{\termabs{x}{t}}{\intertype{\functtype{\Gamma_i(x)}{\tau_i}}{i \in I}}}}
     {\ruleVArrowI}
\qquad
\Rule{\sequ{\Gamma}{\assign{t}{\intertype{\functtype{\M}{\tau}}{}}}
      \quad
      \sequ{\Delta}{\assign{u}{\M}}
     }
     {\sequ{\ctxtsum{\Gamma}{\Delta}{}}{\assign{\termapp{t}{u}}{\tau}}}
     {\ruleVArrowE}
\end{array} $
\caption{Typing schemes for CBN/CBV.}
\label{fig:typingSchemesNV}
\end{figure}

We write $\derivable{}{\sequN{\Gamma}{\assign{t}{\sigma}}}{\SysCBN}$
(resp.  $\derivable{}{\sequV{\Gamma}{\assign{t}{\sigma}}}{\SysCBV}$)
if there exists a type derivation in system $\SysCBN$
(resp. $\SysCBV$). We use names for type derivations as in Sec.~\ref{s:system}.
A key  point in rule $\ruleVArrowE$  is that left hand sides
of applications are typed with multisets of the form
$\intertype{\functtype{\M}{\tau}}{}$, where $\tau$ is any type, potentially a
base one, while~\cite{GuerrieriM18} necessarily requires a multiset of the form
$\intertype{\functtype{\M}{\M'}}{}$, a subtle difference which breaks
completeness. System $\SysCBN$ (resp. $\SysCBV$) can be understood as a
relational model of the Call-by-Name (resp. Call-by-Value) calculus, in the
sense that typing is stable by reduction and expansion (see
Sec.~\ref{s:technical-name} and~\ref{s:technical-value}).

The CBV translation is not complete for the system in~\cite{GuerrieriM18}, \ie
there exists a $\lambda$-term $t$ such that
$\sequB{\Gamma}{\assign{\cbv{t}}{\sigma}}$ is derivable in $\SysBang$ but
$\sequV{\Gamma}{\assign{t}{\sigma}}$ is not derivable in their system (see \cite{GuerrieriM18} Proposition 15).
In this article, we recover the completeness of the translations.  More precisely, our two embeddings are sound and
complete w.r.t. system $\SysBang$:

\begin{theorem}[Soundness/Completeness of the Embeddings]
Let ${t \in \TermLambda}$.
\begin{enumerate}
  \item \label{t:soundness-completeness-cbn}
  $\derivable{}{\sequN{\Gamma}{\assign{t}{\sigma}}}{\SysCBN}$ iff
  $\derivable{}{\sequB{\Gamma}{\assign{\cbn{t}}{\sigma}}}{\SysBang}$.

  \item \label{t:soundness-completeness-cbv}
  $\derivable{}{\sequV{\Gamma}{\assign{t}{\sigma}}}{\SysCBV}$ iff
  $\derivable{}{\sequB{\Gamma}{\assign{\cbv{t}}{\sigma}}}{\SysBang}$.
\end{enumerate}
\label{t:soundness-completeness-cbn-cbv}
\end{theorem}

\begin{proof}
\begin{enumerate}
  \item $\Rightarrow$) By induction on $t$.
  \begin{itemize}
    \item $t = x$. Then, $\Gamma = \assign{x}{\intertype{\sigma}{}}$
    and we conclude by rule $\ruleBAxiom$, since $\cbn{t} = x$.
    
    \item $t = \termapp{s}{u}$. Then, $\Gamma = \ctxtsum{\Gamma'}{\Delta_i}{i
    \in I}$, $\sequN{\Gamma'}{\assign{s}{\functtype{\intertype{\tau_i}{i \in
    I}}{\sigma}}}$ and $\many{\sequN{\Delta_i}{\assign{u}{\tau_i}}}{i \in I}$.
    By \ih $\sequB{\Gamma'}{\assign{\cbn{s}}{\functtype{\intertype{\tau_i}{i
    \in I}}{\sigma}}}$ and $\many{\sequB{\Delta_i}{\assign{\cbn{u}}{\tau_i}}}{i
    \in I}$. Moreover, by definition $\cbn{t} =
    \termapp{\cbn{s}}{\termbang{\cbn{u}}}$. Thus, we conclude by rules
    $\ruleBBang$ and $\ruleBArrowE$ \[
    \Rule{
      \sequB{\Gamma'}{\assign{\cbn{s}}{\functtype{\intertype{\tau_i}{i \in I}}{\sigma}}}
      \quad
      \Rule{
        \many{\sequB{\Delta_i}{\assign{\cbn{u}}{\tau_i}}}{i \in I}
      }{
        \sequ{\ctxtsum{}{\Delta_i}{i \in I}}{\assign{\termbang{\cbn{u}}}{\intertype{\tau_i}{i \in I}}}
      }{}
    }{
      \sequ{\ctxtsum{\Gamma'}{\Delta_i}{i \in I}}{\assign{\termapp{\cbn{s}}{\termbang{\cbn{u}}}}{\sigma}}
    }{} \]
    
    \item $t = \termabs{x}{s}$. Then, $\sigma = \functtype{\Gamma'(x)}{\tau}$,
    $\Gamma = \ctxtres{\Gamma'}{x}{}$ and $\sequN{\Gamma'}{\assign{s}{\tau}}$.
    By \ih $\sequB{\Gamma'}{\assign{\cbn{s}}{\tau}}$. Moreover, by definition
    $\cbn{t} = \termabs{x}{\cbn{s}}$, hence we conclude by rule $\ruleBArrowI$.
    
    \item $t = \termsubs{x}{u}{s}$. Then, $\Gamma =
    \ctxtsum{(\ctxtres{\Gamma'}{x}{})}{\Delta_i}{i \in I}$,
    $\sequN{\Gamma'}{\assign{s}{\sigma}}$ and
    $\many{\sequN{\Delta_i}{\assign{u}{\tau_i}}}{i \in I}$ with $\Gamma'(x) =
    \intertype{\tau_i}{i \in I}$. By \ih
    $\sequB{\Gamma'}{\assign{\cbn{s}}{\sigma}}$ and
    $\many{\sequB{\Delta_i}{\assign{\cbn{u}}{\tau_i}}}{i \in I}$ hold. Finally,
    since $\cbn{t} = \termsubs{x}{\termbang{\cbn{u}}}{\cbn{s}}$, we conclude by
    rules $\ruleBBang$ and $\ruleBESubs$.
  \end{itemize}
  
  $\Leftarrow$) By induction on $t$.
  \begin{itemize}
    \item $t = x$. Then, $\cbn{x} = x$, $\Gamma =
    \assign{x}{\intertype{\sigma}{}}$ and the result is immediate using rule
    $\ruleNAxiom$.
    
    \item $t = \termapp{s}{u}$. Then, $\cbn{t} =
    \termapp{\cbn{s}}{\termbang{\cbn{u}}}$, $\Gamma =
    \ctxtsum{\Gamma'}{\Delta}{}$,
    $\sequB{\Gamma'}{\assign{\cbn{s}}{\functtype{\M}{\sigma}}}$ and
    $\sequB{\Delta}{\assign{\termbang{\cbn{u}}}{\M}}$. Moreover, by rule
    $\ruleBBang$, $\Delta = \ctxtsum{}{\Delta_i}{i \in I}$, $\M =
    \intertype{\tau_i}{i \in I}$ and
    $\many{\sequB{\Delta_i}{\assign{\cbn{u}}{\tau_i}}}{i \in I}$. The \ih 
    gives $\sequN{\Gamma'}{\assign{s}{\functtype{\intertype{\tau_i}{i \in
    I}}{\sigma}}}$ and $\many{\sequN{\Delta_i}{\assign{u}{\tau_i}}}{i \in I}$.
    Finally, we conclude by rule $\ruleNArrowE$.
    
    \item $t = \termabs{x}{s}$. Then $\cbn{t} = \termabs{x}{\cbn{s}}$, $\sigma
    = \functtype{\Gamma'(x)}{\tau}$, $\Gamma = \ctxtres{\Gamma'}{x}{}$ and
    $\sequB{\Gamma'}{\assign{\cbn{s}}{\tau}}$. The result follow immediately
    from the \ih using rule $\ruleNArrowI$.
    
    \item $t = \termsubs{x}{u}{s}$. Then, $\cbn{t} =
    \termsubs{x}{\termbang{\cbn{u}}}{\cbn{s}}$, $\Gamma =
    \ctxtsum{(\ctxtres{\Gamma'}{x}{})}{\Delta}{}$,
    $\sequB{\Gamma'}{\assign{\cbn{s}}{\sigma}}$ and
    $\sequB{\Delta}{\assign{\termbang{\cbn{u}}}{\Gamma'(x)}}$. Moreover, by
    rule $\ruleBBang$, $\Delta = \ctxtsum{}{\Delta_i}{i \in I}$, $\Gamma'(x) =
    \intertype{\tau_i}{i \in I}$ and
    $\many{\sequB{\Delta_i}{\assign{\cbn{u}}{\tau_i}}}{i \in I}$. The \ih 
    gives $\sequN{\Gamma'}{\assign{s}{\sigma}}$ and $\many{\sequN{\Delta_i}{\assign{u}{\tau_i}}}{i \in I}$.
    Finally, we conclude by rule $\ruleNESubs$.
  \end{itemize}
  
  \item $\Rightarrow$) By induction on $t$.
  \begin{itemize}
    \item $t = x$. Then, $\sigma = \M$ and $\Gamma = \assign{x}{\M}$ where, by
    definition, $\M = \intertype{\tau_i}{i \in I}$. Moreover, by definition of
    the embedding, $\cbv{t} = \termbang{x}$. Thus, we conclude by rules
    $\ruleBAxiom$ and $\ruleBBang$ \[
    \Rule{
      \raisebox{0.5em}{$\Big(
        \raisebox{-0.25em}{\hspace{-1em}
          \Rule{}{\sequ{\assign{x}{\intertype{\tau_i}{}}}{\assign{x}{\tau_i}}}{}
        \hspace{-1em}}
      \Big)_{i \in I}$}
    }{
      \sequ{\assign{x}{\intertype{\tau_i}{i \in I}}}{\assign{\termbang{x}}{\intertype{\tau_i}{i \in I}}}
    }{} \]
    
    \item $t = \termapp{s}{u}$. Then, $\Gamma = \ctxtsum{\Gamma'}{\Delta}{}$, 
    $\sequV{\Gamma'}{\assign{s}{\intertype{\functtype{\M}{\sigma}}{}}}$ and
    $\sequV{\Delta}{\assign{u}{\M}}$. By \ih
    $\sequB{\Gamma'}{\assign{\cbv{s}}{\intertype{\functtype{\M}{\sigma}}{}}}$
    and $\sequB{\Delta}{\assign{\cbv{u}}{\M}}$. There are two cases:
    \begin{itemize}
      \item If $\cbv{s} = \ctxtapp{\ctxt{L}}{\termbang{r}}$, then $\cbv{t}
      = \termapp{\ctxtapp{\ctxt{L}}{r}}{\cbv{u}}$. From
      $\sequB{\Gamma'}{\assign{\ctxtapp{\ctxt{L}}{\termbang{r}}}{\intertype{\functtype{\M}{\sigma}}{}}}$
      it is immediate to see,
      by induction on $\ctxt{L}$ using $\ruleBESubs$, that there exists $\Gamma''$ s.t.
      $\sequB{\Gamma''}{\assign{\termbang{r}}{\intertype{\functtype{\M}{\sigma}}{}}}$. Then,
      $\sequB{\Gamma''}{\assign{r}{\functtype{\M}{\sigma}}}$ holds from rule
      $\ruleBBang$. Moreover, by straighforward induction on $\ctxt{L}$ once
      again, we get
      $\sequB{\Gamma'}{\assign{\ctxtapp{\ctxt{L}}{r}}{\functtype{\M}{\sigma}}}$.
      Finally, we conclude using  $\ruleBArrowE$ by using also
      $\sequB{\Delta}{\assign{\cbv{u}}{\M}}$.

      \item Otherwise, $\cbv{t} = \termapp{\termder{(\cbv{s})}}{\cbv{u}}$.
      Then, we resort to rules $\ruleBDer$ and $\ruleBArrowE$ to conclude \[
    \Rule{
      \Rule{
        \sequB{\Gamma'}{\assign{\cbv{s}}{\intertype{\functtype{\M}{\sigma}}{}}}
      }{
        \sequ{\Gamma'}{\assign{\termder{\cbv{s}}}{\functtype{\M}{\sigma}}}
      }{}
      \quad
      \sequB{\Delta}{\assign{\cbv{u}}{\M}}
    }{
      \sequ{\ctxtsum{\Gamma'}{\Delta}{}}{\assign{\termapp{\termder{(\cbv{s})}}{\cbv{u}}}{\sigma}}
    }{} \]
    \end{itemize}

    \item $t = \termabs{x}{s}$. Then $\sigma =
    \intertype{\functtype{\Gamma_i(x)}{\tau_i}}{i \in I}$, $\Gamma =
    \ctxtsum{}{\ctxtres{\Gamma_i}{x}{}}{i \in I}$ and
    $\many{\sequV{\Gamma_i}{\assign{s}{\tau_i}}}{i \in I}$. By \ih
    $\many{\sequB{\Gamma_i}{\assign{\cbv{s}}{\tau_i}}}{i \in I}$. Thus, we
    conclude with rules $\ruleBArrowI$ and $\ruleBBang$ since, by definition,
    $\cbv{t} = \termbang{\termabs{x}{\cbv{s}}}$ \[
    \Rule{
      \raisebox{1em}{$\Bigg(
        \raisebox{-0.75em}{\hspace{-1em}
          \Rule{
            \sequB{\Gamma_i}{\assign{\cbv{s}}{\tau_i}}
          }{
            \sequ{\ctxtres{\Gamma_i}{x}{}}{\assign{\termabs{x}{\cbv{s}}}{\functtype{\Gamma_i(x)}{\tau_i}}}
          }{}
        \hspace{-1em}}
      \Bigg)_{i \in I}$}
    }{
      \sequ{\ctxtsum{}{\ctxtres{\Gamma_i}{x}{}}{i \in I}}{\assign{\termbang{\termabs{x}{\cbv{s}}}}{\intertype{\functtype{\Gamma_i(x)}{\tau_i}}{i \in I}}}
    }{} \]
    
    \item $t = \termsubs{x}{u}{s}$. Then, $\Gamma =
    \ctxtsum{(\ctxtres{\Gamma'}{x}{})}{\Delta}{}$,
    $\sequV{\Gamma'}{\assign{s}{\sigma}}$ and
    $\sequV{\Delta}{\assign{u}{\Gamma'(x)}}$. By \ih
    $\sequB{\Gamma'}{\assign{\cbv{s}}{\sigma}}$ and
    $\sequB{\Delta}{\assign{\cbv{u}}{\Gamma'(x)}}$ hold. Hence, since $\cbv{t} =
    \termsubs{x}{\cbv{u}}{\cbv{s}}$, we conclude by rule $\ruleBESubs$.
  \end{itemize}
  
  $\Leftarrow$) By induction on $t$.
  \begin{itemize}
    \item $t = x$. Then $\cbv{x} = \termbang{x}$, $\sigma =
    \intertype{\tau_i}{i \in I}$ and $\Gamma = \assign{x}{\intertype{\tau_i}{i
    \in I}}$. The result is immediate using rule $\ruleVAxiom$.
    
    \item $t = \termapp{s}{u}$. There are two cases to consider:
    \begin{itemize}
      \item If $\cbv{s} = \ctxtapp{\ctxt{L}}{\termbang{r}}$, then $\cbv{t} =
      \termapp{\ctxtapp{\ctxt{L}}{r}}{\cbv{u}}$. Then, $\Gamma =
      \ctxtsum{\Gamma'}{\Delta}{}$,
      $\sequB{\Gamma'}{\assign{\ctxtapp{\ctxt{L}}{r}}{\functtype{\M}{\sigma}}}$
      and $\sequB{\Delta}{\assign{\cbv{u}}{\M}}$. Moreover, by induction on
      $\ctxt{L}$ using $\ruleBESubs$, it is immediate to see that there exists
      $\Gamma''$ s.t. $\sequB{\Gamma''}{\assign{r}{\functtype{\M}{\sigma}}}$. By
      rule $\ruleBBang$ we get
      $\sequB{\Gamma''}{\assign{\termbang{r}}{\intertype{\functtype{\M}{\sigma}}{}}}$
      and, by straightforward induction on $\ctxt{L}$ once again,
      $\sequB{\Gamma'}{\assign{\ctxtapp{\ctxt{L}}{\termbang{r}}}{\intertype{\functtype{\M}{\sigma}}{}}}$.
      Recall that the induction is on $t = \termapp{s}{u}$ and $\cbv{s}
      = \ctxtapp{\ctxt{L}}{\termbang{r}}$, hence by the \ih
      $\sequV{\Gamma'}{\assign{s}{\intertype{\functtype{\M}{\sigma}}{}}}$ and
      $\sequV{\Delta}{\assign{u}{\M}}$. Thus, we conclude by rule
      $\ruleVArrowE$.

      \item Otherwise, $\cbv{t} =
      \termapp{\termder{(\cbv{s})}}{\cbv{u}}$. We then have $\Gamma =
      \ctxtsum{\Gamma'}{\Delta}{}$,
      $\sequB{\Gamma'}{\assign{\termder{(\cbv{s})}}{\functtype{\M}{\sigma}}}$ and
      $\sequB{\Delta}{\assign{\cbv{u}}{\M}}$. Moreover, from rule $\ruleBDer$,
      we have
      $\sequB{\Gamma'}{\assign{\cbv{s}}{\intertype{\functtype{\M}{\sigma}}{}}}$.
      Then, by \ih
      $\sequV{\Gamma'}{\assign{s}{\intertype{\functtype{\M}{\sigma}}{}}}$
      and $\sequV{\Delta}{\assign{u}{\M}}$. Finally, we conclude by rule
      $\ruleVArrowE$.
    \end{itemize}
    
    \item $t = \termabs{x}{s}$. Then, $\cbv{t} =
    \termbang{\termabs{x}{\cbv{s}}}$, $\sigma = \intertype{\sigma_i}{i \in I}$,
    $\Gamma = \ctxtsum{}{\Gamma_i}{i \in I}$ and
    $\many{\sequB{\Gamma_i}{\assign{\termabs{x}{\cbv{s}}}{\sigma_i}}}{i \in I}$.
    Moreover, by rule $\ruleBArrowI$, for every $i \in I$ we have $\sigma_i =
    \functtype{\Gamma'_i(x)}{\tau_i}$, $\Gamma_i = \ctxtres{\Gamma'_i}{x}{}$
    and $\sequB{\Gamma'_i}{\assign{\cbv{s}}{\tau_i}}$. By \ih
    $\many{\sequV{\Gamma'_i}{\assign{s}{\tau_i}}}{i \in I}$ and we conclude by
    rule $\ruleVArrowI$.
    
    \item $t = \termsubs{x}{u}{s}$. Then, $\cbv{t} =
    \termsubs{x}{\cbv{u}}{\cbv{s}}$, $\Gamma =
    \ctxtsum{(\ctxtres{\Gamma'}{x}{})}{\Delta}{}$,
    $\sequB{\Gamma'}{\assign{\cbv{s}}{\sigma}}$ and
    $\sequB{\Delta}{\assign{\cbv{u}}{\Gamma'(x)}}$. The result follow
    immediately from the \ih using rule $\ruleVESubs$.
  \end{itemize}
\end{enumerate}
\end{proof}

To illustrate our previous theorem, we take the term
  which is the counter-example
  to completeness in~\cite{GuerrieriM18} Proposition 15.
  Indeed, we show how $t = \termabs{x}{\termapp{x}{x}}$, and $\cbv{t}=\termbang{\termabs{x}{\termapp{x}{\termbang{x}}}}$ can now be  typed with the same type context and type
  in the systems $\SysCBV$ and $\SysBang$ respectively.
  Let $\sigma = \functtype{\emul}{\emul}$. \[
\Rule{
  \Rule{
    \Rule{}
         {\sequV{x:\intertype{\functtype{\intertype{\sigma}{}}{\sigma}}{}}{x:\intertype{\functtype{\intertype{\sigma}{}}{\sigma}}{}}}
         {\ruleVAxiom}
    \quad
    \Rule{}
         {\sequV{x:\intertype{\sigma}{}}{x:\intertype{\sigma}{}}}
         {\ruleVAxiom}
  }
  {\sequV{x:\intertype{\functtype{\intertype{\sigma}{}}{\sigma}, \sigma}{}}{\termapp{x}{x}}:\sigma }
  {\ruleVArrowE}
}
{\sequV{}{\termabs{x}{\termapp{x}{x}}}:\intertype{\functtype{\intertype{\functtype{\intertype{\sigma}{}}{\sigma}, \sigma}{}}{\sigma}}{} }
{\ruleVArrowI} \] \[
\Rule{
  \Rule{
    \Rule{
      \Rule{}
           {\sequB{x: \intertype{\functtype{\intertype{\sigma}{}}{\sigma}}{}}{x: \functtype{\intertype{\sigma}{}}{\sigma}}}
           {\ruleBAxiom}
      \quad
      \Rule{
        \Rule{}
             {\sequB{x:\intertype{\sigma}{}}{x:\sigma}}
             {\ruleBAxiom}
      }
      {\sequB{x:\intertype{\sigma}{}}{\termbang{x}}:\intertype{\sigma}{}}
      {\ruleBBang}
    }
    {\sequB{x:\intertype{\functtype{\intertype{\sigma}{}}{\sigma}, \sigma}{}}{\termapp{x}{\termbang{x}}}:\sigma }
    {\ruleBArrowE}
  }
  {\sequB{}{\termabs{x}{\termapp{x}{\termbang{x}}}}:\functtype{\intertype{\functtype{\intertype{\sigma}{}}{\sigma}, \sigma}{}}{\sigma} }
  {\ruleBArrowI} 
}
{\sequB{}{\termbang{\termabs{x}{\termapp{x}{\termbang{x}}}}}:\intertype{\functtype{\intertype{\functtype{\intertype{\sigma}{}}{\sigma}, \sigma}{}}{\sigma}}{} }
{\ruleBBang} \]

The type system $\SysCBN$ (resp.  $\SysCBV$) characterises
$\callbyname$-normalisation (resp. $\callbyvalue$-normalisation). More
precisely:

\begin{theorem}[{\bf Characterisation of CBN/CBV Normalisation}]
Let ${t \in \TermLambda}$.
\begin{itemize}
  \item\label{t:normalisation-cbn} $t$ is $\SysCBN$-typable iff $t$ is
  $\callbyname$-normalising.
  \item\label{t:normalisation-cbv} $t$ is $\SysCBV$-typable iff $t$ is
  $\callbyvalue$-normalising.
\end{itemize}
\label{t:normalisation-cbn-cbv}
\end{theorem}

\begin{proof}
Sec.~\ref{s:technical-name} and~\ref{s:technical-value} below show the two
statements separately as Thm.~\ref{t:normalisation-name} and
Thm.~\ref{t:normalisation-value} respectively.
\end{proof}

\subsection{Call-by-Name}
\label{s:technical-name}
The following lemmas aim to prove that $\SysCBN$ is indeed a model for the
Call-by-Name reduction strategy presented at the beginning of
Sec.~\ref{s:cbname-cbvalue}. Similar results have been already presented in the
literature for different formulations of call-by-name languages, with or
without explicit substitutions, some examples
are~\cite{BucciarelliKV17,KesnerV14}. In this subsection we measure
$\SysCBN$-derivations by using a function $\sizen{\_}$ which counts $1$
for all typing rules. We follow here the same pattern used
in Sec.~\ref{s:system} for system $\SysBang$: a substitution (resp.
anti-substitution) lemma is established and used in the proof of subject
reduction (resp. expansion). In the following subsection the same methodology
is used for $\SysCBV$.

\begin{lemma}[Substitution]
\label{l:substitution-name}
Let
$\derivable{\Phi_{t}}{\sequN{\Gamma;\assign{x}{\intertype{\sigma_i}{i \in I}}}{\assign{t}{\tau}}}{\SysCBN}$
and
$\many{\derivable{\Phi^{i}_{u}}{\sequN{\Delta_i}{\assign{u}{\sigma_i}}}{\SysCBN}}{i \in I}$.
Then, there exists
$\derivable{\Phi_{\substitute{x}{u}{t}}}{\sequN{\ctxtsum{\Gamma}{\Delta_i}{i \in I}}{\assign{\substitute{x}{u}{t}}{\tau}}}{\SysCBN}$
such that
$\sizen{\Phi_{\substitute{x}{u}{t}}} = \sizen{\Phi_{t}} +_{i \in I} \sizen{\Phi^{i}_{u}} - |I|$.
\end{lemma}

\begin{proof}
By induction on $\Phi_t$, reasoning exactly as in Lemma~\ref{l:substitution-bang}.
\end{proof}

\begin{lemma}[Weighted Subject Reduction]
\label{l:sr-name}
Let $\derivable{\Phi}{\sequN{\Gamma}{\assign{t}{\tau}}}{\SysCBN}$ and $t
\rewrite{\callbyname} t'$. Then, there exists
$\derivable{\Phi'}{\sequN{\Gamma}{\assign{t'}{\tau}}}{\SysCBN}$ such that 
$\sizen{\Phi} > \sizen{\Phi'}$.
\end{lemma}

\begin{proof}
By induction on $t \rewrite{\callbyname} t'$, where, in
particular, Lemma~\ref{l:substitution-name} is used in the base case
$\termsubs{x}{u}{t} \rrule{\sTerm} \substitute{x}{u}{t}$.
\end{proof}

\begin{lemma}[Anti-Substitution]
\label{l:antisubstitution-name}
If
$\derivable{\Phi_{\substitute{x}{u}{t}}}{\sequB{\Gamma'}{\assign{\substitute{x}{u}{t}}{\tau}}}{\SysCBN}$,
then there exist
$\derivable{\Phi_t}{\sequB{\Gamma; \assign{x}{\intertype{\sigma_i}{i \in I}}}{\assign{t}{\tau}}}{\SysCBN}$
and
$\many{\derivable{\Phi^i_u}{\sequB{\Delta_i}{\assign{u}{\sigma_i}}}{\SysCBN}}{i \in I}$
such that $\Gamma' = \ctxtsum{\Gamma}{\Delta_i}{i \in I}$ and
$\sizen{\Phi_{\substitute{x}{u}{t}}} = \sizen{\Phi_{t}}  +_{i \in I}{\sizen{\Phi^i_{u}}} - |I|$.
\end{lemma}

\begin{proof}
By induction on $t$, reasoning exactly as in Lemma~\ref{l:antisubstitution-bang}.
\end{proof}

\begin{lemma}[Weighted Subject Expansion]
\label{l:se-name}
Let $\derivable{\Phi'}{\sequN{\Gamma}{\assign{t'}{\tau}}}{\SysCBN}$ and $t
\rewrite{\callbyname} t'$. Then, there exists
$\derivable{\Phi}{\sequN{\Gamma}{\assign{t}{\tau}}}{\SysCBN}$ such that
$\sizen{\Phi} > \sizen{\Phi'}$.
\end{lemma}

\begin{proof}
By induction on $t \rewrite{\callbyname} t'$, where, in
particular, Lemma~\ref{l:antisubstitution-name} is used in the base case
$\termsubs{x}{u}{t} \rrule{\sTerm} \substitute{x}{u}{t}$.
\end{proof}

The following lemma, showing that $\callbyname$-normal forms are typable, is used in
the proof of Theorem~\ref{t:normalisation-name} as the base case for
establishing that a $\callbyname$-normalising term is typable.
\begin{lemma}
\label{l:normal-forms-name}
Let $t \not\rewrite{\callbyname}$. Then, $t$ is $\SysCBN$-typable.
\end{lemma}

\begin{proof}
  By straightforward induction on $\callbyname$-normal forms,
  using Lemma~\ref{l:normal-forms-cbn-cbv}.
\end{proof}

We now relate $\SysCBN$-typability with $\callbyname$-normalisation.  We also
put in evidence the quantitative aspect of system $\SysCBN$, so that we
introduce the \emphdef{$\mathtt n$-size} of terms as: $\nsize{x} \eqdef 0$,
$\nsize{\termabs{x}{t}} \eqdef 1 + \nsize{t}$, $\nsize{\termapp{t}{u}}
\eqdef 1 + \nsize{t}$, and $\nsize{\termsubs{x}{u}{t}} \eqdef
1 + \nsize{t}$. Note in particular that given a
derivation $\Phi_t$ for a term $t$ we always have $\sizen{\Phi_t} \geq
\nsize{t}$.

\begin{theorem}[Soundness and Completeness for System $\SysCBN$]
\label{t:normalisation-name}
Let ${t \in \TermLambda}$. Then, $t$ is $\SysCBN$-typable iff $t$ is
$\callbyname$-normalising. Moreover, if
$\derivable{\Phi}{\sequN{\Gamma}{\assign{t}{\tau}}}{\SysCBN}$, then
there exists $p \in \CBNNF$ such that $t
\rewriten{\callbyname}^{(\cbeta,\cexp)} p$ and $\sizen{\Phi} \geq \cbeta +
\cexp + \nsize{p}$.
\end{theorem}

\begin{proof}
The $\Rightarrow$ direction holds by WSR (Lemma~\ref{l:sr-name}), while the
$\Leftarrow$ direction follows from Lemma~\ref{l:normal-forms-name} and WSE
(Lemma~\ref{l:se-name}). The \emph{moreover} statement holds by
Lemma~\ref{l:sr-name} and the fact that the size of
the type derivation of $p$ is greater than or equal to $\nsize{p}$.
\end{proof}

\subsection{Call-by-Value}
\label{s:technical-value}

The following lemmas aim to prove that $\SysCBV$ is indeed a model for the
Call-by-Value reduction strategy. In this subsection we measure
$\SysCBV$-derivations by using a function $\sizev{\_}$. Formally,
  \begin{definition} Size of derivations is defined by induction as follows: 
  \begin{itemize}
  \setlength\itemsep{1em}
  \item $\sizev{\Rule{\vphantom{\Gamma}}
                     {\sequ{\assign{x}{\M}}{\assign{x}{\M}}}
                     {\ruleVAxiom}} = |\M|$,
  \item $\sizev{\Rule{\Pi_1 \quad \Pi_2}
                     {\sequ{\Gam}{\assign{\termsubs{x}{u}{t}}{\sigma}}}
                     {\ruleVESubs}} = 1+ \sizev{\Pi_1} + \sizev{\Pi_2}$,
  \item $\sizev{\Rule{\Pi_1 \quad \Pi_2}
                     {\sequ{\Gam}{\assign{\termapp{t}{u}}{\tau}}}
                     {\ruleVArrowE}} = 1 + \sizev{\Pi_1} + \sizev{\Pi_2}$,
  \item $\sizev{\Rule{\many{\Pi_i}{i \in I}}
                     {\sequ{\ctxtsum{}{\ctxtres{\Gamma_i}{x}{}}{i \in I}}{\assign{\termabs{x}{t}}{\intertype{\functtype{\Gamma_i(x)}{\tau_i}}{i \in I}}}}
                     {\ruleVArrowI}} = |I| +_{\iI} \sizev{\Pi_i}$.
\end{itemize}
\end{definition}
Thus, $\sizev{\Pi}$ counts $1$ for rule
$\ruleVArrowE$, while rule $\ruleVESubs$ does not count, the axiom $\ruleVAxiom$
$\sequV{\assign{x}{\M}}{\assign{x}{\M}}$ counts $|\M|$, and
$\ruleVArrowI$ contributes with its number of premises.

In order to prove a substitution lemma for CBV, we need an additional lemma allowing the decomposition of the multiset types of values.
\begin{lemma}[Split type for value]
\label{l:split-value}
Let $\derivable{\Phi_{v}}{\sequV{\Gamma}{\assign{v}{\M}}}{\SysCBV}$ such that $\M =
\ctxtsum{}{\M_i}{i \in I}$. Then, there exist
$\many{\derivable{\Phi_{v}^{i}}{\sequV{\Gamma_i}{\assign{v}{\M_i}}}{\SysCBV}}{i \in I}$
such that $\Gamma = \ctxtsum{}{\Gamma_i}{i \in I}$
and $\sizev{\Phi_v} = +_{i \in I} \sizev{\Phi_{v}^{i}}$.
\end{lemma}

\begin{proof}
By case analysis on the shape  of $v$.
\begin{itemize}
  \item $v = x$. By $\ruleVAxiom$,
  $\derivable{\Phi_v}{\sequV{\assign{x}{\M}}{\assign{x}{\M}}}{\SysCBV}$. From $\M =
  \ctxtsum{\!}{\M_i}{i \in I}$ we get
  $\many{\derivable{\Phi^i_{v}}{\sequV{\assign{x}{\M_i}}{\assign{x}{\M_i}}}{\SysCBV}}{i \in I}$
  by $\ruleVAxiom$ several times. We then conclude, since $\sizev{\Phi_v} = |\M| = 
  +_{i \in I} |\M_i| = +_{i \in I} \sizev{\Phi^i_v}$.
  
  \item $v = \termabs{x}{s}$. By $\ruleVArrowI$, we have $\Gamma =
  \ctxtsum{\!}{\ctxtres{\Gamma_j}{x}{}}{j \in J}$, $\M =
  \intertype{\functtype{\Gamma_j(x)}{\tau_j}}{j \in J}$ and
  $\many{\derivable{\Phi^j_{s}}{\sequV{\Gamma_j}{\assign{s}{\tau_j}}}{\SysCBV}}{j \in J}$
  for some $J$. Since 
  $\M=\ctxtsum{\!}{\M_i}{i \in I}$, we have also $J = \ctxtsum{\!}{J_i}{i \in I}$,  with  $\M_i =
  \intertype{\functtype{\Gamma_j(x)}{\tau_j}}{j \in J_i}$ for each $i \in I$.
  We  construct the type derivations
  $\many{\derivable{\Phi^i_v}{\sequV{\ctxtsum{\!}{\ctxtres{\Gamma_j}{x}{}}{j \in J_i}}{\assign{\termabs{x}{s}}{\M_i}}}{\SysCBV}}{i \in I}$, using for each $i\in I$ the rule $\ruleVArrowI$ with premises
  $\many{\derivable{\Phi^j_{s}}{\sequV{\Gamma_j}{\assign{s}{\tau_j}}}{\SysCBV}}{j \in J_i}$.
  We have that  $\Gamma =
  \ctxtsum{\!}{\ctxtres{\Gamma_j}{x}{}}{j \in J} =
  \ctxtsum{\!}{(\ctxtsum{\!}{\ctxtres{\Gamma_j}{x}{}}{j \in J_i})}{i \in I}$.
  We conclude since $\sizev{\Phi_v} = |J| +_{j \in J} \sizev{\Phi^j_s} = 
  +_{i \in I} (|J_i| +_{j \in J_i} \sizev{\Phi^j_v}) = +_{i \in I} \sizev{\Phi^i_v}$.
\end{itemize}
\end{proof}


\begin{lemma}[Substitution]
\label{l:substitution-value}
Let $\derivable{\Phi_{t}}{\sequV{\Gamma}{\assign{t}{\tau}}}{\SysCBV}$
and
$\derivable{\Phi_{v}}{\sequV{\Delta}{\assign{v}{\Gam(x)}}}{\SysCBV}$. Then,
there exists
$\derivable{\Phi_{\substitute{x}{v}{t}}}{\sequV{\ctxtsum{\ctxtres{\Gamma}{x}{}}{\Delta}{}}{\assign{\substitute{x}{v}{t}}{\tau}}}{\SysCBV}$
such that $\sizev{\Phi_{\substitute{x}{v}{t}}} =
\sizev{\Phi_{t}} + \sizev{\Phi_{v}} - |\Gam(x)|$.
\end{lemma}

\begin{proof}
By induction on $\Phi_t$. If $\Phi_t$ is $\ruleVAxiom$ and $t = x$, then
$\substitute{x}{v}{t} = v$ and $\Phi_t$ is of the form
$\sequ{\assign{x}{\Gam(x)}}{\assign{x}{\Gam(x)}}$. We let
$\Phi_{\substitute{x}{v}{t}} = \Phi_v$. We conclude since $\sizev{\Phi_t} =
|\Gam(x)|$. If $\Phi_t$ is $\ruleVAxiom$ and $t = y \neq x$, then
$\substitute{x}{v}{t} = y$, $\Gam(x) = \emul$ (so that $|\Gam(x)| = 0$) and
$\Phi_y$ is of the form $\sequ{\assign{y}{\M}}{\assign{y}{\M}}$. We let
$\Phi_{\substitute{x}{v}{t}} = \Phi_y$. Moreover,
$\derivable{\Phi_{v}}{\sequV{\Delta}{\assign{v}{\emul}}}{\SysCBV}$ implies
$\sizev{\Phi_v} = 0$. Then, we immediately conclude.

If $\Phi_t$ ends with $\ruleVArrowE$, then $t=t_1t_2$, $\Gamma=\Gamma_1+\Gamma_2$
  and there exist a type $\M$ and two derivations $\derivable{\Phi_{t_1}}{\sequV{\Gamma_1}{\assign{t_1}{[\M\rightarrow \tau]}}}{\SysCBV}$ and
  $\derivable{\Phi_{t_2}}{\sequV{\Gamma_2}{\assign{t_2}{\M}}}{\SysCBV}$.
  Since $\derivable{\Phi_{v}}{\sequV{\Delta}{\assign{v}{\Gam(x)}}}{\SysCBV}$ and $\Gam(x)=\Gam_1(x)+\Gam_2(x)$, Lemma~\ref{l:split-value} ensures that there exist
  two derivations  $\derivable{\Phi^1_{v}}{\sequV{\Delta_1}{\assign{v}{\Gam_1(x)}}}{\SysCBV}$ and $\derivable{\Phi^2_{v}}{\sequV{\Delta_2}{\assign{v}{\Gam_2(x)}}}{\SysCBV}$ such that 
  $\Delta=\Delta_1+\Delta_2$ and $\size{\Phi_v}=\size{\Phi^1_v}+\size{\Phi^2_v}$. Using the \ih on $\Phi_{t_i},\Phi^i_{v}$, $i=1,2$, we get two derivations
  $\derivable{\Phi_{\substitute{x}{v}{t_1}}}{\sequV{\ctxtsum{\ctxtres{\Gamma_1}{x}{}}{\Delta_1}{}}{\assign{\substitute{x}{v}{t_1}}{[\M\rightarrow \tau]}}}{\SysCBV}$ and
  $\derivable{\Phi_{\substitute{x}{v}{t_2}}}{\sequV{\ctxtsum{\ctxtres{\Gamma_2}{x}{}}{\Delta_2}{}}{\assign{\substitute{x}{v}{t_2}}{\M}}}{\SysCBV}$ such that
  $\size{\Phi_{\substitute{x}{v}{t_i}}} =\sizev{\Phi_{t_i}} + \sizev{\Phi^i_{v}} - |\Gam_i(x)|$, for  $i=1,2$. By observing that $\substitute x v t=\termapp{\substitute x v {t_1}}{\substitute x v {t_2}}$
  and by using $\ruleVArrowE$, we get a derivation $\derivable{\Phi_{\substitute{x}{v}{t}}}{\sequV{\ctxtsum{\ctxtres{\Gamma}{x}{}}{\Delta}{}}{\assign{\substitute{x}{v}{t}}{\tau}}}{\SysCBV}$ such that
  $\sizev{\Phi_{\substitute{x}{v}{t}}} = (\sizev{\Phi_{t_1}} + \sizev{\Phi^1_{v}} - |\Gam_1(x)|)+(\sizev{\Phi_{t_2}} + \sizev{\Phi^2_{v}} - |\Gam_2(x)|)+1=
(\sizev{\Phi_{t_1}} + \sizev{\Phi_{t_2}} +1)+ (\sizev{\Phi^1_{v}} + \sizev{\Phi^2_{v}})-(|\Gam_1(x)|)+|\Gam_2(x)|)=  
\sizev{\Phi_{t}} + \sizev{\Phi_{v}} - |\Gam(x)|$.

All the other cases are similar to the previous one, and follow easily from the \ih and Lemma~\ref{l:split-value}.
\end{proof}

\begin{lemma}[Weighted Subject Reduction]
\label{l:sr-value}
Let $\derivable{\Phi}{\sequV{\Gamma}{\assign{t}{\tau}}}{\SysCBV}$ and $t
\rewrite{\callbyvalue} t'$. Then, there exists
$\derivable{\Phi'}{\sequV{\Gamma}{\assign{t'}{\tau}}}{\SysCBV}$
such that $\sizev{\Phi} > \sizev{\Phi'}$.
\end{lemma}

\begin{proof}
  By induction on $t \rewrite{\callbyvalue} t'$. We only show the base case
  $t \rewrite{\sVal} t'$ as the case $t \rewrite{\dBeta} t'$ is very similar
  to the one in Lemma~\ref{l:wsr-bang}. The inductive cases are
  straightforward.

  Let 
  $t = \termsubs{x}{\ctxtapp{\ctxt{L}}{v}}{s}$ and
  $t' = \ctxtapp{\ctxt{L}}{\substitute{x}{v}{s}}$. We proceed by
  induction on $\ctxt{L}$.
    \begin{itemize}
      \item $\ctxt{L} = \Box$. Then $\Gamma = \ctxtsum{(\ctxtres{\Gamma'}{x}{})}{\Delta}{}$
        s.t. $\Gamma'(x) = \M$ and $\Phi$ has the following form
        $$
\prooftree
  \derivable{\Phi_{s}}{\sequV{\Gamma'}{\assign{s}{\tau}}}{\SysCBV}
  \quad
 \derivable{\Phi_{v}}{ \sequV{\Delta}{\assign{v}{\M}}}{\SysCBV}
  \justifies
  \sequV{\Gamma}{\assign{\termsubs{x}{v}{s}}{\tau}}
\using
  \ruleVESubs
\endprooftree $$
      Thus, we conclude directly by Lemma~\ref{l:substitution-value} with $\Phi_{s}$ and
      $\Phi_v$. Notice that $\sizev{\Phi} = 1+ \sizev{\Phi_s} 
      + \sizev{\Phi_v}$, while $\sizev{\Phi'} = \sizev{\Phi_s}
      + \sizev{\Phi_v} - |\M|$.
      
      \item $\ctxt{L} = \termsubs{y}{r}{\ctxt{L'}}$. Then $\Gamma =
      \ctxtsum{(\ctxtres{\Gamma'}{x}{})}{\ctxtsum{(\ctxtres{\Delta}{y}{})}{\Delta'}{}}{}$
      with $\Gamma'(x) = \M$, $\Delta(y) = \M'$ and $$
\prooftree
  \derivable{\Phi_{s}}{\sequV{\Gamma'}{\assign{s}{\tau}}}{\SysCBV}
  \quad
  \Rule{
    \derivable{\Phi_{\ctxt{L'}}}{\sequV{\Delta}{\assign{\ctxtapp{\ctxt{L'}}{v}}{\M}}}{\SysCBV}
    \quad
    \derivable{\Phi_{r}}{\sequV{\Delta'}{\assign{r}{\M'}}}{\SysCBV}
  }{
  \sequV{\ctxtsum{(\ctxtres{\Delta}{y}{})}{\Delta'}{}}{\assign{\ctxtapp{\termsubs{y}{r}{\ctxt{L'}}}{v}}{\M}}
  }{\ruleVESubs}
\justifies
  \sequV{\Gamma}{\assign{\termsubs{x}{
  \ctxtapp{\ctxt{L}}{v}
  }{s}}{\tau}}
\using
  \ruleVESubs
\endprooftree
      $$ Thus, we build $$
\Rule{
  \derivable{\Phi_{s}}{\sequV{\Gamma'}{\assign{s}{\tau}}}{\SysCBV}
  \quad
  \derivable{\Phi_{\ctxt{L'}}}{\sequV{\Delta}{\assign{\ctxtapp{\ctxt{L'}}{v}}{\M}}}{\SysCBV}
}{
  \derivable{\Psi}{\sequV{\ctxtsum{(\ctxtres{\Gamma'}{x}{})}{\Delta}{}}{\assign{\termsubs{x}{
  \ctxtapp{\ctxt{L'}}{v}
  }{s}}{\tau}}}{\SysCBV}
}{\ruleVESubs}
      $$ and by the \ih there exists
      $\derivable{\Psi'}{\sequV{\ctxtsum{(\ctxtres{\Gamma'}{x}{})}{\Delta}{}}{\assign{\ctxtapp{\ctxt{L'}}{\substitute{x}{v}{s}}}{\tau}}}{\SysCBV}$
      such that $\sizev{\Psi} > \sizev{\Psi'}$. Then, we conclude with $$
\Rule{
  \derivable{\Psi'}{\sequV{\ctxtsum{(\ctxtres{\Gamma'}{x}{})}{\Delta}{}}{\assign{\ctxtapp{\ctxt{L'}}{\substitute{x}{v}{s}}}{\tau}}}{\SysCBV}
  \quad
  \derivable{\Phi_{r}}{\sequV{\Delta'}{\assign{r}{\M'}}}{\SysCBV}
}{
  \derivable{\Phi'}{\sequV{\ctxtsum{(\ctxtres{\Gamma'}{x}{})}{\ctxtsum{(\ctxtres{\Delta}{y}{})}{\Delta'}{}}{}}{\assign{\ctxtapp{\ctxt{L}}{\substitute{x}{v}{s}}}{\tau}}}{\SysCBV  }
}{\ruleVESubs} $$
      since we may assume that $y \notin \dom{\Gamma'}$. Notice that
      $\sizev{\Phi} = \sizev{\Psi} + \sizev{\Phi_{r}} +1 > \sizev{\Psi'} +
      \sizev{\Phi_{r}} +1  = \sizev{\Phi'}$.
    \end{itemize}
    
\end{proof}

The ability of merging different type derivations of a given value is necessary
for proving an anti-substitution lemma for CBV.

\begin{lemma}[Merge type for value]
\label{l:merge-value}
Let $\many{\derivable{\Phi_{v}^{i}}{\sequV{\Gamma_i}{\assign{v}{\M_i}}}{\SysCBV}}{i \in I}$.
Then, there exists $\derivable{\Phi_{v}}{\sequV{\Gamma}{\assign{v}{\M}}}{\SysCBV}$ such that
$\Gamma = \ctxtsum{}{\Gamma_i}{i \in I}$, $\M = \ctxtsum{}{\M_i}{i \in I}$
and $\sizev{\Phi_v} = +_{i \in I} \sizev{\Phi_{v}^{i}}$.
\end{lemma}

\begin{proof}
The proof is straightforward by case analysis on $v$.
\end{proof}

\begin{lemma}[Anti-Substitution]
\label{l:antisubstitution-value}
Let $\derivable{\Phi_{\substitute{x}{v}{t}}}{\sequV{\Gamma}{\assign{\substitute{x}{v}{t}}{\tau}}}{\SysCBV}$.
Then, there exist $\derivable{\Phi_{t}}{\sequV{\Gamma'}{\assign{t}{\tau}}}{\SysCBV}$ and
$\derivable{\Phi_{v}}{\sequV{\Delta}{\assign{v}{\Gamma'(x)}}}{\SysCBV}$ such
that $\Gamma = \ctxtsum{\ctxtres{\Gamma'}{x}{}}{\Delta}{}$ and
$\sizev{\Phi_{\substitute{x}{v}{t}}} = \sizev{\Phi_{t}} + \sizev{\Phi_{v}}
- |\Gam'(x)|$.
\end{lemma}

\begin{proof}
By induction on $t$. If $t = x$, then $\substitute{x}{v}{t} = v$ and it is
necessarily the case $\tau = \M$. We let $\Phi_v = \Phi_{\substitute{x}{v}{t}}$
and $\derivable{\Phi_t}{\sequ{\assign{x}{\M}}{\assign{x}{\M}}}{\SysCBV}$ by
$\ruleVAxiom$. We conclude since $\Gamma'(x) = \M$ and $\sizev{\Phi_t} = |\M|$
by definition. If $t = y \neq x$, then $\substitute{x}{v}{t} = y$ and we let
$\Phi_t = \Phi_{\substitute{x}{v}{t}}$. Therefore, by $\ruleVAxiom$, $\Gamma'(x)
= \emul$. Moreover, by rules $\ruleVAxiom$ and $\ruleVArrowI$,
$\derivable{\Phi_{v}}{\sequV{}{\assign{v}{\emul}}}{\SysCBV}$.
Thus, $\sizev{\Phi_{v}} = 0 = |\Gamma'(x)|$ and we conclude. 

If $t = \termapp{t_1}{t_2}$, then $\substitute{x}{v}{t} =
\termapp{\substitute{x}{v}{t_1}}{\substitute{x}{v}{t_2}}$ and there exist a
type $\M$ and two derivations
$\derivable{\Phi_{\substitute{x}{v}{t_1}}}{\sequV{\ctxtsum{\ctxtres{\Gamma_1}{x}{}}{\Delta_1}{}}{\assign{\substitute{x}{v}{t_1}}{\intertype{\functtype{\M}{\tau}}{}}}}{\SysCBV}$
and
$\derivable{\Phi_{\substitute{x}{v}{t_2}}}{\sequV{\ctxtsum{\ctxtres{\Gamma_2}{x}{}}{\Delta_2}{}}{\assign{\substitute{x}{v}{t_2}}{\M}}}{\SysCBV}$
such that $\Gamma = \Gamma_1 + \Gamma_2$ and
$\size{\Phi_{\substitute{x}{v}{t}}} =
\size{\Phi_{\substitute{x}{v}{t_1}}} + \size{\Phi_{\substitute{x}{v}{t_2}}} +
1$. Using the \ih on $t_1$ and $t_2$ we get four derivations:
\begin{enumerate}
  \item\label{l:antisubstitution-value:deriv:t1}
  $\derivable{\Phi_{t_1}}{\sequV{\Gamma'_1}{\assign{t_1}{\intertype{\functtype{\M}{\tau}}{}}}}{\SysCBV}$ 

  \item\label{l:antisubstitution-value:deriv:v1}
  $\derivable{\Phi^1_{v}}{\sequV{\Delta_1}{\assign{v}{\Gamma'_1(x)}}}{\SysCBV}$

  \item\label{l:antisubstitution-value:deriv:t2}
  $\derivable{\Phi_{t_2}}{\sequV{\Gamma'_2}{\assign{t_2}{\M}}}{\SysCBV}$

  \item\label{l:antisubstitution-value:deriv:v2}
  $\derivable{\Phi^2_{v}}{\sequV{\Delta_2}{\assign{v}{\Gamma'_2(x)}}}{\SysCBV}$ 
\end{enumerate}
such that $\Gamma_i = \ctxtsum{\ctxtres{\Gamma'_i}{x}{}}{\Delta_i}{}$ and
$\size{\Phi_{\substitute{x}{v}{t_i}}} = \size{\Phi_{t_i}} + \size{\Phi^i_{v}} -
|\Gamma'_i(x)|$ for $i = 1,2$. Lemma~\ref{l:merge-value} applied to
(\ref{l:antisubstitution-value:deriv:v1}) and
(\ref{l:antisubstitution-value:deriv:v2}) above provides a derivation
$\derivable{\Phi_{v}}{\sequV{\Delta}{\assign{v}{\Gamma'(x)}}}{\SysCBV}$ 
where $\Delta = \Delta_1 + \Delta_2$ and $\Gamma'(x) = \Gamma_1'(x) +
\Gamma_2'(x)$ . Using (\ref{l:antisubstitution-value:deriv:t1}) and
(\ref{l:antisubstitution-value:deriv:t2}) above we get
$\derivable{\Phi_{t}}{\sequV{\Gamma'}{\assign{t}{\tau}}}{\SysCBV}$ where
$\Gamma' = \Gamma'_1 + \Gamma'_2$ and $\size{\Phi_{t}} = \size{\Phi_{t_1}} +
\size{\Phi_{t_2}} + 1$. We have that $\Gamma =
\ctxtsum{\ctxtres{\Gamma'}{x}{}}{\Delta}{}$ and
$\sizev{\Phi_{\substitute{x}{v}{t}}} = \sizev{\Phi_{t}} + \sizev{\Phi_{v}} -
|\Gamma'(x)|$, and we conclude.

All the other cases are similar to the previous one, and follow easily from the
\ih and Lemma~\ref{l:merge-value}.
\end{proof}

\begin{lemma}[Weighted Subject Expansion]
\label{l:se-value}
Let $\derivable{\Phi'}{\sequV{\Gamma}{\assign{t'}{\tau}}}{\SysCBV}$ and $t
\rewrite{\callbyvalue} t'$. Then, there exists
$\derivable{\Phi}{\sequV{\Gamma}{\assign{t}{\tau}}}{\SysCBV}$
such that $\sizev{\Phi} >  \sizev{\Phi'}$.
\end{lemma}

\begin{proof}
By induction on $t \rewrite{\callbyvalue} t'$ where, in
particular, Lemma~\ref{l:antisubstitution-value} is used in the base case
$\termsubs{x}{v}{t} \rrule{\sVal} \substitute{x}{v}{t}$.
\end{proof}

As in the case of  CBN,  a lemma establishing that  normal forms are typable is needed.

\begin{lemma}
\label{l:normal-forms-value}
Let $t \not\rewrite{\callbyvalue}$. Then, $t$ is $\SysCBV$-typable.
\end{lemma}

\begin{proof}
By induction on $t$, we prove simultaneously the following statements:
\begin{enumerate}
  \item\label{l:normal-forms-value:vr} If $t \in \VarHCBVNF$, then for every
  multiset type $\M$ there exists $\Gamma$ such that $\derivable{}{\sequV{\Gamma}{\assign{t}{\M}}}{\SysCBV}$.
  \item\label{l:normal-forms-value:ne} If $t \in \HCBVNF$, then for every type
  $\tau$ there exists $\Gamma$ such that  $\derivable{}{\sequV{\Gamma}{\assign{t}{\tau}}}{\SysCBV}$.
  \item\label{l:normal-forms-value:no} If $t \in \CBVNF$, then
    there exist $\Gamma$ and $\M$ such that 
  $\derivable{}{\sequV{\Gamma}{\assign{t}{\M}}}{\SysCBV}$.
\end{enumerate}
We only analyse the key cases. If $t = x \in \VarHCBVNF$ we conclude by
$\ruleVAxiom$. If $t = \termapp{s}{u} \in \HCBVNF$, then $u \in \CBVNF$ and
there are two possible cases: $s \in \VarHCBVNF$ or $s \in \HCBVNF$. Let
$\tau$ by any type. By \ih (\ref{l:normal-forms-value:no}),
$\derivable{}{\sequV{\Delta}{\assign{u}{\M}}}{\SysCBV}$. Moreover, by \ih
(\ref{l:normal-forms-value:vr}) or (\ref{l:normal-forms-value:ne}) resp., we have
$\derivable{}{\sequV{\Gamma}{\assign{s}{\multiset{\functtype{\M}{\tau}}}}}{\SysCBV}$
and we conclude by $\ruleVArrowE$. If $t = \termabs{x}{s} \in \CBVNF$ we
conclude by $\ruleVArrowI$ with
$\derivable{}{\sequV{\Gamma}{\assign{t}{\emul}}}{\SysCBV}$. If $t =
\termsubs{x}{u}{s}$, there are three possible cases: $s \in \VarHCBVNF$, $s \in
\HCBVNF$ or $s \in \CBVNF$. In either case, by the proper \ih we get
$\derivable{}{\sequV{\Gamma}{\assign{s}{\M}}}{\SysCBV}$ (resp. $\tau$) and
conclude by $\ruleVESubs$, given that $u \in \HCBVNF$ and
$\derivable{}{\sequV{\Delta}{\assign{u}{\Gamma(x)}}}{\SysCBV}$ by \ih
(\ref{l:normal-forms-value:ne}).
\end{proof}

We  now relate $\SysCBV$-typability with $\callbyvalue$-normalisation.  We also
put in evidence the quantitative aspect of system $\SysCBV$, so that we
introduce the \emphdef{$\callbyvalue$-size} of terms as: $\valsize{x} \eqdef 0$,
$\valsize{\termabs{x}{t}} \eqdef 0$, $\valsize{\termapp{t}{u}} \eqdef 1 +
\valsize{t} + \valsize{u}$, and $\valsize{\termsubs{x}{u}{t}} \eqdef
1 + \valsize{t} + \valsize{u}$.
Note in particular that given a derivation
$\Phi_t$ for a term $t$ we always have $\sizev{\Phi_t} \geq \valsize{t}$.

\begin{theorem}[Soundness and Completeness for System $\SysCBV$]
\label{t:normalisation-value}
Let ${t \in \TermLambda}$. Then, $t$ is $\SysCBV$-typable iff $t$ is
$\callbyvalue$-normalising. Moreover, if
$\derivable{\Phi}{\sequV{\Gamma}{\assign{t}{\tau}}}{\SysCBV}$, then there
exists $p \in \CBVNF$ such that
$t \rewriten{\callbyvalue}^{(\cbeta,\cexp)} p$ and $\sizev{\Phi} \geq \cbeta +
\cexp + \valsize{p}$.
\end{theorem}

\begin{proof}
The $\Rightarrow$ direction holds by WSR (Lemma~\ref{l:sr-value}), while the
$\Leftarrow$ direction follows from Lemma~\ref{l:normal-forms-value} and WSE
(Lemma~\ref{l:se-value}). The \emph{moreover} statement holds by
Lemma~\ref{l:sr-value} and the fact that the size of the type derivation of $p$ is greater than or equal to $\valsize{p}$.
\end{proof}


\section{A Tight Type System Giving Exact Bounds}
\label{s:tight}

In order to count exactly the length of $\bangweak$-reduction sequences to
normal forms, we first fix a \emph{deterministic} strategy for the
$\BangRev$-calculus, called $\weakstg$, which computes the \emph{same}
$\bangweak$-normal forms. We then define the \emph{tight} type system
$\SysTight$, being able to count exactly the length of $\weakstg$-reduction
sequences. Theorem~\ref{t:confluence}, stating that any two different reduction
paths to normal form have the same length, guarantees that system $\SysTight$
is able to count exactly the length of \emph{any} $\bangweak$-reduction
sequence to $\bangweak$-normal form.

\paragraph{{\bf A Deterministic Strategy for the $\BangRev$-Calculus}}
The reduction relation $\lored$ defined below is a deterministic version of
$\rewrite{\bangweak}$ and is used, as explained, as a technical tool of
our development.

\[
\begin{array}{c}
\Rule{}
     {\termapp{\ctxtapp{\ctxt{L}}{\termabs{x}{t}}}{u} \lored \ctxtapp{\ctxt{L}}{\termsubs{x}{u}{t}}}
     {}
\qquad
\Rule{}
     {\termsubs{x}{\ctxtapp{\ctxt{L}}{\termbang{u}}}{t} \lored \ctxtapp{\ctxt{L}}{\substitute{x}{u}{t}}}
     {}
\qquad
\Rule{}
     {\termder{(\ctxtapp{\ctxt{L}}{\termbang{t}})} \lored \ctxtapp{\ctxt{L}}{t}}
     {}
\\
\\
\Rule{t \lored u }
     {\termabs{x}{t} \lored \termabs{x}{u}}
     {}
\qquad
\Rule{t \lored u \quad \neg\pbang{t}}
     {\termsubs{x}{t}{r} \lored \termsubs{x}{u}{r}}        
     {}
\qquad
\Rule{t \lored u \quad \neg\pbang{t}}
     {\termder{t} \lored \termder{u}}
     {}
\\
\\
\Rule{t \lored u \quad \neg\pabs{t}}
     {\termapp{t}{r} \lored \termapp{u}{r}}
     {}
\qquad
\Rule{t \lored u \quad \pnabs{r}}
     {\termapp{r}{t} \lored \termapp{r}{u}}        
     {}
\qquad
\Rule{t \lored u \quad \pnbang{r}}
     {\termsubs{x}{r}{t} \lored \termsubs{x}{r}{u}}
     {}
\end{array} \]

The rules in the first line correspond to the base cases.
  The first of these rules is  the multiplicative case, while the
  other two  are the exponential cases.
  The six remaining rules specify the
  closure by weak  contexts. 

Normal forms of $\rewrite{\bangweak}$ and $\lored$ are the same, both
characterised by the set $\ipnrml$.

\begin{proposition}
\label{l:equivalence}
Let $t \in \TermExplicit$. Then,
\begin{inparaenum}[(1)]
  \item\label{l:equivalence:bangweak} $t \not\rewrite{\bangweak}$ iff
  \item\label{l:equivalence:lored} $t \not\lored$ iff
  \item\label{l:equivalence:pnrml} $\pnrml{t}$.
\end{inparaenum}
\end{proposition}

\begin{proof}
Notice that $(\ref{l:equivalence:bangweak}) \implies
(\ref{l:equivalence:lored})$ follows from $\lored \mathbin{\subset}
\rewrite{\bangweak}$. Moreover, $(\ref{l:equivalence:bangweak})$ iff
$(\ref{l:equivalence:pnrml})$ holds by Proposition~\ref{prop:normal}. The proof of
$(\ref{l:equivalence:lored}) \implies (\ref{l:equivalence:pnrml})$ follows from
a straightforward adaptation of the proof of Proposition~\ref{prop:normal}.   
\end{proof}

\paragraph{{\bf The Type System $\SysTight$}}
We now extend the type system $\SysBang$ to a \emph{tight} one, called
$\SysTight$, being able to provide \emph{exact} bounds for
$\weakstg$-normalising sequences and size of normal forms. The
technique is based on~\cite{AccattoliGK18}, which defines type
systems to count reduction lengths for different strategies in the
$\lambda$-calculus. The notion of tight derivation turns out to be a
particular implementation of \emph{minimal derivation}, pioneered
by de Carvalho in~\cite{Carvalho:thesis}, where exact bounds for CBN abstract
machines are inferred from minimal type derivations.

We define the following sets of types:
\begin{center}
\begin{tabular}{rrcll}
\textbf{(Tight Constants)}     & $\typetight$   & $\Coloneq$ & $\typeabs \mid \typebang \mid \typeneutral$ \\
\textbf{(Types)}           & $\sigma, \tau$ & $\Coloneq$ & $\typetight \mid \M \mid \functtype{\M}{\sigma}$ \\
\textbf{(Multiset Types)}  & $\M$           & $\Coloneq$ & $\intertype{\sigma_i}{i \in I}$  where $I$ is a finite set
\end{tabular}
\end{center}

In contrast with $\SysBang$ where an infinite countable set of
  variable is used, following the standard presentations in the
  literature, system $\SysTight$ relies only in the use of a few type
  constants with a specific semantics.  Inspired
by~\cite{AccattoliGK18}, which only uses two constant types $\typeabs$
and $\typeneutral$ for abstractions and neutral terms respectively, we
now use three tight constants. Indeed, the
constant $\typeabs$ (resp.  $\typebang$) types terms whose normal form
has the shape $\ctxtapp{\ctxt{L}}{\termabs{x}{t}}$ (resp.
$\ctxtapp{\ctxt{L}}{\termbang{t}}$), and the constant $\typeneutral$
types terms whose normal form is in $\icfntrl$. As a matter of
notation, given an arbitrary tight constant
$\typetight_0$ we write $\overline{\typetight_0}$ to denote a tight
constant different from $\typetight_0$. Thus for
instance, $\overline{\typeabs} \in \{\typebang, \typeneutral\}$.

Typing contexts are functions from variables to multiset types, assigning the
empty multiset to all but a finite number of variables. Sequents are of the
form $\sequT{\Gamma}{\assign{t}{\sigma}}{\cbeta}{\cexp}{\csize}$, where the
natural numbers $\cbeta$, $\cexp$ and $\csize$ provide information on the
reduction of $t$ to normal form, and on the size of its normal form. More
precisely, $\cbeta$ (resp. $\cexp$) indicates the number of multiplicative
(resp. exponential) steps to normal form, while $\csize$ indicates the
$\bangweak$-size of this normal form. Observe that we do not count $\sBang$ and
$\dBang$ steps separately, because both of them are exponential steps of the
same nature. It is also worth noticing that only two counters suffice in the
case of the $\lambda$-calculus~\cite{AccattoliGK18}, one to count
$\beta$-reduction steps, and another to count the $\bangweak$-size of normal
forms. The difficulty in the case of the $\BangRev$-calculus is to statically
discriminate between multiplicative and exponential steps.

A multiset type $\intertype{\sigma_i}{i \in I}$ is
  \emphdef{tight}, written $\ptight{\intertype{\sig_i}{i \in I}}$, if
  $\sigma_i \in \typetight$ for all $i \in I$. A context $\Gamma$ is
  said to be \emphdef{tight} if it assigns tight multisets to all
  variables. A type derivation
  $\derivable{\Phi}{\sequT{\Gamma}{\assign{t}{\sigma}}{\cbeta}{\cexp}{\csize}}{\SysTight}$
  is \emphdef{tight} if $\Gamma$ is tight and $\sigma \in
  \typetight$.

Typing rules
(Figure~\ref{fig:typingSchemesTight}) are split in two groups: the
\emph{persistent} and the \emph{consuming} ones. A constructor is consuming
(resp. persistent) if it is consumed (resp. not consumed) during
$\bangweak$-reduction to $\bangweak$-normal form. For instance, in
$\termapp{\termapp{\termder{(\termbang{\Kterm})}}{(\termbang{\id})}}{(\termbang{\Omega})}$
the two abstractions of $\Kterm$ are consuming, while the abstraction of $\id$
is persistent, and all the other constructors are also consuming, except those
of $\Omega$ that turns out to be an untyped subterm. This dichotomy between
consuming/persistent constructors has been used in~\cite{AccattoliGK18} for
the $\lambda$-calculus, and adapted here for the $\BangRev$-calculus.

\begin{figure}
\centering $
\begin{array}{c}
\textbf{Persistent Typing Rules}
\\[1em]
\Rule{\sequT{\Gamma}{\assign{t}{\typeneutral}}{\cbeta}{\cexp}{\csize}
      \quad
      \sequT{\Delta}{\assign{u}{\overline{\typeabs}}}{\cbeta'}{\cexp'}{\csize'}
     }
     {\sequT{\ctxtsum{\Gamma}{\Delta}{}}{\assign{\termapp{t}{u}}{\typeneutral}}{\cbeta+\cbeta'}{\cexp+\cexp'}{\csize+\csize'+1}}
     {\ruleTArrowE}
\qquad
\Rule{\sequT{\Gamma}{\assign{t}{\typetight}}{\cbeta}{\cexp}{\csize}
      \quad
      \ptight{\Gamma(x)}
     }
     {\sequT{\ctxtres{\Gamma}{x}{}}{\assign{\termabs{x}{t}}{\typeabs}}{\cbeta}{\cexp}{\csize+1}}
     {\ruleTArrowI}
\\
\\
\Rule{\vphantom{\sequT{}{\assign{\termbang{t}}{\typebang}}{0}{0}{0}}}
     {\sequT{}{\assign{\termbang{t}}{\typebang}}{0}{0}{0}}
     {\ruleTBang}
\qquad
\Rule{\sequT{\Gamma}{\assign{t}{\typeneutral}}{\cbeta}{\cexp}{\csize}}
     {\sequT{\Gamma}{\assign{\termder{t}}{\typeneutral}}{\cbeta}{\cexp}{\csize+1}}
     {\ruleTDer}
\\
\\
\Rule{\sequT{\Gamma}{\assign{t}{\tau}}{\cbeta}{\cexp}{\csize}
      \quad
      \sequT{\Delta}{\assign{u}{\typeneutral}}{\cbeta'}{\cexp'}{\csize'}
      \quad
      \ptight{\Gamma(x)}
     }
     {\sequT{\ctxtsum{(\ctxtres{\Gamma}{x}{})}{\Delta}{}}{\assign{\termsubs{x}{u}{t}}{\tau}}{\cbeta+\cbeta'}{\cexp+\cexp'}{\csize+\csize'}}
     {\ruleTESubs}
\\[2em]
\textbf{Consuming Typing Rules}
\\[1em]
\Rule{\vphantom{\sequT{}{\assign{\termbang{t}}{\typebang}}{0}{0}{0}}}
     {\sequT{\assign{x}{\intertype{\sigma}{}}}{\assign{x}{\sigma}}{0}{0}{0}}
     {\ruleDAxiom}
\qquad
\Rule{\sequT{\Gamma}{\assign{t}{\functtype{\M}{\tau}}}{\cbeta}{\cexp}{\csize}
      \quad
      \sequT{\Delta}{\assign{u}{\M}}{\cbeta'}{\cexp'}{\csize'}
     }
     {\sequT{\ctxtsum{\Gamma}{\Delta}{}}{\assign{\termapp{t}{u}}{\tau}}{\cbeta+\cbeta'+1}{\cexp+\cexp'}{\csize+\csize'}}
     {\ruleDArrowE}
\\
\\
\Rule{\sequT{\Gamma}{\assign{t}{\functtype{\M}{\tau}}}{\cbeta}{\cexp}{\csize}
      \quad
      \sequT{\Delta}{\assign{u}{\typeneutral}}{\cbeta'}{\cexp'}{\csize'}
      \quad
      \ptight{\M}
     }
     {\sequT{\ctxtsum{\Gamma}{\Delta}{}}{\assign{\termapp{t}{u}}{\tau}}{\cbeta+\cbeta'+1}{\cexp+\cexp'}{\csize+\csize'}}
     {\ruleBTApp}
\\
\\
\Rule{\sequT{\Gamma}{\assign{t}{\tau}}{\cbeta}{\cexp}{\csize}}
     {\sequT{\ctxtres{\Gamma}{x}{}}{\assign{\termabs{x}{t}}{\functtype{\Gamma(x)}{\tau}}}{\cbeta}{\cexp}{\csize}}
     {\ruleDArrowI}
\qquad
\Rule{(\sequT{\Gamma_i}{\assign{t}{\sigma_i}}{\cbeta_i}{\cexp_i}{\csize_i})_{i \in I}}
     {\sequT{\ctxtsum{}{\Gamma_i}{i \in I}}{\assign{\termbang{t}}{\intertype{\sigma_i}{\iI}}}{+_{\iI}{\cbeta_i}}{1+_{\iI}{\cexp_i}}{+_{\iI}{\csize_i}}}
     {\ruleDBang}
\\
\\
\Rule{\sequT{\Gamma}{\assign{t}{\intertype{\sigma}{}}}{\cbeta}{\cexp}{\csize}}
     {\sequT{\Gamma}{\assign{\termder{t}}{\sigma}}{\cbeta}{\cexp}{\csize}}
     {\ruleDDer}
\qquad
\Rule{\sequT{\Gamma}{\assign{t}{\sigma}}{\cbeta}{\cexp}{\csize}
      \quad
      \sequT{\Delta}{\assign{u}{\Gamma(x)}}{\cbeta'}{\cexp'}{\csize'}
     }
     {\sequT{\ctxtsum{(\ctxtres{\Gamma}{x}{})}{\Delta}{}}{\assign{\termsubs{x}{u}{t}}{\sigma}}{\cbeta+\cbeta'}{\cexp+\cexp'}{\csize+\csize'}}
     {\ruleDESubs}
\end{array} $
\caption{System $\SysTight$ for the $\BangRev$-Calculus.}
\label{fig:typingSchemesTight}
\end{figure}

Observe that in every typing rule the counters of the conclusion are at least
the sums of the corresponding counters  of the premises. In some cases, one
counter may  undergo an additional increment, as explained below. The
persistent rules are those typing persistent constructors, so that none of them
increases the first two counters, but only possibly the third one, which
contributes to the size of the normal form. The consuming rules type consuming
constructors, so that they may increase one of the first two counters,
contributing to the length of the normalisation sequence. More precisely, rules
$\ruleDArrowE$ and $\ruleBTApp$ increment the first counter because the
(consuming) application will be used to perform a $\dBeta$-step, while rule
$\ruleDBang$ increments the second counter because the (consuming) bang will be
used to perform either a $\sBang$ or a $\dBang$-step. Rule $\ruleBTApp$ is
particularly useful to type $\dBeta$-redexes whose reduction does not create an
exponential redex, because the argument of the substitution created by the
$\dBeta$-step does not reduce to a bang.

\begin{example}
\label{example:t0-tight}
The following tight typing can be derived for term $t_0$ of
Example~\ref{example:t0}:
{\small \[
\Rule{
  \Rule{
    \Rule{
      \Rule{
        \Rule{
          \Rule{
            \Rule{}{
              \sequT{\assign{x}{\multiset{\typeabs}}}{\assign{x}{\typeabs}}{0}{0}{0}
            }{\ruleDAxiom}
          }{
            \sequT{\assign{x}{\multiset{\typeabs}}}{\assign{\termabs{y}{x}}{\functtype{\emul}{\typeabs}}}{0}{0}{0}
          }{\ruleDArrowI}
        }{
          \sequT{}{\assign{\termabs{x}{\termabs{y}{x}}}{\functtype{\multiset{\typeabs}}{\functtype{\emul}{\typeabs}}}}{0}{0}{0}
        }{\ruleDArrowI}
      }{
        \sequT{}{\assign{\termbang{\Kterm}}{\multiset{\functtype{\multiset{\typeabs}}{\functtype{\emul}{\typeabs}}}}}{0}{1}{0}
      }{\ruleDBang}
    }{
      \sequT{}{\assign{\termder{(\termbang{\Kterm})}}{\functtype{\multiset{\typeabs}}{\functtype{\emul}{\typeabs}}}}{0}{1}{0}
    }{\ruleDDer}
    \Rule{
      \Rule{
        \Rule{}{
          \sequT{\assign{x}{\multiset{\typeneutral}}}{\assign{x}{\typeneutral}}{0}{0}{0}
        }{\ruleDAxiom}
      }{
        \sequT{}{\assign{\termabs{x}{x}}{\typeabs}}{0}{0}{1}
      }{\ruleTArrowI}
    }{
      \sequT{}{\assign{\termbang{\id}}{\multiset{\typeabs}}}{0}{1}{1}
    }{\ruleDBang}
  }{
    \sequT{}{\assign{\termapp{\termder{(\termbang{\Kterm})}}{(\termbang{\id})}}{\functtype{\emul}{\typeabs}}}{1}{2}{1}
  }{\ruleDArrowE}
  \Rule{}{
    \sequT{}{\assign{\termbang{\Omega}}{\emul}}{0}{1}{0}
  }{\ruleDBang}
}{
  \sequT{}{\assign{\termapp{\termapp{\termder{(\termbang{\Kterm})}}{(\termbang{\id})}}{(\termbang{\Omega})}}{\typeabs}}{2}{3}{1}
}{\ruleDArrowE}
\] }
Note that the only persistent rule used is $\ruleTArrowI$ when typing $\id$,
thus contributing to count the $\bangweak$-size of the $\bangweak$-normal form
of $t_0$. Indeed, $\id$ is the $\bangweak$-normal form of $t_0$.
\end{example}

\paragraph{{\bf Soundness}}
\label{s:correctness}
We now study soundness of the type system $\SysTight$, which does not only
guarantee that typable terms are normalising --a qualitative property-- but
also provides quantitative (exact) information for normalising sequences. More
precisely, given a tight type derivation $\Phi$ with counters
$(\cbeta,\cexp,\csize)$ for a term $t$, $t$ is $\bangweak$-normalisable in
$(\cbeta+\cexp)$-steps and its $\bangweak$-normal form has $\bangweak$-size
$\csize$. Therefore, information about a \emph{dynamic} behaviour of $t$, is
extracted from a static typing property of $t$. The soundness proof is mainly
based on a subject reduction property (Lemma~\ref{l:QSR}), as well as on some
auxiliary results.

We start by the following remark, which is proved by inspecting the typing
rules:

\begin{remark}
\label{r:typing}
If
$\derivable{\Phi}{\sequT{\Gamma}{\assign{t}{\sigma}}{\cbeta}{\cexp}{\csize}}{\SysTight}$
then:
\begin{itemize}
  \item $\pabs{t}$ implies $\sigma = \typeabs$ or $\sigma = \functtype{\M}{\tau}$.
  \item $\pbang{t}$ implies $\sigma = \typebang$ or $\sigma = \M$.
\end{itemize}
\end{remark}

As in system $\SysBang$, typable terms are weakly clash-free, as stated in the
following lemma. This lemma is needed for proving
Lemma~\ref{l:tight-spreading-neutral}, which is one of the two \emph{tight
spreading} lemmas established here. By \emph{tight spreading} we mean that in a
type derivation the tightness of the final context implies (under some
additional hypotheses) the tightness of the final type, and hence of the
derivation itself.

\begin{lemma}
\label{l:clashes-do-not-type}
If
$\derivable{\Phi}{\sequT{\Gamma}{\assign{t}{\sigma}}{\cbeta}{\cexp}{\csize}}{\SysTight}$,
then $t$ is $\cfz$.
\end{lemma}

\begin{proof}
By straightforward induction in $t$.
\end{proof}

The following tight spreading lemmas will be used in
Lemma~\ref{l:tight-spreading-czero} and Lemma~\ref{l:czero-normal}, which in turn
ensure that in tight derivations the counters work as expected for normal
forms.

\begin{lemma}[Tight Spreading for Neutral Terms]
\label{l:tight-spreading-neutral}
Let
$\derivable{\Phi}{\sequT{\Gamma}{\assign{t}{\sigma}}{\cbeta}{\cexp}{\csize}}{\SysTight}$
such that $\pntrl{t}$. If $\Gamma$ is tight, then $\sigma \in \typetight$.
\end{lemma}

\begin{proof}
  We reason by induction on $t$.
    Notice that by \ih  for every
    subterm $u$ of $t$  verifying $\pntrl{u}$, 
    every derivation of $u$ having a tight typing context
   must also have a tight type subject. 

\begin{itemize}
  \item $t = x$. Then, $\Phi$ ends with rule $\ruleDAxiom$ and $\Gamma =
  \assign{x}{\intertype{\sig}{}}$ tight implies $\sig \in \typetight$.
  
  \item $t = \termapp{r}{u}$. By definition of $\pntrl{t}$, $\pnabs{r}$ and
  $\pnrml{u}$ hold. Moreover, $\neg\pbang{r}$ by
  Lemma~\ref{l:clashes-do-not-type}. Then, by Remark~\ref{r:normal}, it is
  necessarily the case that $\pnbang{r}$ holds, and hence $\pntrl{r}$ as
  well. There are three cases for $\Phi$:
  \begin{enumerate}
    \item if $\Phi$ ends with rule $\ruleTArrowE$, then $\sigma =
    \typeneutral$ and the statement trivially holds.

    \item if $\Phi$ ends with rule $\ruleDArrowE$, then $\Gamma =
      \ctxtsum{\Gamma'}{\Delta}{}$, $\cbeta = \cbeta' + \cbeta'' + 1$,
      $\cexp = \cexp' + \cexp''$, $\csize = \csize' + \csize''$ and,
      in particular,
      $\derivable{\Phi_r}{\sequT{\Gamma'}{\assign{r}{\functtype{\M}{\tau}}}{\cbeta'}{\cexp'}{\csize'}}{\SysTight}$
      with $\Gamma'$ tight (since $\Gamma$ is tight).  Then,
        $\pntrl{r}$ gives $\functtype{\M}{\tau} \in \typetight$ by \ih
        This is clearly a contradiction. Hence, this case does not
        apply.
    \item if $\Phi$ ends with rule $\ruleBTApp$, then we reason exactly as in
    the previous case, so that this case does not apply neither.
  \end{enumerate}
  
  \item $t = \termder{u}$. By definition of $\pntrl{t}$, $\pnbang{u}$ holds.
  Moreover, $\neg\pabs{u}$ by Lemma~\ref{l:clashes-do-not-type}. Then, by
  Remark~\ref{r:normal}, it is necessarily the case that $\pnabs{u}$ holds,
  and hence $\pntrl{u}$ as well. Then, there are two cases for
  $\Phi$:
  \begin{enumerate}
    \item if $\Phi$ ends with rule $\ruleTDer$, then $\sigma = \typeneutral$
    and the statement  trivially holds. 
    
    \item if $\Phi$ ends with rule $\ruleDDer$, then
    $\derivable{\Phi_u}{\sequT{\Gamma}{\assign{u}{\intertype{\sigma}{}}}{\cbeta}{\cexp}{\csize}}{\SysTight}$.
    Therefore, $\pntrl{u}$ gives $\intertype{\sigma}{} \in \typetight$ by \ih
    This is clearly a contradiction. Hence, this case does not apply.
  \end{enumerate}
  
  \item $t = \termsubs{x}{u}{r}$. By definition of $\pntrl{t}$, $\pntrl{r}$ and
  $\pnbang{u}$ hold. Moreover, $\neg\pabs{u}$ by
  Lemma~\ref{l:clashes-do-not-type}. Then, by Remark~\ref{r:normal}, it is
  necessarily the case that $\pnabs{u}$ holds, and hence $\pntrl{u}$ as well.
  Then, there are two cases for $\Phi$:
  \begin{enumerate}
    \item if $\Phi$ ends with rule $\ruleDESubs$, $\Gamma =
    \ctxtsum{\ctxtres{\Gamma'}{x}{}}{\Delta}{}$, $\cbeta = \cbeta' +
    \cbeta''$, $\cexp = \cexp' + \cexp''$, $\csize = \csize' + \csize''$ and,
    in particular,
    $\derivable{\Phi_u}{\sequT{\Delta}{\assign{u}{\Gamma'(x)}}{\cbeta''}{\cexp''}{\csize''}}{\SysTight}$ with $\Delta$ tight (since $\Gamma$ is tight).
    Then, $\pntrl{u}$ gives $\Gamma'(x) \in \typetight$ by \ih
    This leads to a contradiction, since $\Gamma'(x)$ is a multiset type by
    definition. Hence, this case does not apply.
  \item if $\Phi$ ends with rule $\ruleTESubs$, then $\Gamma =
    \ctxtsum{\ctxtres{\Gamma'}{x}{}}{\Delta}{}$, $\ptight{\Gamma'(x)}$, $\cbeta
    = \cbeta' + \cbeta''$, $\cexp = \cexp' + \cexp''$, $\csize = \csize' +
    \csize''$ and, in particular,
    $\derivable{\Phi_r}{\sequT{\Gamma'}{\assign{r}{\sigma}}{\cbeta'}{\cexp'}{\csize'}}{\SysTight}$.
    Moreover, $\Gamma$ tight and $\ptight{\Gamma'(x)}$ give $\Gamma'$ tight as
    well. We conclude by \ih with $\pntrl{r}$ that $\sigma \in \typetight$.
    \end{enumerate}
\end{itemize}
\end{proof}

\begin{lemma}[Tight Spreading for Zero Counters]
\label{l:tight-spreading-czero}
Let
$\derivable{\Phi}{\sequT{\Gamma}{\assign{t}{\sigma}}{\cbeta}{\cexp}{\csize}}{\SysTight}$
such that $\cbeta = \cexp = 0$ and $\sigma$ is not an arrow type. If $\Gamma$
is tight, then $\sigma \in \typetight$.
\end{lemma}

\begin{proof}
By induction on $\Phi$. Note that the statement trivially holds for the
rules $\ruleTArrowE$, $\ruleTArrowI$, $\ruleTBang$, $\ruleTDer$ in
Figure~\ref{fig:typingSchemesTight} since all of them conclude with $\sigma \in
\typetight$. We proceed by analysing the other rules in
Figure~\ref{fig:typingSchemesTight}.
\begin{itemize}
  \item $\ruleTESubs$. Then $\Gamma =
  \ctxtsum{\ctxtres{\Gamma'}{x}{}}{\Delta}{}$, $\ptight{\Gamma'(x)}$, $\csize =
  \csize' + \csize''$ and, in particular,
  $\derivable{\Phi_r}{\sequT{\Gamma'}{\assign{r}{\sigma}}{0}{0}{\csize'}}{\SysTight}$.
  Moreover, $\Gamma$ tight and $\ptight{\Gamma'(x)}$ give $\Gamma'$ tight as
  well. We then conclude directly by \ih with $\Phi_r$ that $\sigma \in
  \typetight$.

  \item $\ruleDAxiom$. Then $t = x$ and $\Gamma =
  \assign{x}{\intertype{\sig}{}}$ tight which implies $\sig \in \typetight$.

  \item $\ruleBTApp$. This case does not apply since it concludes with
  $\cbeta > 0$.
  
  \item $\ruleDArrowE$. This case does not apply since it concludes with
  $\cbeta > 0$.
  
  \item $\ruleDArrowI$. This case does not apply since it concludes with an
  arrow type.
  
  \item $\ruleDBang$. This case does not apply since it concludes with
  $\cexp > 0$.
  
  \item $\ruleDDer$. Then $t = \termder{u}$ and
  $\sequT{\Gamma}{\assign{u}{\intertype{\sigma}{}}}{0}{0}{\csize}$ is
  derivable. Then, the \ih gives $\intertype{\sigma}{} \in \typetight$
  which is clearly a contradiction. Thus, this case does not apply either.
  
  \item $\ruleDESubs$. Then $t = \termsubs{x}{u}{r}$, $\csize = \csize' +
  \csize''$ and $\Gamma = \ctxtsum{(\ctxtres{\Gamma'}{x}{})}{\Delta}{}$ tight
  such that, in particular
  $\sequT{\Delta}{\assign{u}{\Gamma(x)}}{0}{0}{\csize''}$ with $\Delta$
  tight. The \ih gives $\Gamma(x) \in \typetight$ which leads to a
  contradiction since it is a multiset type. Hence, this case does not apply.
\end{itemize}
\end{proof}

The following two lemmas are needed to establish the base case of the induction
proving soundness. 

\begin{lemma}
\label{l:czero-normal}
If
$\derivable{\Phi}{\sequT{\Gamma}{\assign{t}{\sigma}}{\cbeta}{\cexp}{\csize}}{\SysTight}$ is
tight, then $\cbeta = \cexp = 0$ iff $\pnrml{t}$.
\end{lemma}

\begin{proof}
$\left.\Rightarrow\right)$ By induction on $\Phi$.
\begin{itemize}
  \item $\ruleDAxiom$. Then $t = x$, $\Gamma = \intertype{\sigma}{}$ with
  $\sigma \in \typetight$. By definition $\pntrl{t}$ holds, which implies
  $\pnrml{t}$ as well.
  
  \item $\ruleTArrowE$. Then $t = \termapp{r}{u}$, $\sigma = \typeneutral$,
  $\Gamma = \ctxtsum{\Gamma'}{\Delta}{}$, $\csize = \csize' + \csize'' + 1$,
  $\sequT{\Gamma'}{\assign{r}{\typeneutral}}{0}{0}{\csize'}$,
  $\sequT{\Delta}{\assign{u}{\overline{\typeabs}}}{0}{0}{\csize''}$. Then,
  $\Gamma'$ and $\Delta$ are both tight, hence the \ih gives $\pnrml{r}$ and 
  $\pnrml{u}$. There are two cases to consider based on $\pnrml{r}$:
  \begin{enumerate}
    \item if $\pnabs{r}$ the result is immediate.

    \item if $\pnbang{r}$ then we can assume $r \notin \ipnabs$ too (since
    $\pnabs{r}$ is already considered). Then, by Remark~\ref{r:normal},
    $\pabs{r}$ holds. This leads to a contradiction with $r$ having type
    $\typeneutral$ (\cf Remark~\ref{r:typing}). Hence, this case does not apply.
  \end{enumerate}
  
  \item $\ruleBTApp$. This case does not apply since it concludes
  with $\cbeta > 0$.
  
  \item $\ruleDArrowE$. This case does not apply since it concludes
  with $\cbeta > 0$.
  
  \item $\ruleTArrowI$. Then $t = \termabs{x}{u}$, $\sigma = \typeabs$, $\Gamma
  = \ctxtres{\Gamma'}{x}{}$, $\csize = \csize' + 1$,
  $\sequT{\Gamma'}{\assign{u}{\typetight}}{0}{0}{\csize'}$ and
  $\ptight{\Gamma'(x)}$. Since $\Gamma$ is tight and $\ptight{\Gamma'(x)}$,
  $\Gamma'$ is tight as well. Then, by \ih $\pnrml{u}$ holds, which implies
  $\pnrml{t}$ too.
  
  \item $\ruleDArrowI$. Then $\sigma = \functtype{\M}{\tau}$ which
  contradicts the hypothesis of $\Phi$ being tight. Hence, this case does not
  apply.
  
  \item $\ruleTBang$. Then $t = \termbang{u}$ which implies $\pnabs{t}$ and
  hence $\pnrml{t}$.
  
  \item $\ruleDBang$. This case does not apply since it concludes with
  $\cexp > 0$.
  
  \item $\ruleTDer$. Then $t = \termder{u}$, $\sigma = \typeneutral$,
  $\csize = \csize' + 1$ and
  $\sequT{\Gamma}{\assign{u}{\typeneutral}}{0}{0}{\csize'}$. By \ih $\pnrml{u}$
  holds. There are two  cases to consider:
  \begin{enumerate}
    \item if $\pnbang{u}$ the result is immediate.
    
    \item if $\pnabs{u}$ and $u \notin \ipnbang$, then $\pbang{u}$ holds by
    Remark~\ref{r:normal}. This leads to a contradiction with $u$ having type
    $\typeneutral$ (\cf Remark~\ref{r:typing}). Hence, this case does not apply.
  \end{enumerate}
  
  \item $\ruleDDer$. Then $t = \termder{u}$ and
  $\derivable{\Phi_u}{\sequT{\Gamma}{\assign{u}{\intertype{\sigma}{}}}{0}{0}{\csize}}{\SysTight}$.
  Then, Lemma~\ref{l:tight-spreading-czero} on $\Phi_u$ give
  $\intertype{\sigma}{} \in \typetight$ which is clearly a contradiction. Thus,
  this case does not apply.
  
  \item $\ruleTESubs$. Then $t = \termsubs{x}{u}{r}$, $\Gamma =
    \ctxtsum{(\ctxtres{\Gamma'}{x}{})}{\Delta}{}$,
    $\csize = \csize' + \csize''$, $\sequT{\Gamma'}{\assign{r}{\sigma}}{0}{0}{\csize'}$,
  $\sequT{\Delta}{\assign{u}{\typeneutral}}{0}{0}{\csize''}$ and
  $\ptight{\Gamma'(x)}$. Since $\Gamma$ is tight and $\ptight{\Gamma'(x)}$,
  then $\Gamma'$ and $\Delta$ are both tight as well. Thus, \ih gives
  $\pnrml{r}$ and $\pnrml{u}$. Moreover, by definition $\pnrml{r}$ means
  $\pnabs{r}$ or $\pnbang{r}$. Same for $\pnrml{u}$, hence there are two
  different cases to analyse:
  \begin{enumerate}
    \item if $\pnbang{u}$ the result is immediate.
    \item if $\pnabs{u}$ and $u \notin \ipnbang$, then $\pbang{u}$ holds by
    Remark~\ref{r:normal}. This leads to a contradiction with $u$ having type
    $\typeneutral$ (\cf Remark~\ref{r:typing}). Hence, this case does not apply.
  \end{enumerate}
  
  \item $\ruleDESubs$. Then $t = \termsubs{x}{u}{r}$, $\csize = \csize' +
  \csize''$ and $\Gamma = \ctxtsum{(\ctxtres{\Gamma'}{x}{})}{\Delta}{}$ tight
  such that, in particular,
  $\derivable{\Phi_u}{\sequT{\Delta}{\assign{u}{\Gamma'(x)}}{0}{0}{\csize''}}{\SysTight}$
  with $\Delta$ tight. By Lemma~\ref{l:tight-spreading-czero} on $\Phi_u$,
  $\Gamma'(x) \in \typetight$ which leads to a contradiction, since
  $\Gamma'(x)$ is a multiset type by definition. Hence, this case does not apply.
\end{itemize}

$\left.\Leftarrow\right)$ By induction on $t$.
\begin{itemize}
  \item $t = x$. Then, $\Phi$ ends with rule $\ruleDAxiom$ and the statement holds
  trivially.
  
  \item $t = \termapp{r}{u}$. By definition $\pnrml{t}$ gives $\pnabs{r}$ and
  $\pnrml{u}$. Thus, $\pnrml{r}$ holds too. There are three  cases to consider:
  \begin{enumerate}
    \item if $\Phi$ ends with rule $\ruleTArrowE$, then $\sigma =
    \typeneutral$, $\Gamma = \ctxtsum{\Gamma'}{\Delta}{}$, $\cbeta = \cbeta' +
    \cbeta''$, $\cexp = \cexp' + \cexp''$, $\csize = \csize' + \csize'' + 1$,
    $\sequT{\Gamma'}{\assign{r}{\typeneutral}}{\cbeta'}{\cexp'}{\csize'}$,
    $\sequT{\Delta}{\assign{u}{\overline{\typeabs}}}{\cbeta''}{\cexp''}{\csize''}$.
    Moreover, $\Gamma'$ and $\Delta$ are both tight. Then, the \ih with
    $\pnrml{r}$ and $\pnrml{u}$ gives $\cbeta' = \cexp' = 0$ and $\cbeta'' =
    \cexp'' = 0$. Hence, $\cbeta = \cexp = 0$.
    
    \item if $\Phi$ ends with rule $\ruleDArrowE$, then $\Gamma =
    \ctxtsum{\Gamma'}{\Delta}{}$, $\cbeta = \cbeta' + \cbeta'' + 1$,
    $\cexp = \cexp' + \cexp''$, $\csize = \csize' + \csize''$ and, in
    particular,
    $\derivable{\Phi_r}{\sequT{\Gamma'}{\assign{r}{\functtype{\M}{\tau}}}{\cbeta'}{\cexp'}{\csize'}}{\SysTight}$.
    Moreover, by contra-positive of Remark~\ref{r:typing} with $\Phi_r$, it is
    necessarily the case that $\neg \pbang{r}$. Thus, together with $\pnabs{r}$ it
    gives $\pntrl{r}$ (\cf Remark~\ref{r:normal}). Also,
    $\Gamma$ tight implies $\Gamma'$ tight as well. Then, by
    Lemma~\ref{l:tight-spreading-neutral}, $\Phi_r$ is a tight typing, which
    leads to a contradiction with $r$ having a functional type. Hence, this
    case does not apply.

    \item if $\Phi$ ends with rule $\ruleBTApp$, then $\Gamma =
    \ctxtsum{\Gamma'}{\Delta}{}$, $\cbeta = \cbeta' + \cbeta'' +1$, $\cexp =
    \cexp' + \cexp''$,
    $\csize = \csize' + \csize''$ and, in particular,
    $\derivable{\Phi_r}{\sequT{\Gamma'}{\assign{r}{\functtype{\M}{\tau}}}{\cbeta'}{\cexp'}{\csize'}}{\SysTight}$.
    This case is identical to the previous one.
    \end{enumerate}
  
  \item $t = \termabs{x}{u}$. By definition $\pnrml{t}$ gives $\pnrml{u}$.
  There are two cases to consider for $\Phi$:
  \begin{enumerate}
    \item if $\Phi$ ends with rule $\ruleTArrowI$, then $\sigma = \typeabs$,
    $\Gamma = \ctxtres{\Gamma'}{x}{}$, $\csize = \csize' + 1$,
    $\derivable{\Phi_u}{\sequT{\Gamma'}{\assign{u}{\typetight}}{\cbeta}{\cexp}{\csize'}}{\SysTight}$
    and $\ptight{\Gamma'(x)}$. Since $\Gamma$ is tight and
    $\ptight{\Gamma'(x)}$, $\Gamma'$ is tight as well. Then, the \ih gives
    $\cbeta = \cexp = 0$.
    
    \item if $\Phi$ ends with rule $\ruleDArrowI$, then $\sigma =
    \functtype{\M}{\tau}$ which contradicts the hypothesis of $\Phi$ being
    tight. Hence, this case does not apply.
  \end{enumerate}
  
  \item $t = \termbang{u}$. There are two cases to consider for $\Phi$:
  \begin{enumerate}
    \item if $\Phi$ ends with rule $\ruleTBang$, then $\cbeta = \cexp = 0$ and
    the statement holds immediately.
    
    \item if $\Phi$ ends with rule $\ruleDBang$, then $\sigma$ is a multiset
    type which contradicts the hypothesis of $\Phi$ being tight. Hence, this
    case does not apply.
  \end{enumerate}
  
  \item $t = \termder{u}$. By definition $\pnrml{t}$ gives $\pnbang{u}$ which
  implies $\pnrml{u}$ as well. There are two cases to consider:
  \begin{enumerate}
    \item if $\Phi$ ends with rule $\ruleTDer$, then $\sigma = \typeneutral$,
    $\csize = \csize' + 1$ and
    $\sequT{\Gamma}{\assign{u}{\typeneutral}}{\cbeta}{\cexp}{\csize'}$. Then,
    the statement follows immediately from the \ih
    
    \item if $\Phi$ ends with rule $\ruleDDer$, then
    $\derivable{\Phi_u}{\sequT{\Gamma}{\assign{u}{\intertype{\sigma}{}}}{\cbeta}{\cexp}{\csize}}{\SysTight}$.
    Moreover, by contra-positive of Remark~\ref{r:typing} with $\Phi_u$, it is
    necessarily the case $\neg\pabs{u}$. Thus, together with $\pnbang{u}$ it
    gives $\pntrl{u}$ (\cf Remark~\ref{r:normal}). Then, by
    Lemma~\ref{l:tight-spreading-neutral}, $\Phi_u$ is a tight typing, which
    leads to a contradiction with $u$ having a multiset type. Hence, this case
    does not apply.
  \end{enumerate}
  
  \item $t = \termsubs{x}{u}{r}$. By definition $\pnrml{t}$ implies $\pnrml{r}$
  and $\pnbang{u}$, which in turn implies $\pnrml{u}$ as well. Then, there are
  two cases to consider for $\Phi$:
  \begin{enumerate}
    \item if $\Phi$ ends with rule $\ruleTESubs$, then $\Gamma =
    \ctxtsum{(\ctxtres{\Gamma'}{x}{})}{\Delta}{}$, $\cbeta = \cbeta' +
    \cbeta''$, $\cexp = \cexp' + \cexp''$, $\csize = \csize' + \csize''$,
    $\sequT{\Gamma'}{\assign{r}{\sigma}}{\cbeta'}{\cexp'}{\csize'}$,
    $\sequT{\Delta}{\assign{u}{\typeneutral}}{\cbeta''}{\cexp''}{\csize''}$ and
    $\ptight{\Gamma'(x)}$. Since $\Gamma$ is tight and $\ptight{\Gamma'(x)}$,
    then $\Gamma'$ and $\Delta$ are both tight as well. Then, the \ih with
    $\pnrml{r}$ and $\pnrml{u}$ gives $\cbeta' = \cexp' = 0$ and $\cbeta'' =
    \cexp'' = 0$. Hence, $\cbeta = \cexp = 0$.
    
    \item if $\Phi$ ends with rule $\ruleDESubs$, then $\cbeta = \cbeta' +
    \cbeta''$, $\cexp = \cexp' + \cexp''$, $\csize = \csize' + \csize''$ and
    $\Gamma = \ctxtsum{(\ctxtres{\Gamma'}{x}{})}{\Delta}{}$ tight such that, in
    particular,
    $\derivable{\Phi_u}{\sequT{\Delta}{\assign{u}{\Gamma'(x)}}{\cbeta''}{\cexp''}{\csize''}}{\SysTight}$
    with $\Delta$ tight. Moreover, by contra-positive of Remark~\ref{r:typing}
    with $\Phi_u$, it is necessarily the case $\neg\pabs{u}$. Thus, together
    with $\pnbang{u}$ it gives $\pntrl{u}$ (\cf Remark~\ref{r:normal}). Then, by
    Lemma~\ref{l:tight-spreading-neutral}, $\Phi_u$ is a tight typing, which
    leads to a contradiction with $u$ having a multiset type. Hence, this case
    does not apply.
  \end{enumerate}
\end{itemize}
\end{proof}

\begin{lemma}
\label{l:tight-size}
If
$\derivable{\Phi}{\sequT{\Gamma}{\assign{t}{\sigma}}{0}{0}{\csize}}{\SysTight}$
is tight, then $\csize = \wsize{t}$.
\end{lemma}

\begin{proof}
By induction on $\Phi$.
\begin{itemize}
  \item $\ruleDAxiom$. Then $t = x$ and $\csize = 0 = \wsize{t}$.
  
  \item $\ruleTArrowE$. Then $t = \termapp{r}{u}$, $\sigma = \typeneutral$,
  $\Gamma = \ctxtsum{\Gamma'}{\Delta}{}$, $\csize = \csize' + \csize'' + 1$,
  $\derivable{\Phi_r}{\sequT{\Gamma'}{\assign{r}{\typeneutral}}{0}{0}{\csize'}}{\SysTight}$,
  $\derivable{\Phi_u}{\sequT{\Delta}{\assign{u}{\overline{\typeabs}}}{0}{0}{\csize''}}{\SysTight}$.
  Then, $\Gamma'$ and $\Delta$ are both tight, hence the \ih gives $\csize' =
  \wsize{r}$ and $\csize'' = \wsize{u}$. Hence, $\csize = \wsize{r} +
  \wsize{u} + 1 = \wsize{t}$.
  
  \item $\ruleBTApp$. This case does not apply since it concludes
  with $\cbeta > 0$.
  
  \item $\ruleDArrowE$. This case does not apply since it concludes
  with $\cbeta > 0$.
  
  \item $\ruleTArrowI$. Then $t = \termabs{x}{u}$, $\sigma = \typeabs$, $\Gamma
  = \ctxtres{\Gamma'}{x}{}$, $\csize = \csize' + 1$,
  $\derivable{\Phi_u}{\sequT{\Gamma'}{\assign{u}{\typetight}}{0}{0}{\csize'}}{\SysTight}$
  and $\ptight{\Gamma'(x)}$. Since $\Gamma$ is tight and $\ptight{\Gamma'(x)}$,
  $\Gamma'$ is tight as well. Then, by \ih $\csize' = \wsize{u}$, which implies
  $\csize = \wsize{u} + 1 = \wsize{t}$.
  
  \item $\ruleDArrowI$. Then $\sigma = \functtype{\M}{\tau}$ which
  contradicts the hypothesis of $\Phi$ being tight. Hence, this case does not
  apply.
  
  \item $\ruleTBang$. Then $t = \termbang{u}$ and $\csize = 0 = \wsize{t}$.
  
  \item $\ruleDBang$. This case does not apply since it concludes with
  $\cexp > 0$.
  
  \item $\ruleTDer$. Then $t = \termder{u}$, $\sigma = \typeneutral$,
  $\csize = \csize' + 1$ and
  $\derivable{\Phi_u}{\sequT{\Gamma}{\assign{u}{\typeneutral}}{0}{0}{\csize'}}{\SysTight}$.
  By \ih $\csize' = \wsize{u}$, which implies $\csize = \wsize{u} + 1 =
  \wsize{t}$.
  
  \item $\ruleDDer$. Then $t = \termder{u}$ and
  $\derivable{\Phi_u}{\sequT{\Gamma}{\assign{u}{\intertype{\sigma}{}}}{0}{0}{\csize}}{\SysTight}$.
  Then, Lemma~\ref{l:tight-spreading-czero} on $\Phi_u$ gives
  $\intertype{\sigma}{} \in \typetight$ which is clearly a contradiction. Thus,
  this case does not apply.
  
  \item $\ruleTESubs$. Then $t = \termsubs{x}{u}{r}$, $\Gamma =
  \ctxtsum{(\ctxtres{\Gamma'}{x}{})}{\Delta}{}$, $\csize = \csize' + \csize''$,
  $\derivable{\Phi_r}{\sequT{\Gamma'}{\assign{r}{\sigma}}{0}{0}{\csize'}}{\SysTight}$,
  $\derivable{\Phi_u}{\sequT{\Delta}{\assign{u}{\typeneutral}}{0}{0}{\csize''}}{\SysTight}$
  and $\ptight{\Gamma'(x)}$. Since $\Gamma$ is tight and $\ptight{\Gamma'(x)}$,
  then $\Gamma'$ and $\Delta$ are both tight as well. Thus, \ih gives
  $\csize' = \wsize{r}$ and $\csize'' = \wsize{u}$. Then, $\csize = \wsize{r} +
  \wsize{u}  = \wsize{t}$. 
  
  \item $\ruleDESubs$. Then $t = \termsubs{x}{u}{r}$, $\csize = \csize' +
  \csize''$ and $\Gamma = \ctxtsum{(\ctxtres{\Gamma'}{x}{})}{\Delta}{}$ tight
  such that, in particular,
  $\derivable{\Phi_u}{\sequT{\Delta}{\assign{u}{\Gamma'(x)}}{0}{0}{\csize''}}{\SysTight}$
  with $\Delta$ tight. By Lemma~\ref{l:tight-spreading-czero} on $\Phi_u$,
  $\Gamma'(x) \in \typetight$ which leads to a contradiction, since
  $\Gamma'(x)$ is a multiset type by definition. Hence, this case does not
  apply.
\end{itemize}
\end{proof}

As well as $\SysBang$-typability, $\SysTight$-typability  of a term may provide
additional information about the neutrality/normality of its subterms:

\begin{lemma}\mbox{}
  \label{l:clashfree-tight}
  Let $u\in \TermExplicit$:
\begin{enumerate}
  \item If $\pnabs{t}$ and $\termapp{t}{u}$ is $\SysTight$-typable, then
  $\pntrl{t}$.
  \item If $\pnbang{t}$ and $\termsubs{x}{t}{u}$ is $\SysTight$-typable, then
  $\pntrl{t}$.
  \item If $\pnbang{t}$ and $\termder{t}$ is $\SysTight$-typable, then
  $\pntrl{t}$.
  \item If $\pnbang{t}$ and $\termapp{u}{t}$ is $\SysTight$-typable, then
  $\pntrl{t}$.
  \item If $\pnrml{t}$ and $\termapp{u}{t}$ is $\SysTight$-typable, then
  $\pnabs{t}$.
\end{enumerate}
\end{lemma}

\begin{proof}
Straightforward case analysis using the characterisation in the proof of
Proposition~\ref{prop:normal} and resorting to Remark~\ref{r:typing}.
Notice that a similar property was shown for $\SysBang$-typability (Lemma~\ref{l:clashfree-bang}).
\end{proof}

As well as the type system $\SysBang$, the type system $\SysTight$ captures
clash-freeness of normal terms:

\begin{theorem}
\label{t:clashfree-tight}
Let $t \in \TermExplicit$. Then, $\cfnrml{t}$ iff
$\pnrml{t}$ and $t$ is $\SysTight$-typable.
\end{theorem}

\begin{proof}
A  similar property was shown for $\SysBang$-typability (Theorem~\ref{t:clashfree-bang}). This proof is analogous to that one, but now using
  Lemma~\ref{l:clashfree-tight}. Notice that the consuming rules of
  system $\SysTight$ are essentially the typing rules of system
  $\SysBang$.
\end{proof}

As usual, in order to prove soundness, the key property is subject reduction,
stating that every reduction step decreases one of the first two counters of
tight derivations by exactly one. We first prove a substitution lemma.

\begin{lemma}[Substitution]
\label{l:substitution-tight}
Let us consider
$\derivable{\Phi_{t}}{\sequT{\Gamma; \assign{x}{\intertype{\sigma_i}{i \in I}}}{\assign{t}{\tau}}{\cbeta}{\cexp}{\csize}}{\SysTight}$
and derivations
$\many{\derivable{\Phi^{i}_{u}}{\sequT{\Delta_i}{\assign{u}{\sigma_i}}{\cbeta_i}{\cexp_i}{\csize_i}}{\SysTight}}{i \in I}$,
then there exists a derivation
of the form 
$\derivable{\Phi_{\substitute{x}{u}{t}}}{\sequT{\ctxtsum{\Gamma}{\Delta_i}{i \in I}}{\assign{\substitute{x}{u}{t}}{\tau}}{\cbeta +_{\iI}{\cbeta_i}}{\cexp +_{\iI}{\cexp_i}}{\csize +_{\iI}{\csize_i}}}{\SysTight}$.
\end{lemma}

\begin{proof}
Straightforward induction on $\Phi_t$.  The detailed proof of few chosen cases
follow. Suppose that the last rule of  $\Phi_t$ is:
\begin{itemize}
  \item $\ruleDAxiom$. Then $t = y$ and  $\Gamma = y:\intertype{\tau}{}$. If
  $x \neq y$ then $I=\emptyset$, and the required typing is $\Phi_{t}$. If $x =
  y$ then $I$ is a singleton $\{*\}$ and the required typing is $\Phi^{*}_{u}$.
  The counters are as expected since $\cbeta = \cexp = \csize = 0$.
  
  \item $\ruleBTApp$. Then $t = \termapp{t_1}{t_2}$ with
  $\derivable{\Phi_{t_1}}{\sequT{\Gamma_1;\assign{x}{\intertype{\sigma_i}{i \in I_1}}}{\assign{t_1}{\functtype{\M}{\tau}}}{\cbeta_1}{\cexp_1}{\csize_1}}{\SysTight}$
  and
  $\derivable{\Phi_{t_2}}{\sequT{\Gamma_2; \assign{x}{\intertype{\sigma_i}{i \in I_2}}}{\assign{t_2}{\typeneutral}}{\cbeta_2}{\cexp_2}{\csize_2}}{\SysTight}$
  such that $\Gamma = \ctxtsum{\Gamma_1}{\Gamma_2}{}$, $I = I_1 \uplus I_2$,
  $\cbeta = \cbeta_1 + \cbeta_2 + 1$, $\cexp = \cexp_1 + \cexp_2$,
  $\csize =
  \csize_1 + \csize_2$, and $\M$ is tight. The \ih provides the typings 
  $\derivable{\Phi_{\substitute{x}{u}{t_1}}}{\sequT{\ctxtsum{\Gamma_1}{\Delta_i}{i \in I_1}}{\assign{\substitute{x}{u}{t_1}}{\functtype{\M}{\tau}}}{\cbeta_1 +_{\iI_1}{\cbeta_i}}{\cexp_1 +_{\iI_1}{\cexp_i}}{\csize_1 +_{\iI_1}{\csize_i}}}{\SysTight}$
  and
  $\derivable{\Phi_{\substitute{x}{u}{t_2}}}{\sequT{\ctxtsum{\Gamma_2}{\Delta_i}{i \in I_2}}{\assign{\substitute{x}{u}{t_2}}{\typeneutral }}{\cbeta_2 +_{\iI_2}{\cbeta_i}}{\cexp_2 +_{\iI_2}{\cexp_i}}{\csize_2 +_{\iI_2}{\csize_i}}}{\SysTight}$.
  The required typing is obtained by applying the rule $\ruleBTApp$ to these,
  using the fact that $\M$ is tight. The counters are as expected. 
  
  \item $\ruleDESubs$. Then $ t = \termsubs{y}{t_2}{t_1}$, and we have
  $\derivable{\Phi_{t_1}}{\sequT{\Gamma_1; \assign{x}{\intertype{\sigma_i}{i \in I_1}}}{\assign{t_1}{\tau}}{\cbeta_1}{\cexp_1}{\csize_1}}{\SysTight}$
  and
  $\derivable{\Phi_{t_2}}{\sequT{\Gamma_2; \assign{x}{\intertype{\sigma_i}{i \in I_2}}}{\assign{t_2}{\Gamma_1(y)}}{\cbeta_2}{\cexp_2}{\csize_2}}{\SysTight}$
  such that $\Gamma = \ctxtsum{(\ctxtres{\Gamma_1}{y}{})}{\Gamma_2}{}$, $I = I_1
  \uplus I_2$, $\cbeta = \cbeta_1 + \cbeta_2$, $\cexp = \cexp_1 + \cexp_2$,
  $\csize = \csize_1 + \csize_2$. The  \ih provides the typings 
  $\derivable{\Phi_{\substitute{x}{u}{t_1}}}{\sequT{\ctxtsum{\Gamma_1}{\Delta_i}{i \in I_1}}{\assign{\substitute{x}{u}{t_1}}{\tau}}{\cbeta_1 +_{\iI_1}{\cbeta_i}}{\cexp_1 +_{\iI_1}{\cexp_i}}{\csize_1 +_{\iI_1}{\csize_i}}}{\SysTight}$
  and 
  $\derivable{\Phi_{\substitute{x}{u}{t_2}}}{\sequT{\ctxtsum{\Gamma_2}{\Delta_i}{i \in I_2}}{\assign{\substitute{x}{u}{t_2}}{ \Gamma_1(y)}}{\cbeta_2 +_{\iI_2}{\cbeta_i}}{\cexp_2 +_{\iI_2}{\cexp_i}}{\csize_2 +_{\iI_2}{\csize_i}}}{\SysTight}$
  The required typing is obtained by  applying the rule $\ruleDESubs$ to
  these. The counters are as expected. 
\end{itemize}
\end{proof}

The goal of exact subject reduction is to show that tight derivations are
preserved by reduction. To apply the \ih on a sub-derivation of the original
tight type derivation, one would need this sub-derivation to be  also tight.
However, tightness is a global property not necessarily true for all
sub-derivations. A subtle property is then needed, whose precise formulation
uses an idea in~\cite{AccattoliGK18}: the original typed term $t$ is required
not to be an abstraction-like term, or tightly typable. This is sufficient to
show the desired property. 
Moreover, subject
reduction for the system $\SysTight$ proceeds by induction on the definition of
$t \lored t'$, and in the three base cases of the recursive definition of
$t\lored t'$ the list $\ctxt{L}$ has arbitrary length. Therefore, for each base
case there is a further induction on the length of
$\ctxt{L}$.
Formally,

\begin{lemma}[Exact Subject Reduction]
\label{l:QSR}
Let
$\derivable{\Phi}{\sequT{\Gamma}{\assign{t}{\sigma}}{\cbeta}{\cexp}{\csize}}{\SysTight}$
such that $\Gamma$ is tight, and either
$\sigma \in \typetight$ or $\neg\pabs{t}$.
If $t \lored t'$, then there exists
$\derivable{\Phi'}{\sequT{\Gamma}{\assign{t'}{\sigma}}{\cbeta'}{\cexp'}{\csize}}{\SysTight}$
such that 
\begin{itemize}
\item $\cbeta' = \cbeta - 1$ and $\cexp' = \cexp$ if $t \lored t'$ is an $\mStep$-step.
\item $\cexp' = \cexp - 1$ and $\cbeta' = \cbeta$ if $t \lored t'$ is an $\eStep$-step.
\end{itemize}
\end{lemma}

\begin{proof}
By induction on $t \lored t'$.
\begin{itemize}
  \item $t = \termapp{\ctxtapp{\ctxt{L}}{\termabs{x}{u}}}{r} \lored
  \ctxtapp{\ctxt{L}}{\termsubs{x}{r}{u}} = t'$. We reason by induction on
  $\ctxt{L}$.
  \begin{itemize}
    \item $\ctxt{L} = \Box$. We first note that $\Phi$ cannot end with rule
    $\ruleTArrowE$ since $\termabs{x}{u}$ cannot be typed with $\typeneutral$,
    then there are two cases depending on the last rule of $\Phi$.
    \begin{enumerate}
      \item If $\Phi$ has the following form: \[
\Rule{
  \Rule{
    \sequT{\Gamma_u;\assign{x}{\M}}{\assign{u}{\sigma}}{\cbeta_u}{\cexp_u}{\csize_u}
  }{
    \sequT{\Gamma_u}{\assign{\termabs{x}{u}}{\functtype{\M}{\sigma}}}{\cbeta_u}{\cexp_u}{\csize_u}
  }{\ruleDArrowI}
  \quad
  \sequT{\Gamma_r}{\assign{r}{\M}}{\cbeta_r}{\cexp_r}{\csize_r}
}{
  \sequT{\ctxtsum{\Gamma_u}{\Gamma_r}{}}{\assign{\termapp{(\termabs{x}{u})}{r}}{\sigma}}{\cbeta_u+\cbeta_r+1}{\cexp_u+\cexp_r}{\csize_u+\csize_r}
}{\ruleDArrowE}
      \] We can then construct the following derivation $\Phi'$. \[
\Rule{
  \sequT{\Gamma_u;\assign{x}{\M}}{\assign{u}{\sigma}}{\cbeta_u}{\cexp_u}{\csize_u}
  \quad
  \sequT{\Gamma_r}{\assign{r}{\M}}{\cbeta_r}{\cexp_r}{\csize_r}
}{
  \sequT{\ctxtsum{\Gamma_u}{\Gamma_r}{}}{\assign{\termsubs{x}{r}{u}}{\sigma}}{\cbeta_u+\cbeta_r}{\cexp_u+\cexp_r}{\csize_u+\csize_r}
}{\ruleDESubs}
      \] The counters verify the expected property. 

      \item If $\Phi$ has the following form: \[\kern-2em
\Rule{
  \Rule{
    \sequT{\Gamma_u;\assign{x}{\M}}{\assign{u}{\sigma}}{\cbeta_u}{\cexp_u}{\csize_u}
  }{
    \sequT{\Gamma_u}{\assign{\termabs{x}{u}}{\functtype{\M}{\sigma}}}{\cbeta_u}{\cexp_u}{\csize_u}
  }{\ruleDArrowI}
  \quad
  \sequT{\Gamma_{r}}{\assign{r}{\typeneutral}}{\cbeta_r}{\cexp_r}{\csize_r}
  \quad
  \ptight{\M}
}{
  \sequT{\ctxtsum{\Gamma_u}{\Gamma_r}{}}{\assign{\termapp{(\termabs{x}{u})}{r}}{\sigma}}{\cbeta_u+\cbeta_r+1}{\cexp_u+\cexp_r}{\csize_u+\csize_r}
}{\ruleBTApp}
      \] We can then construct the following derivation $\Phi'$. \[
\Rule{
  \sequT{\Gamma_u;\assign{x}{\M}}{\assign{u}{\sigma}}{\cbeta_u}{\cexp_u}{\csize_u}
  \quad
  \sequT{\Gamma_r}{\assign{r}{\typeneutral}}{\cbeta_r}{\cexp_r}{\csize_r}
  \quad
  \ptight{\M}
}{
  \sequT{\ctxtsum{\Gamma_u}{\Gamma_{r}}{}}{\assign{\termsubs{x}{r}{u}}{\sigma}}{\cbeta_u+\cbeta_r}{\cexp_u+\cexp_r}{\csize_u+\csize_r}
}{\ruleTESubs}
      \] The counters verify the expected property.
    \end{enumerate}
    
    \item $\ctxt{L} = \termsubs{y}{s}{\ctxt{L'}}$. Immediate from the \ih
  \end{itemize}
  
  \item $t = \termsubs{x}{\ctxtapp{\ctxt{L}}{\termbang{r}}}{u}
  \lored \ctxtapp{\ctxt{L}}{\substitute{x}{r}{u}} = t'$. We reason by induction
  on $\ctxt{L}$.
  \begin{itemize}
    \item $\ctxt{L} = \Box$. We first note that $\Phi$ cannot end with
    rule $\ruleTESubs$ since $\termbang{r}$ cannot be typed with
    $\typeneutral$, then $\Phi$ has the following form: \[\kern-3em
\Rule{
  \sequT{\Gamma_u; x: \intertype{\sigma_i}{\iI}}{\assign{u}{\sigma}}{\cbeta_u}{\cexp_u}{\csize_u}
  \quad
  \Rule{(\sequT{\Gamma_i}{\assign{r}{\sigma_i}}{\cbeta_i}{\cexp_i}{\csize_i})_{\iI}}{\sequT{\ctxtsum{}{\Gamma_i}{i \in I}}{\assign{\termbang{r}}{\intertype{\sigma_i}{\iI}}}{+_{\iI} \cbeta_i}{1 +_{\iI} \cexp_i}{+_{\iI}\csize_i}}{\ruleDBang}
}{
  \sequT{\ctxtsum{\Gamma_u}{\Gamma_i}{i \in I}}{\assign{\termsubs{x}{\termbang{r}}{u}}{\sigma}}{\cbeta_u +_{\iI} \cbeta_i}{\cexp_u +1 +_{\iI} \cexp_i }{\csize_u +_{\iI} \csize_i}
}{\ruleDESubs}
    \] By applying Lemma~\ref{l:substitution-tight} to the premises we obtain a
    derivation \[
\derivable{\Phi'}{\sequT{\ctxtsum{\Gamma_u}{\Gamma_i}{i \in I}}{\assign{\substitute{x}{r}{u}}{\sigma}}{\cbeta_u+_{\iI}{\cbeta_i}}{\cexp_u+_{\iI}{\cexp_i}}{\csize_u+_{\iI}{\csize_i}}}{}
    \] The counters verify the expected property.
    
    \item $\ctxt{L}= \termsubs{y}{s}{\ctxt{L'}}$. Immediate from the \ih
  \end{itemize}

  \item $t = \termder{(\ctxtapp{\ctxt{L}}{\termbang{u}})} \lored
  \ctxtapp{\ctxt{L}}{u}=t'$. We reason by induction on $\ctxt{L}$.
  \begin{itemize}
    \item $\ctxt{L} = \Box$. Then $\Phi$ has necessarily the following form: \[
\Rule{
  \Rule{
    \sequT{\Gamma_u}{\assign{u}{\sigma}}{\cbeta_u}{\cexp_u}{\csize_u}
  }{
    \sequT{\Gamma_u}{\assign{\termbang{u}}{\intertype{\sigma}{}}}{\cbeta_u}{1+\cexp_u}{\csize_u}
  }{\ruleDBang}
}{
  \sequT{\Gamma_u}{\assign{\termder{\termbang{u}}}{\sigma}}{\cbeta_u}{1+\cexp_u}{\csize_u}
}{\ruleDDer}
    \] We conclude with the derivation $\Phi'$ given by the premise. The
    counters verify the expected property.
    
    \item $\ctxt{L} = \termsubs{y}{s}{\ctxt{L'}}$. Immediate from the \ih
  \end{itemize}
  
  \item All the inductive cases for internal reductions are straightforward by \ih
\end{itemize}  
\end{proof}

\begin{theorem}[Soundness for System $\SysTight$]
\label{t:correctness}
If
$\derivable{\Phi}{\sequT{\Gamma}{\assign{t}{\sigma}}{\cbeta}{\cexp}{\csize}}{\SysTight}$
is tight, then there exists $p$ such that $\cfnrml{p}$ and
$t \rewriten{\bangweak}^{(\cbeta,\cexp)} p$ with $\cbeta$ $\mStep$-steps,
$\cexp$ $\eStep$-steps, and $\wsize{p} = \csize$.
\end{theorem}

\begin{proof}
We prove the statement by showing that $t \loredn^{(\cbeta,\cexp)} p$ holds for
the deterministic strategy, then we conclude since $\lored \mathbin{\subseteq}
\rewrite{\bangweak}$. Let
$\derivable{\Phi}{\sequT{\Gamma}{\assign{t}{\sigma}}{\cbeta}{\cexp}{\csize}}{\SysTight}$.
We reason by induction on $\cbeta+\cexp$.

If $\cbeta+\cexp = 0$, then $\cbeta = \cexp = 0$ and Lemma~\ref{l:czero-normal}
gives $\pnrml{t}$. Moreover, by Lemma~\ref{l:tight-size} and
Theorem~\ref{t:clashfree-tight}, we get both $\wsize{t} = \csize$ and
$\cfnrml{t}$. Thus, we conclude with $p = t$.

If $\cbeta+\cexp > 0$, then $t \notin \ipnrml$ holds by
Lemma~\ref{l:czero-normal} and thus there exists $t'$ such that $t \loredn^{(1,0)} t'$
or $t \loredn^{(0,1)} t'$
by Proposition~\ref{l:equivalence}. By Lemma~\ref{l:QSR} there is
$\derivable{\Phi'}{\sequT{\Gamma}{\assign{t'}{\sigma}}{\cbeta'}{\cexp'}{\csize}}{\SysTight}$
such that $1 + \cbeta' + \cexp' = \cbeta + \cexp$. By the \ih there is $p$ such
that $\cfnrml{p}$ and $t' \loredn^{(\cbeta',\cexp')} p$
with $\csize = \wsize{p}$. Then $t \loredn^{(1,0)} t' \loredn^{(\cbeta',\cexp')} p$
(resp. $t \loredn^{(0,1)} t' \ldots$)
which means $t \loredn^{(\cbeta,\cexp)} p$, as expected.
\end{proof}


\paragraph{{\bf Completeness}}
\label{s:completeness}
We now study completeness of the type system $\SysTight$, which does not only
guarantee that normalising terms are typable --a qualitative property-- but
also provides a tight type derivation having appropriate counters. More
precisely, given a term $t$ which is $\bangweak$-normalisable by means of
$\cbeta$ $\dBeta$-steps and $\cexp$ 
$\{\sBang, \dBang\}$-steps,
and having a
$\bangweak$-normal form of $\bangweak$-size $\csize$, there is a tight
derivation $\Phi$ for $t$ with counters $(\cbeta,\cexp,\csize)$. The
completeness proof is mainly based on a subject expansion property
(Lemma~\ref{l:QSE-tight}), as well as on an auxiliary lemma providing tight
derivations with appropriate counters for $\bangweak$-normal weak clash-free
terms.

\begin{lemma}
\label{l:typing-normal-clash-free}
If $\cfnrml{t}$, then there is a tight derivation
$\derivable{\Phi}{\sequT{\Gamma}{\assign{t}{\sigma}}{0}{0}{\wsize{t}}}{\SysTight}$.
\end{lemma}

\begin{proof}
We proceed by induction on the derivation of  $\cfnrml{t}$ (resp. $\cfntrl{t}$,
$\cfnbang{t}$ and $\cfnabs{t}$). Moreover, we generalise the statement and
simultaneously show the following:
\begin{enumerate}
  \item\label{l:typing-normal-clash-free:ne} If $\cfntrl{t}$, then $\sigma = \typeneutral$; 
  \item\label{l:typing-normal-clash-free:nb} If $\cfnbang{t}$, then $\sigma \in \overline{\typebang}$;
  \item\label{l:typing-normal-clash-free:na} If $\cfnabs{t}$, then $\sigma \in \overline{\typeabs}$.
\end{enumerate}

We analyse the possible cases for $t$.
\begin{itemize}
  \item Let $\cfntrl{x}$, so that $\wsize{x} = 0$. Then, take
  $\derivable{}{\sequT{\assign{x}{\multiset{\typeneutral}}}{\assign{x}{\typeneutral}}{0}{0}{0}}{\SysTight}$
  to conclude.

  \item Let $\cfntrl{\termder{t}}$ with $\cfntrl{t}$. Then, apply the
  \ih (\ref{l:typing-normal-clash-free:ne}) to obtain
  $\derivable{}{\sequT{\Gamma}{\assign{t}{\typeneutral}}{0}{0}{\wsize{t}}}{\SysTight}$
  and conclude 
  $\derivable{}{\sequT{\Gamma}{\assign{\termder{t}}{\typeneutral}}{0}{0}{\wsize{t} + 1}}{\SysTight}$
  with rule $\ruleTDer$. The size is as expected since $\wsize{\termder{t}} =
  \wsize{t} + 1$.
  
  \item Let $\cfntrl{\termapp{t}{u}}$ with $\cfntrl{t}$ and $\cfnabs{u}$.
  Then, apply the \ih (\ref{l:typing-normal-clash-free:ne}) and
  (\ref{l:typing-normal-clash-free:na}) to obtain
  $\derivable{}{\sequT{\Gamma}{\assign{t}{\typeneutral}}{0}{0}{\wsize t}}{\SysTight}$
  and
  $\derivable{}{\sequT{\Delta}{\assign{u}{\overline \typeabs}}{0}{0}{\wsize{u}}}{\SysTight}$
  resp., and conclude
  $\derivable{}{\sequT{\ctxtsum{\Gamma}{\Delta}{}}{\assign{\termapp{t}{u}}{\typeneutral}}{0}{0}{\wsize{t} + \wsize{u} + 1}}{\SysTight}$
  with rule $\ruleTArrowE$. The size is as expected since
  $\wsize{\termapp{t}{u}} = \wsize{t} + \wsize{u} + 1$.
  
  \item Let $\cfntrl{\termsubs{x}{u}{t}}$ with $\cfntrl{t}$ and $\cfntrl{u}$,
  or $\cfnabs{\termsubs{x}{u}{t}}$ with $\cfnabs{t}$ and $\cfntrl{u}$, or
  $\cfnbang{\termsubs{x}{u}{t}}$ with $\cfnbang{t}$ and $\cfntrl{u}$.
  Then, apply the appropriate \ih to obtain
  $\derivable{}{\sequT{\Gamma}{\assign{t}{\typetight}}{0}{0}{\wsize{t}}}{\SysTight}$
  and
  $\derivable{}{\sequT{\Delta}{\assign{u}{\typeneutral}}{0}{0}{\wsize{u}}}{\SysTight}$,
  and conclude
  $\derivable{}{\sequT{\ctxtsum{(\ctxtres{\Gamma}{x}{})}{\Delta}{}}{\assign{\termsubs{x}{u}{t}}{\typetight}}{0}{0}{\wsize{t} + \wsize{u}}}{\SysTight}$
  by rule $\ruleTESubs$ since $\Gamma(x)$ is tight. The size is as expected
  since $\wsize{\termsubs{x}{u}{t}} = \wsize{t} + \wsize{u}$.
  
  \item Let $\cfnabs{\termbang{t}}$. Then, conclude
  $\derivable{}{\sequT{}{\assign{\termbang{t}}{\typebang}}{0}{0}{0}}{\SysTight}$
  by $\ruleTBang$. The size is as expected.
  
  \item Let $\cfnbang{\termabs{x}{t}}$ with $\cfnrml{t}$. Then, apply the
  appropriate \ih to obtain
  $\derivable{}{\sequT{\Gamma}{\assign{t}{\typetight}}{0}{0}{\wsize{t}}}{\SysTight}$
  and conclude
  $\derivable{}{\sequT{\ctxtres{\Gamma}{x}{}}{\assign{\termabs{x}{t}}{\typeabs}}{0}{0}{\wsize{t} + 1}}{\SysTight}$
  by $\ruleTArrowI$ since $\Gamma(x)$ is tight.
\end{itemize}
\end{proof}

\begin{lemma}[Anti-Substitution]
\label{l:bang:anti-substitution-tight}
If $\derivable{\Phi_{\substitute{x}{u}{t}}}{\sequT{\Gamma'}{\assign{\substitute{x}{u}{t}}{\tau}}{\cbeta'}{\cexp'}{\csize'}}{\SysTight}$,
then there exists 
$\derivable{\Phi_t}{\sequT{\Gamma; \assign{x}{\intertype{\sigma_i}{i \in I}}}{\assign{t}{\tau}}{\cbeta}{\cexp}{\csize}}{\SysTight}$
and
$\many{\derivable{\Phi^i_u}{\sequT{\Delta_i}{\assign{u}{\sigma_i}}{\cbeta_i}{\cexp_i}{\csize_i}}{\SysTight}}{i \in I}$
such that $\Gamma' = \ctxtsum{\Gamma}{\Delta_i}{i \in I}$,
$\cbeta' = \cbeta +_{\iI}{\cbeta_i}$, $\cexp' = \cexp +_{\iI}{\cexp_i}$ and
$\csize' = \csize +_{\iI}{\csize_i}$.
\end{lemma}

\begin{proof}
By induction on $t$.
\begin{itemize}
  \item $t = x$. Then, $\substitute{x}{u}{t} = u$ and we set $I = \{\ast\}$,
  $\sigma_\ast = \tau$, $\Gamma = \emptyset$, $\Delta_\ast = \Gamma'$, $\Phi^{\ast}_{u} =
  \Phi_{\substitute{x}{u}{t}}$, and
  $\derivable{\Phi_{t}}{\sequT{\assign{x}{\intertype{\tau}{}}}{\assign{x}{\tau}}{0}{0}{0}}{\SysTight}$
  by rule $\ruleDAxiom$.
  
  \item $t = y \neq x$. Then, $\substitute{x}{u}{t} = y$ and we conclude with
  $I = \emptyset$ (hence, $\intertype{\sigma_i}{i \in I} = \emul$), $\Gamma =
  \Gamma'$ and $\Phi_{t} = \Phi_{\substitute{x}{u}{t}}$. 
  
  \item $t = \termapp{t_1}{t_2}$. Then, $\substitute{x}{u}{t} =
  \termapp{\substitute{x}{u}{t_1}}{\substitute{x}{u}{t_2}}$ and there are three
  possible cases.
  \begin{enumerate}
    \item $\ruleTArrowE$. Then, $\tau = \typeneutral$,  $\Gamma' =
    \ctxtsum{\Gamma'_1}{\Gamma'_2}{}$, $\cbeta' = \cbeta'_1 + \cbeta'_2$,
    $\cexp' = \cexp'_1 + \cexp'_2$ and $\csize' = \csize'_1 + \csize'_2 + 1$
    with premises
    $\derivable{\Phi_{\substitute{x}{u}{t_1}}}{\sequT{\Gamma'_1}{\assign{\substitute{x}{u}{t_1}}{\typeneutral}}{\cbeta'_1}{\cexp'_1}{\csize'_1}}{\SysTight}$
    and
    $\derivable{\Phi_{\substitute{x}{u}{t_2}}}{\sequT{\Gamma'_2}{\assign{\substitute{x}{u}{t_2}}{\overline{\typeabs}}}{\cbeta'_2}{\cexp'_2}{\csize'_2}}{\SysTight}$.
    By \ih on both, there exist type derivations
    $\derivable{\Phi_{t_1}}{\sequT{\Gamma_1;\assign{x}{\intertype{\sigma_i}{i \in I_1}}}{\assign{t_1}{\typeneutral}}{\cbeta_1}{\cexp_1}{\csize_1}}{\SysTight}$, \ \ 
    $\many{\derivable{\Phi^{i}_{u}}{\sequT{\Delta_i}{\assign{u}{\sigma_i}}{\cbeta_i}{\cexp_i}{\csize_i}}{\SysTight}}{i \in I_1}$, \ \ 
    $\derivable{\Phi_{t_2}}{\sequT{\Gamma_2;\assign{x}{\intertype{\sigma_i}{i \in I_2}}}{\assign{t_2}{\overline{\typeabs}}}{\cbeta_2}{\cexp_2}{\csize_2}}{\SysTight}$ \ \
    and
    $\many{\derivable{\Phi^{i}_{u}}{\sequT{\Delta_i}{\assign{u}{\sigma_i}}{\cbeta_i}{\cexp_i}{\csize_i}}{\SysTight}}{i \in I_2}$
    (note that $I_1$ and $I_2$ can be assumed disjoint) such that $\Gamma'_1 =
    \ctxtsum{\Gamma_1}{\Delta_i}{i \in I_1}$, $\Gamma'_2 =
    \ctxtsum{\Gamma_2}{\Delta_i}{i \in I_2}$, $\cbeta'_1 = \cbeta_1
    +_{i \in I_1}{\cbeta_i}$, $\cbeta'_2 = \cbeta_2 +_{i \in I_2}{\cbeta_i}$,
    $\cexp'_1 = \cexp_1 +_{i \in I_1}{\cexp_i}$, $\cexp'_2 = \cexp_2
    +_{i \in I_2}{\cexp_i}$, $\csize'_1 = \csize_1 +_{i \in I_1}{\csize_i}$ and
    $\csize'_2 = \csize_2 +_{i \in I_2}{\csize_i}$. Then, by $\ruleTArrowE$ we
    have
    $\derivable{\Phi_{t}}{\sequT{\ctxtsum{\Gamma_1}{\Gamma_2}{};\assign{x}{\intertype{\sigma_i}{i \in I_1 \cup I_2}}}{\assign{\termapp{t_1}{t_2}}{\typeneutral}}{\cbeta_1+\cbeta_2}{\cexp_1+\cexp_2}{\csize_1+\csize_2+1}}{\SysTight}$
    and we conclude with $\Gamma = \ctxtsum{\Gamma_1}{\Gamma_2}{}$, $I = I_1
    \cup I_2$, $\cbeta = \cbeta_1 + \cbeta_2$, $\cexp = \cexp_1 + \cexp_2$ and
    $\csize = \csize_1 + \csize_2 + 1$ since:
    \begin{itemize}
       \item $\Gamma' = \ctxtsum{\Gamma'_1}{\Gamma'_2}{} =
       \ctxtsum{\ctxtsum{\Gamma_1}{\Delta_i}{i \in I_1}}{\ctxtsum{\Gamma_2}{\Delta_i}{i \in I_2}}{}
       = \ctxtsum{\Gamma}{\Delta_i}{i \in I}$.
       
       \item $\cbeta' = \cbeta'_1 + \cbeta'_2 = \cbeta_1
       +_{i \in I_1}{\cbeta_i} + \cbeta_2 +_{i \in I_2}{\cbeta_i} = \cbeta
       +_{i \in I}{\cbeta_i}$.
       
       \item $\cexp' = \cexp'_1 + \cexp'_2 = \cexp_1
       +_{i \in I_1}{\cexp_i} + \cexp_2 +_{i \in I_2}{\cexp_i} = \cexp
       +_{i \in I}{\cexp_i}$.
       
       \item $\csize' = \csize'_1 + \csize'_2 + 1 = \csize_1
       +_{i \in I_1}{\csize_i} + \csize_2 +_{i \in I_2}{\csize_i} + 1 = \csize
       +_{i \in I}{\csize_i}$.
    \end{itemize}
    
    \item $\ruleBTApp$. Then, $\Gamma' = \ctxtsum{\Gamma'_1}{\Gamma'_2}{}$,
    $\cbeta' = \cbeta'_1 + \cbeta'_2 + 1$, $\cexp' = \cexp'_1 + \cexp'_2$ and
    $\csize' = \csize'_1 + \csize'_2$
    with premises
    $\derivable{\Phi_{\substitute{x}{u}{t_1}}}{\sequT{\Gamma'_1}{\assign{\substitute{x}{u}{t_1}}{\functtype{\M}{\tau}}}{\cbeta'_1}{\cexp'_1}{\csize'_1}}{\SysTight}$
    and
    $\derivable{\Phi_{\substitute{x}{u}{t_2}}}{\sequT{\Gamma'_2}{\assign{\substitute{x}{u}{t_2}}{\typeneutral}}{\cbeta'_2}{\cexp'_2}{\csize'_2}}{\SysTight}$
    for some multiset type $\M$ such that $\ptight{\M}$. By \ih there
   exist 
    $\derivable{\Phi_{t_1}}{\sequT{\Gamma_1;\assign{x}{\intertype{\sigma_i}{i \in I_1}}}{\assign{t_1}{\functtype{\M}{\tau}}}{\cbeta_1}{\cexp_1}{\csize_1}}{\SysTight}$,
    $\many{\derivable{\Phi^{i}_{u}}{\sequT{\Delta_i}{\assign{u}{\sigma_i}}{\cbeta_i}{\cexp_i}{\csize_i}}{\SysTight}}{i \in I_1}$,
    $\derivable{\Phi_{t_2}}{\sequT{\Gamma_2;\assign{x}{\intertype{\sigma_i}{i \in I_2}}}{\assign{t_2}{\typeneutral}}{\cbeta_2}{\cexp_2}{\csize_2}}{\SysTight}$
    and
    $\many{\derivable{\Phi^{i}_{u}}{\sequT{\Delta_i}{\assign{u}{\sigma_i}}{\cbeta_i}{\cexp_i}{\csize_i}}{\SysTight}}{i \in I_2}$
    (note that $I_1$ and $I_2$ can be assumed disjoint)
    such that $\Gamma'_1 = \ctxtsum{\Gamma_1}{\Delta_i}{i \in I_1}$, $\Gamma'_2
    = \ctxtsum{\Gamma_2}{\Delta_i}{i \in I_2}$, $\cbeta'_1 = \cbeta_1
    +_{i \in I_1}{\cbeta_i}$, $\cbeta'_2 = \cbeta_2 +_{i \in I_2}{\cbeta_i}$,
    $\cexp'_1 = \cexp_1 +_{i \in I_1}{\cexp_i}$, $\cexp'_2 = \cexp_2
    +_{i \in I_2}{\cexp_i}$, $\csize'_1 = \csize_1 +_{i \in I_1}{\csize_i}$ and
    $\csize'_2 = \csize_2 +_{i \in I_2}{\csize_i}$. Then, by $\ruleBTApp$ we
    have the type derivation
    $\derivable{\Phi_{t}}{\sequT{\ctxtsum{\Gamma_1}{\Gamma_2}{};\assign{x}{\intertype{\sigma_i}{i \in I_1 \cup I_2}}}{\assign{\termapp{t_1}{t_2}}{\tau}}{\cbeta_1+\cbeta_2+1}{\cexp_1+\cexp_2}{\csize_1+\csize_2}}{\SysTight}$
    and we conclude with $\Gamma = \ctxtsum{\Gamma_1}{\Gamma_2}{}$, $I = I_1
    \cup I_2$, $\cbeta = \cbeta_1 + \cbeta_2 + 1$, $\cexp = \cexp_1 + \cexp_2$
    and $\csize = \csize_1 + \csize_2$ since:
    \begin{itemize}
       \item $\Gamma' = \ctxtsum{\Gamma'_1}{\Gamma'_2}{} =
       \ctxtsum{\ctxtsum{\Gamma_1}{\Delta_i}{i \in I_1}}{\ctxtsum{\Gamma_2}{\Delta_i}{i \in I_2}}{}
       = \ctxtsum{\Gamma}{\Delta_i}{i \in I}$.
       
       \item $\cbeta' = \cbeta'_1 + \cbeta'_2 + 1 = \cbeta_1
       +_{i \in I_1}{\cbeta_i} + \cbeta_2 +_{i \in I_2}{\cbeta_i} + 1 = \cbeta
       +_{i \in I}{\cbeta_i}$.
       
       \item $\cexp' = \cexp'_1 + \cexp'_2 = \cexp_1
       +_{i \in I_1}{\cexp_i} + \cexp_2 +_{i \in I_2}{\cexp_i} = \cexp
       +_{i \in I}{\cexp_i}$.
       
       \item $\csize' = \csize'_1 + \csize'_2 = \csize_1
       +_{i \in I_1}{\csize_i} + \csize_2 +_{i \in I_2}{\csize_i} = \csize
       +_{i \in I}{\csize_i}$.
    \end{itemize}
    
    \item $\ruleDArrowE$. Then, $\Gamma' = \ctxtsum{\Gamma'_1}{\Gamma'_2}{}$,
    $\cbeta' = \cbeta'_1 + \cbeta'_2 + 1$, $\cexp' = \cexp'_1 + \cexp'_2$ and
    $\csize' = \csize'_1 + \csize'_2$ with premises
    $\derivable{\Phi_{\substitute{x}{u}{t_1}}}{\sequT{\Gamma'_1}{\assign{\substitute{x}{u}{t_1}}{\functtype{\M}{\tau}}}{\cbeta'_1}{\cexp'_1}{\csize'_1}}{\SysTight}$
    and
    $\derivable{\Phi_{\substitute{x}{u}{t_2}}}{\sequT{\Gamma'_2}{\assign{\substitute{x}{u}{t_2}}{\M}}{\cbeta'_2}{\cexp'_2}{\csize'_2}}{\SysTight}$
    for some multiset type $\M$. By \ih there exist derivations
    $\derivable{\Phi_{t_1}}{\sequT{\Gamma_1;\assign{x}{\intertype{\sigma_i}{i \in I_1}}}{\assign{t_1}{\functtype{\M}{\tau}}}{\cbeta_1}{\cexp_1}{\csize_1}}{\SysTight}$, \ \ 
    $\many{\derivable{\Phi^{i}_{u}}{\sequT{\Delta_i}{\assign{u}{\sigma_i}}{\cbeta_i}{\cexp_i}{\csize_i}}{\SysTight}}{i \in I_1}$, \ \ 
    $\derivable{\Phi_{t_2}}{\sequT{\Gamma_2;\assign{x}{\intertype{\sigma_i}{i \in I_2}}}{\assign{t_2}{\M}}{\cbeta_2}{\cexp_2}{\csize_2}}{\SysTight}$ \ \ 
    and \ \ 
    $\many{\derivable{\Phi^{i}_{u}}{\sequT{\Delta_i}{\assign{u}{\sigma_i}}{\cbeta_i}{\cexp_i}{\csize_i}}{\SysTight}}{i \in I_2}$
    (note that $I_1$ and $I_2$ can be assumed disjoint)
    such that $\Gamma'_1 = \ctxtsum{\Gamma_1}{\Delta_i}{i \in I_1}$, $\Gamma'_2
    = \ctxtsum{\Gamma_2}{\Delta_i}{i \in I_2}$, $\cbeta'_1 = \cbeta_1
    +_{i \in I_1}{\cbeta_i}$, $\cbeta'_2 = \cbeta_2 +_{i \in I_2}{\cbeta_i}$,
    $\cexp'_1 = \cexp_1 +_{i \in I_1}{\cexp_i}$, $\cexp'_2 = \cexp_2
    +_{i \in I_2}{\cexp_i}$, $\csize'_1 = \csize_1 +_{i \in I_1}{\csize_i}$ and
    $\csize'_2 = \csize_2 +_{i \in I_2}{\csize_i}$. Then, by $\ruleDArrowE$ we
    obtain the type derivation
    $\derivable{\Phi_{t}}{\sequT{\ctxtsum{\Gamma_1}{\Gamma_2}{};\assign{x}{\intertype{\sigma_i}{i \in I_1 \cup I_2}}}{\assign{\termapp{t_1}{t_2}}{\tau}}{\cbeta_1+\cbeta_2+1}{\cexp_1+\cexp_2}{\csize_1+\csize_2}}{\SysTight}$
    and we conclude with $\Gamma = \ctxtsum{\Gamma_1}{\Gamma_2}{}$, $I = I_1
    \cup I_2$, $\cbeta = \cbeta_1 + \cbeta_2 + 1$, $\cexp = \cexp_1 + \cexp_2$
    and $\csize = \csize_1 + \csize_2$ since:
    \begin{itemize}
       \item $\Gamma' = \ctxtsum{\Gamma'_1}{\Gamma'_2}{} =
       \ctxtsum{\ctxtsum{\Gamma_1}{\Delta_i}{i \in I_1}}{\ctxtsum{\Gamma_2}{\Delta_i}{i \in I_2}}{}
       = \ctxtsum{\Gamma}{\Delta_i}{i \in I}$.
       
       \item $\cbeta' = \cbeta'_1 + \cbeta'_2 + 1 = \cbeta_1
       +_{i \in I_1}{\cbeta_i} + \cbeta_2 +_{i \in I_2}{\cbeta_i} + 1 = \cbeta
       +_{i \in I}{\cbeta_i}$.
       
       \item $\cexp' = \cexp'_1 + \cexp'_2 = \cexp_1
       +_{i \in I_1}{\cexp_i} + \cexp_2 +_{i \in I_2}{\cexp_i} = \cexp
       +_{i \in I}{\cexp_i}$.
       
       \item $\csize' = \csize'_1 + \csize'_2 = \csize_1
       +_{i \in I_1}{\csize_i} + \csize_2 +_{i \in I_2}{\csize_i} = \csize
       +_{i \in I}{\csize_i}$.
    \end{itemize}
  \end{enumerate}
  
  \item $t = \termabs{y}{t'}$. By $\alpha$-conversion we assume $y \neq x$ and $y
  \notin \fv{u}$. Then, $\substitute{x}{u}{t} =
  \termabs{y}{\substitute{x}{u}{t'}}$ and there are two possible cases.
  \begin{enumerate}
    \item $\ruleTArrowI$. Then, $\tau = \typeabs$, $\Gamma' =
    \ctxtres{\Gamma'_1}{y}{}$ and $\csize' = \csize'_1 + 1$ with premise
    $\derivable{\Phi_{\substitute{x}{u}{t'}}}{\sequT{\Gamma'_1}{\assign{\substitute{x}{u}{t'}}{\typetight}}{\cbeta'}{\cexp'}{\csize'_1}}{\SysTight}$
    where $\ptight{\Gamma'_1(y)}$ holds. By \ih there exist
    $\derivable{\Phi_{t'}}{\sequT{\Gamma_1;\assign{x}{\intertype{\sigma_i}{i \in I}}}{\assign{t'}{\typetight}}{\cbeta}{\cexp}{\csize_1}}{\SysTight}$ \ \ 
    and \ \ 
    $\many{\derivable{\Phi^{i}_{u}}{\sequT{\Delta_i}{\assign{u}{\sigma_i}}{\cbeta_i}{\cexp_i}{\csize_i}}{\SysTight}}{i \in I}$ \ \ 
    such that $\Gamma'_1 = \ctxtsum{\Gamma_1}{\Delta_i}{i \in I}$, $\cbeta' =
    \cbeta +_{i \in I}{\cbeta_i}$, $\cexp' = \cexp +_{i \in I}{\cexp_i}$ and
    $\csize'_1 = \csize_1 +_{i \in I}{\csize_i}$. Moreover, $y \notin \fv{u}$
    implies $y \notin \dom{\Delta_i}$ for all $i \in I$ and, hence, $\Gamma'_1(y) =
    \Gamma_1(y)$. Then, $\ptight{\Gamma_1(y)}$ holds as well. Finally, by
    $\ruleTArrowE$ with $y \neq x$ we have
    $\derivable{\Phi_{t}}{\sequT{\ctxtres{\Gamma_1}{y}{};\assign{x}{\intertype{\sigma_i}{i \in I}}}{\assign{\termabs{y}{t'}}{\typeabs}}{\cbeta}{\cexp}{\csize_1+1}}{\SysTight}$
    and we conclude with $\Gamma = \ctxtres{\Gamma_1}{y}{}$ and $\csize =
    \csize_1 + 1$ since:
    \begin{itemize}
       \item $\Gamma' = \ctxtres{\Gamma'_1}{y}{} =
       \ctxtres{(\ctxtsum{\Gamma_1}{\Delta_i}{i \in I})}{y}{} =
       \ctxtsum{\Gamma}{\Delta_i}{i \in I}$.
       
       \item $\cbeta' = \cbeta +_{i \in I}{\cbeta_i}$.
       
       \item $\cexp' = \cexp +_{i \in I}{\cexp_i}$.
       
       \item $\csize' = \csize'_1 + 1 = \csize_1 +_{i \in I}{\csize_i} + 1 =
       \csize +_{i \in I}{\csize_i}$.
    \end{itemize}
    
    \item $\ruleDArrowI$. Then, $\tau = \functtype{\Gamma'_1(y)}{\tau'}$,
    $\Gamma' = \ctxtres{\Gamma'_1}{y}{}$ with premise
    $\derivable{\Phi_{\substitute{x}{u}{t'}}}{\sequT{\Gamma'_1}{\assign{\substitute{x}{u}{t'}}{\tau'}}{\cbeta'}{\cexp'}{\csize'}}{\SysTight}$.
    Then, by \ih there exist type derivations
    $\derivable{\Phi_{t'}}{\sequT{\Gamma_1;\assign{x}{\intertype{\sigma_i}{i \in I}}}{\assign{t'}{\tau'}}{\cbeta}{\cexp}{\csize}}{\SysTight}$
    and
    $\many{\derivable{\Phi^{i}_{u}}{\sequT{\Delta_i}{\assign{u}{\sigma_i}}{\cbeta_i}{\cexp_i}{\csize_i}}{\SysTight}}{i \in I}$
    such that $\Gamma'_1 = \ctxtsum{\Gamma_1}{\Delta_i}{i \in I}$, $\cbeta' =
    \cbeta +_{i \in I}{\cbeta_i}$, $\cexp' = \cexp +_{i \in I}{\cexp_i}$ and
    $\csize' = \csize +_{i \in I}{\csize_i}$. Moreover, $y \notin \fv{u}$
    implies $y \notin \dom{\Delta_i}$ for all $i \in I$ and, hence,
    $\Gamma'_1(y) = \Gamma_1(y)$. Finally, by $\ruleDArrowE$ with $y \neq x$ we
    construct the type derivation
    $\derivable{\Phi_{t}}{\sequT{\ctxtres{\Gamma_1}{y}{};\assign{x}{\intertype{\sigma_i}{i \in I}}}{\assign{\termabs{y}{t'}}{\functtype{\Gamma'_1(y)}{\tau'}}}{\cbeta}{\cexp}{\csize}}{\SysTight}$
    and we conclude with $\Gamma = \ctxtres{\Gamma_1}{y}{}$ since:
    \begin{itemize}
       \item $\Gamma' = \ctxtres{\Gamma'_1}{y}{} =
       \ctxtres{(\ctxtsum{\Gamma_1}{\Delta_i}{i \in I})}{y}{} =
       \ctxtsum{\Gamma}{\Delta_i}{i \in I}$.
       
       \item $\cbeta' = \cbeta +_{i \in I}{\cbeta_i}$.
       
       \item $\cexp' = \cexp +_{i \in I}{\cexp_i}$.
       
       \item $\csize' = \csize +_{i \in I}{\csize_i}$.
    \end{itemize}
  \end{enumerate}
  
  \item $t = \termbang{t'}$. Then, $\substitute{x}{u}{t} =
  \termbang{\substitute{x}{u}{t'}}$ and there are two possible cases.
  \begin{enumerate}
    \item $\ruleTBang$. Then, $\tau = \typebang$, $\Gamma' = \emptyset$ and 
    $\cbeta' = \cexp' = \csize' = 0$. We set $I = \emptyset$ (hence,
    $\intertype{\sigma_i}{i \in I} = \emul$), $\Gamma = \emptyset$ and conclude
    by $\ruleTBang$ with
    $\derivable{\Phi_{t}}{\sequT{}{\assign{t}{\typebang}}{0}{0}{0}}{\SysTight}$.
    
    \item $\ruleDBang$. Then, $\tau = \intertype{\tau_j}{j \in J}$,
    $\Gamma' = \ctxtsum{}{\Gamma'_j}{j \in J}$, $\cbeta' =
    +_{j \in J}{\cbeta'_j}$, $\cexp' = 1 +_{j \in J}{\cexp'_j}$, $\csize' =
    +_{j \in J}{\csize'_j}$ with premise
    $\many{\derivable{\Phi^{j}_{\substitute{x}{u}{t'}}}{\sequT{\Gamma'_j}{\assign{\substitute{x}{u}{t'}}{\tau_j}}{\cbeta'_j}{\cexp'_j}{\csize'_j}}{\SysTight}}{j \in J}$.
    By \ih there exist type derivations
    $\derivable{\Phi^{j}_{t'}}{\sequT{\Gamma_j;\assign{x}{\intertype{\sigma_i}{i \in I_j}}}{\assign{t'}{\tau_j}}{\cbeta_j}{\cexp_j}{\csize_j}}{\SysTight}$
    and
    $\many{\derivable{\Phi^{i}_{u}}{\sequT{\Delta_i}{\assign{u}{\sigma_i}}{\cbeta_i}{\cexp_i}{\csize_i}}{\SysTight}}{i \in I_j}$
    (note that all $I_j$ can be assumed pairwise disjoint)
    such that $\Gamma'_j = \ctxtsum{\Gamma_j}{\Delta_i}{i \in I_j}$, $\cbeta'_j
    = \cbeta_j +_{i \in I_j}{\cbeta_i}$, $\cexp'_j = \cexp_j
    +_{i \in I_j}{\cexp_i}$ and $\csize'_j = \csize_j +_{i \in I_j}{\csize_i}$
    for each $j \in J$. Then, by $\ruleDBang$ we have
    $\derivable{\Phi_{t}}{\sequT{\ctxtsum{}{\Gamma_j;\assign{x}{\intertype{\sigma_i}{i \in I_j}}}{j \in J}}{\assign{\termbang{t'}}{\intertype{\tau_j}{j \in J}}}{+_{j \in J}{\cbeta_j}}{1+_{j \in J}{\cexp_j}}{+_{j \in J}{\csize_j}}}{\SysTight}$
    and we conclude with $\Gamma = \ctxtsum{}{\Gamma_j}{j \in J}$,
    $I = \bigcup_{j \in J}{I_j}$, $\cbeta = +_{j \in J}{\cbeta_j}$,
    $\cexp = 1 +_{j \in J}{\cexp_j}$, $\csize = +_{j \in J}{\csize_j}$ since:
    \begin{itemize}
       \item $\Gamma' = \ctxtsum{}{\Gamma'_j}{j \in J} =
       \ctxtsum{}{(\ctxtsum{\Gamma_j}{\Delta_i}{i \in I_j})}{j \in J} =
       \ctxtsum{\Gamma}{\Delta_i}{i \in I}$.
       
       \item $\cbeta' = +_{j \in J}{\cbeta'_j} =
       +_{j \in J}{(\cbeta_j +_{i \in I_j}{\cbeta_i})} = \cbeta
       +_{i \in I}{\cbeta_i}$.
       
       \item $\cexp' = 1 +_{j \in J}{\cexp'_j} = 1
       +_{j \in J}{(\cexp_j +_{i \in I_j}{\cexp_i})} = \cexp
       +_{i \in I}{\cexp_i}$.
       
       \item $\csize' = +_{j \in J}{\csize'_j} =
       +_{j \in J}{(\csize_j +_{i \in I_j}{\csize_i})} = \csize
       +_{i \in I}{\csize_i}$.
    \end{itemize}
  \end{enumerate}
  
  \item $t = \termder{t'}$. Then, $\substitute{x}{u}{t} =
  \termder{(\substitute{x}{u}{t'})}$ and there are two possible cases.
  \begin{enumerate}
    \item $\ruleTDer$. Then, $\tau = \typeneutral$ and $\csize' = \csize'_1 +
    1$ with a type derivation for the premise
    $\derivable{\Phi_{\substitute{x}{u}{t'}}}{\sequT{\Gamma'}{\assign{\substitute{x}{u}{t'}}{\typeneutral}}{\cbeta'}{\cexp'}{\csize'_1}}{\SysTight}$.
    By \ih there exist type derivations
    $\derivable{\Phi_{t'}}{\sequT{\Gamma;\assign{x}{\intertype{\sigma_i}{i \in I}}}{\assign{t'}{\typeneutral}}{\cbeta}{\cexp}{\csize_1}}{\SysTight}$
    and
    $\many{\derivable{\Phi^{i}_{u}}{\sequT{\Delta_i}{\assign{u}{\sigma_i}}{\cbeta_i}{\cexp_i}{\csize_i}}{\SysTight}}{i \in I}$
    such that $\Gamma' = \ctxtsum{\Gamma}{\Delta_i}{i \in I}$, $\cbeta' =
    \cbeta +_{i \in I}{\cbeta_i}$, $\cexp' = \cexp +_{i \in I}{\cexp_i}$ and
    $\csize'_1 = \csize_1 +_{i \in I}{\csize_i}$. Finally, we conclude by
    $\ruleTDer$ with
    $\derivable{\Phi_{t}}{\sequT{\Gamma;\assign{x}{\intertype{\sigma_i}{i \in I}}}{\assign{\termder{t'}}{\typeneutral}}{\cbeta}{\cexp}{\csize_1+1}}{\SysTight}$
    and $\csize = \csize_1 + 1$.
    
    \item $\ruleDDer$. Then, we have the premise
    $\derivable{\Phi_{\substitute{x}{u}{t'}}}{\sequT{\Gamma'}{\assign{\substitute{x}{u}{t'}}{\multiset{\tau}}}{\cbeta'}{\cexp'}{\csize'}}{\SysTight}$.
    By \ih there exist type derivations
    $\derivable{\Phi_{t'}}{\sequT{\Gamma;\assign{x}{\intertype{\sigma_i}{i \in I}}}{\assign{t'}{\multiset{\tau}}}{\cbeta}{\cexp}{\csize}}{\SysTight}$
    and
    $\many{\derivable{\Phi^{i}_{u}}{\sequT{\Delta_i}{\assign{u}{\sigma_i}}{\cbeta_i}{\cexp_i}{\csize_i}}{\SysTight}}{i \in I}$
    such that $\Gamma' = \ctxtsum{\Gamma}{\Delta_i}{i \in I}$, $\cbeta' =
    \cbeta +_{i \in I}{\cbeta_i}$, $\cexp' = \cexp +_{i \in I}{\cexp_i}$ and
    $\csize' = \csize +_{i \in I}{\csize_i}$. Finally, we conclude by
    $\ruleDDer$ with
    $\derivable{\Phi_{t}}{\sequT{\Gamma;\assign{x}{\intertype{\sigma_i}{i \in I}}}{\assign{\termder{t'}}{\tau}}{\cbeta}{\cexp}{\csize}}{\SysTight}$.
  \end{enumerate}
  
  \item $t = \termsubs{y}{t_2}{t_1}$. By $\alpha$-conversion we assume $y \neq x$,
  $y \notin \fv{t_2}$ and $y \notin \fv{u}$. Then, $\substitute{x}{u}{t} =
  \termsubs{y}{\substitute{x}{u}{t_2}}{\substitute{x}{u}{t_1}}$ and there are two
  possible cases.
  \begin{enumerate}
    \item $\ruleTESubs$. Then, $\Gamma' =
    \ctxtsum{\ctxtres{\Gamma'_1}{y}{}}{\Gamma'_2}{}$, $\cbeta' = \cbeta'_1 +
    \cbeta'_2$, $\cexp' = \cexp'_1 + \cexp'_2$ and $\csize' = \csize'_1 +
    \csize'_2 $ with premises
    $\derivable{\Phi_{\substitute{x}{u}{t_1}}}{\sequT{\Gamma'_1}{\assign{\substitute{x}{u}{t_1}}{\tau}}{\cbeta'_1}{\cexp'_1}{\csize'_1}}{\SysTight}$
    and
    $\derivable{\Phi_{\substitute{x}{u}{t_2}}}{\sequT{\Gamma'_2}{\assign{\substitute{x}{u}{t_2}}{\typeneutral}}{\cbeta'_2}{\cexp'_2}{\csize'_2}}{\SysTight}$
    where $\ptight{\Gamma'_1(y)}$ holds. By \ih there exist
    $\derivable{\Phi_{t_1}}{\sequT{\Gamma_1;\assign{x}{\intertype{\sigma_i}{i \in I_1}}}{\assign{t_1}{\tau}}{\cbeta_1}{\cexp_1}{\csize_1}}{\SysTight}$, \ \ 
    $\many{\derivable{\Phi^{i}_{u}}{\sequT{\Delta_i}{\assign{u}{\sigma_i}}{\cbeta_i}{\cexp_i}{\csize_i}}{\SysTight}}{i \in I_1}$, \ \ 
    $\derivable{\Phi_{t_2}}{\sequT{\Gamma_2;\assign{x}{\intertype{\sigma_i}{i \in I_2}}}{\assign{t_2}{\typeneutral}}{\cbeta_2}{\cexp_2}{\csize_2}}{\SysTight}$ \ \ 
    and
    $\many{\derivable{\Phi^{i}_{u}}{\sequT{\Delta_i}{\assign{u}{\sigma_i}}{\cbeta_i}{\cexp_i}{\csize_i}}{\SysTight}}{i \in I_2}$
    (note that $I_1$ and $I_2$ can be assumed disjoint)
    such that $\Gamma'_1 = \ctxtsum{\Gamma_1}{\Delta_i}{i \in I_1}$, $\Gamma'_2
    = \ctxtsum{\Gamma_2}{\Delta_i}{i \in I_2}$, $\cbeta'_1 = \cbeta_1
    +_{i \in I_1}{\cbeta_i}$, $\cbeta'_2 = \cbeta_2 +_{i \in I_2}{\cbeta_i}$,
    $\cexp'_1 = \cexp_1 +_{i \in I_1}{\cexp_i}$, $\cexp'_2 = \cexp_2
    +_{i \in I_2}{\cexp_i}$, $\csize'_1 = \csize_1 +_{i \in I_1}{\csize_i}$ and
    $\csize'_2 = \csize_2 +_{i \in I_2}{\csize_i}$. Moreover, $y \notin \fv{u}$
    implies $y \notin \dom{\Delta_i}$ for all $i \in I_1 \cup I_2$ while $y \notin
    \fv{t_2}$ implies $y \notin \dom{\Gamma_2}$. Hence, $\Gamma'_1(y) = \Gamma_1(y)$
    and $\ptight{\Gamma_1(y)}$ holds as well. Finally, by $\ruleTESubs$ with
    $y \neq x$ we construct the type derivation
    $\derivable{\Phi_{t}}{\sequT{\ctxtsum{\ctxtres{\Gamma_1}{y}{}}{\Gamma_2}{};\assign{x}{\intertype{\sigma_i}{i \in I_1 \cup I_2}}}{\assign{\termsubs{y}{t_2}{t_1}}{\tau}}{\cbeta_1+\cbeta_2}{\cexp_1+\cexp_2}{\csize_1+\csize_2}}{\SysTight}$
    and we conclude with $\Gamma =
    \ctxtsum{\ctxtres{\Gamma_1}{y}{}}{\Gamma_2}{}$, $I = I_1 \cup I_2$, $\cbeta =
    \cbeta_1 + \cbeta_2$, $\cexp = \cexp_1 + \cexp_2$ and $\csize = \csize_1 +
    \csize_2 $ since:
    \begin{itemize}
       \item $\Gamma' = \ctxtsum{\ctxtres{\Gamma'_1}{y}{}}{\Gamma'_2}{} =
       \ctxtsum{\ctxtres{(\ctxtsum{\Gamma_1}{\Delta_i}{i \in I_1})}{y}{}}{\ctxtsum{\Gamma_2}{\Delta_i}{i \in I_2}}{}
       $, and this last context is equal to $\ctxtsum{\ctxtsum{\ctxtres{\Gamma_1}{y}{}}{\Gamma_2}{}}{\Delta_i}{i \in I}
       = \ctxtsum{\Gamma}{\Delta_i}{i \in I}$.
       
       \item $\cbeta' = \cbeta'_1 + \cbeta'_2 = \cbeta_1
       +_{i \in I_1}{\cbeta_i} + \cbeta_2 +_{i \in I_2}{\cbeta_i} = \cbeta
       +_{i \in I}{\cbeta_i}$.
       
       \item $\cexp' = \cexp'_1 + \cexp'_2 = \cexp_1
       +_{i \in I_1}{\cexp_i} + \cexp_2 +_{i \in I_2}{\cexp_i} = \cexp
       +_{i \in I}{\cexp_i}$.
       
       \item $\csize' = \csize'_1 + \csize'_2 = \csize_1
       +_{i \in I_1}{\csize_i} + \csize_2 +_{i \in I_2}{\csize_i} = \csize
       +_{i \in I}{\csize_i}$.
    \end{itemize}
    
    \item $\ruleDESubs$. Then, $\Gamma' =
    \ctxtsum{\ctxtres{\Gamma'_1}{y}{}}{\Gamma'_2}{}$, $\cbeta' = \cbeta'_1 +
    \cbeta'_2$, $\cexp' = \cexp'_1 + \cexp'_2$ and $\csize' = \csize'_1 +
    \csize'_2$ with premises
    $\derivable{\Phi_{\substitute{x}{u}{t_1}}}{\sequT{\Gamma'_1}{\assign{\substitute{x}{u}{t_1}}{\tau}}{\cbeta'_1}{\cexp'_1}{\csize'_1}}{\SysTight}$
    and
    $\derivable{\Phi_{\substitute{x}{u}{t_2}}}{\sequT{\Gamma'_2}{\assign{\substitute{x}{u}{t_2}}{\Gamma'_1(y)}}{\cbeta'_2}{\cexp'_2}{\csize'_2}}{\SysTight}$.
    Thus, by \ih there exist type derivations
    $\derivable{\Phi_{t_1}}{\sequT{\Gamma_1;\assign{x}{\intertype{\sigma_i}{i \in I_1}}}{\assign{t_1}{\tau}}{\cbeta_1}{\cexp_1}{\csize_1}}{\SysTight}$, \ \ 
    $\many{\derivable{\Phi^{i}_{u}}{\sequT{\Delta_i}{\assign{u}{\sigma_i}}{\cbeta_i}{\cexp_i}{\csize_i}}{\SysTight}}{i \in I_1}$, \ \ 
    $\derivable{\Phi_{t_2}}{\sequT{\Gamma_2;\assign{x}{\intertype{\sigma_i}{i \in I_2}}}{\assign{t_2}{\Gamma'_1(y)}}{\cbeta_2}{\cexp_2}{\csize_2}}{\SysTight}$
    and
    $\many{\derivable{\Phi^{i}_{u}}{\sequT{\Delta_i}{\assign{u}{\sigma_i}}{\cbeta_i}{\cexp_i}{\csize_i}}{\SysTight}}{i \in I_2}$
    (note that $I_1$ and $I_2$ can be assumed disjoint)
    such that $\Gamma'_1 = \ctxtsum{\Gamma_1}{\Delta_i}{i \in I_1}$, $\Gamma'_2
    = \ctxtsum{\Gamma_2}{\Delta_i}{i \in I_2}$, $\cbeta'_1 = \cbeta_1
    +_{i \in I_1}{\cbeta_i}$, $\cbeta'_2 = \cbeta_2 +_{i \in I_2}{\cbeta_i}$,
    $\cexp'_1 = \cexp_1 +_{i \in I_1}{\cexp_i}$, $\cexp'_2 = \cexp_2
    +_{i \in I_2}{\cexp_i}$, $\csize'_1 = \csize_1 +_{i \in I_1}{\csize_i}$ and
    $\csize'_2 = \csize_2 +_{i \in I_2}{\csize_i}$. Moreover, $y \notin \fv{u}$
    implies $y \notin \dom{\Delta_i}$ for all $i \in I_1 \cup I_2$ while $y \notin
    \fv{t_2}$ implies $y \notin \dom{\Gamma_2}$. Hence, $\Gamma'_1(y) = \Gamma_1(y)$.
    Finally, by $\ruleDESubs$ with $y \neq x$ we have
    $\derivable{\Phi_{t}}{\sequT{\ctxtsum{\ctxtres{\Gamma_1}{y}{}}{\Gamma_2}{};\assign{x}{\intertype{\sigma_i}{i \in I_1 \cup I_2}}}{\assign{\termsubs{y}{t_2}{t_1}}{\tau}}{\cbeta_1+\cbeta_2}{\cexp_1+\cexp_2}{\csize_1+\csize_2}}{\SysTight}$
    and we conclude with $\Gamma =
    \ctxtsum{\ctxtres{\Gamma_1}{y}{}}{\Gamma_2}{}$, $I = I_1 \cup I_2$, $\cbeta =
    \cbeta_1 + \cbeta_2$, $\cexp = \cexp_1 + \cexp_2$ and $\csize = \csize_1 +
    \csize_2$ since:
    \begin{itemize}
       \item $\Gamma' = \ctxtsum{\ctxtres{\Gamma'_1}{y}{}}{\Gamma'_2}{} =
       \ctxtsum{\ctxtres{(\ctxtsum{\Gamma_1}{\Delta_i}{i \in I_1})}{y}{}}{\ctxtsum{\Gamma_2}{\Delta_i}{i \in I_2}}{}
       $ and this last context is equal to $
       \ctxtsum{\ctxtsum{\ctxtres{\Gamma_1}{y}{}}{\Gamma_2}{}}{\Delta_i}{i \in I}
       = \ctxtsum{\Gamma}{\Delta_i}{i \in I}$.
       
       \item $\cbeta' = \cbeta'_1 + \cbeta'_2 = \cbeta_1
       +_{i \in I_1}{\cbeta_i} + \cbeta_2 +_{i \in I_2}{\cbeta_i} = \cbeta
       +_{i \in I}{\cbeta_i}$.
       
       \item $\cexp' = \cexp'_1 + \cexp'_2 = \cexp_1
       +_{i \in I_1}{\cexp_i} + \cexp_2 +_{i \in I_2}{\cexp_i} = \cexp
       +_{i \in I}{\cexp_i}$.
       
       \item $\csize' = \csize'_1 + \csize'_2 = \csize_1
       +_{i \in I_1}{\csize_i} + \csize_2 +_{i \in I_2}{\csize_i} = \csize
       +_{i \in I}{\csize_i}$.
    \end{itemize}
  \end{enumerate}
\end{itemize}
\end{proof}

\begin{lemma}[Exact Subject Expansion]
\label{l:QSE-tight}
Let
$\derivable{\Phi'}{\sequT{\Gamma}{\assign{t'}{\sigma}}{\cbeta'}{\cexp'}{\csize}}{\SysTight}$
be a tight derivation. If $t \lored t'$, then there exists
$\derivable{\Phi}{\sequT{\Gamma}{\assign{t}{\sigma}}{\cbeta}{\cexp}{\csize}}{\SysTight}$
such that
\begin{itemize}
  \item $\cbeta' = \cbeta - 1$ and $\cexp' = \cexp$ if $t \lored t'$ is an
  $\mStep$-step.
  \item $\cexp' = \cexp - 1$ and $\cbeta' = \cbeta$ if $t \lored t'$ is an
  $\eStep$-step.
\end{itemize}
\end{lemma}

\begin{proof}
By induction on $t \lored t'$.
\begin{itemize}
  \item $t = \termapp{\ctxtapp{\ctxt{L}}{\termabs{x}{s}}}{u} \lored
  \ctxtapp{\ctxt{L}}{\termsubs{x}{u}{s}} = t'$. We reason by induction on
  $\ctxt{L}$.
  \begin{itemize}
    \item $\ctxt{L} = \Box$. There are two cases depending on the last rule
    of $\Phi'$.
    \begin{enumerate}
      \item $\ruleTESubs$. Then, $\Phi'$ has the following form: \[
\Rule{
  \sequT{\Gamma'}{\assign{s}{\sigma}}{\cbeta_s}{\cexp_s}{\csize_s}
  \quad
  \sequT{\Delta}{\assign{u}{\typeneutral}}{\cbeta_u}{\cexp_u}{\csize_u}
  \quad
  \ptight{\Gamma'(x)}
}{
  \sequT{\ctxtsum{\ctxtres{\Gamma'}{x}{}}{\Delta}{}}{\assign{\termsubs{x}{u}{s}}{\sigma}}{\cbeta_s+\cbeta_u}{\cexp_s+\cexp_u}{\csize_s+\csize_u}
}{\ruleTESubs}
      \] We can then construct the following derivation $\Phi$: \[\kern-4em
\Rule{
  \Rule{
    \sequT{\Gamma'}{\assign{s}{\sigma}}{\cbeta_s}{\cexp_s}{\csize_s}
  }{
    \sequT{\ctxtres{\Gamma'}{x}{}}{\assign{\termabs{x}{s}}{\functtype{\Gamma'(x)}{\sigma}}}{\cbeta_s}{\cexp_s}{\csize_s}
  }{\ruleDArrowI}
  \quad
  \sequT{\Delta}{\assign{u}{\typeneutral}}{\cbeta_u}{\cexp_u}{\csize_u}
  \quad
  \ptight{\Gamma'(x)}
}{
  \sequT{\ctxtsum{\ctxtres{\Gamma'}{x}{}}{\Delta}{}}{\assign{\termapp{(\termabs{x}{s})}{u}}{\sigma}}{\cbeta_s+\cbeta_u+1}{\cexp_s+\cexp_u}{\csize_s+\csize_u}
}{\ruleBTApp}
      \] The counters verify the expected property.
      
      \item $\ruleDESubs$. Then, $\Phi'$ has the following form: \[
\Rule{
  \sequT{\Gamma'}{\assign{s}{\sigma}}{\cbeta_s}{\cexp_s}{\csize_s}
  \quad
  \sequT{\Delta}{\assign{u}{\Gamma'(x)}}{\cbeta_u}{\cexp_u}{\csize_u}
}{
  \sequT{\ctxtsum{\ctxtres{\Gamma'}{x}{}}{\Delta}{}}{\assign{\termsubs{x}{u}{s}}{\sigma}}{\cbeta_s+\cbeta_u}{\cexp_s+\cexp_u}{\csize_s+\csize_u}
}{\ruleDESubs}
      \] We can then construct the following derivation $\Phi$: \[
\Rule{
  \Rule{
    \sequT{\Gamma'}{\assign{s}{\sigma}}{\cbeta_s}{\cexp_s}{\csize_s}
  }{
    \sequT{\ctxtres{\Gamma'}{x}{}}{\assign{\termabs{x}{s}}{\functtype{\Gamma'(x)}{\sigma}}}{\cbeta_s}{\cexp_s}{\csize_s}
  }{\ruleDArrowI}
  \quad
  \sequT{\Delta}{\assign{u}{\Gamma'(x)}}{\cbeta_u}{\cexp_u}{\csize_u}
}{
  \sequT{\ctxtsum{\ctxtres{\Gamma'}{x}{}}{\Delta}{}}{\assign{\termapp{(\termabs{x}{s})}{u}}{\sigma}}{\cbeta_s+\cbeta_u+1}{\cexp_s+\cexp_u}{\csize_s+\csize_u}
}{\ruleDArrowE}
      \] The counters verify the expected property.
    \end{enumerate}
    
    \item $\ctxt{L}= \termsubs{y}{r}{\ctxt{L'}}$. Immediate from the \ih
  \end{itemize}
  
  \item $t = \termsubs{x}{\ctxtapp{\ctxt{L}}{\termbang{u}}}{s} \lored
  \ctxtapp{\ctxt{L}}{\substitute{x}{u}{s}} = t'$. We reason by induction on
  $\ctxt{L}$.
  \begin{itemize}
    \item $\ctxt{L} = \Box$. Then $\Phi$ has the form
    $\sequT{\Gamma}{\assign{\substitute{x}{u}{s}}{\sigma}}{\cbeta'}{\cexp'}{\csize}$.
    By Lemma~\ref{l:bang:anti-substitution-tight} there exist
    $\sequT{\Gamma';\assign{x}{\intertype{\sigma_i}{i \in I}}}{\assign{s}{\sigma}}{\cbeta_s}{\cexp_s}{\csize_s}$
    and
    $\many{\sequT{\Delta_i}{\assign{u}{\sigma_i}}{\cbeta_i}{\cexp_i}{\csize_i}}{i \in I}$
    such that $\Gamma = \ctxtsum{\Gamma'}{\Delta_i}{i \in I}$,
    $\cbeta' = \cbeta_s +_{i \in I}{\cbeta_i}$, $\cexp' = \cexp_s +_{i \in I}{\cexp_i}$ and
    $\csize = \csize_s +_{i \in I}{\csize_i}$.
    We can then construct the following derivation $\Phi$: \[\kern-2em
\Rule{
  \sequT{\Gamma';\assign{x}{\intertype{\sigma_i}{i \in I}}}{\assign{s}{\sigma}}{\cbeta_s}{\cexp_s}{\csize_s}
  \quad
  \Rule{
    (\sequT{\Delta_i}{\assign{u}{\sigma_i}}{\cbeta_i}{\cexp_i}{\csize_i})_{i \in I}
  }{
    \sequT{\ctxtsum{}{\Delta_i}{i \in I}}{\assign{\termbang{u}}{\intertype{\sigma_i}{i \in I}}}{+_{i \in I}{\cbeta_i}}{1+_{i \in I}{\cexp_i}}{+_{i \in I}{\csize_i}}
  }{\ruleDBang}
}{
  \sequT{\ctxtsum{\Gamma'}{\Delta_i}{i \in I}}{\assign{\termsubs{x}{\termbang{u}}{s}}{\sigma}}{\cbeta_s+_{i \in I}{\cbeta_i}}{1+\cexp_s+_{i \in I}{\cexp_i}}{\csize_s+_{i \in I}{\csize_i}}
}{\ruleDESubs}
    \] The counters verify the expected property.
    
    \item $\ctxt{L}= \termsubs{y}{r}{\ctxt{L'}}$. Immediate from the \ih
  \end{itemize}
  
  \item $t=\termder{(\ctxtapp{\ctxt{L}}{\termbang{s}})} \lored \ctxtapp{\ctxt{L}}{s} = t'$. We reason by induction on $\ctxt{L}$.
  \begin{itemize}
    \item $\ctxt{L} = \Box$. Then, $t' = s$ and from $\Phi'$ we construct the
    following derivation: \[
\Rule{
  \Rule{
    \sequT{\Gamma}{\assign{s}{\sigma}}{\cbeta'}{\cexp'}{\csize}
  }{
    \sequT{\Gamma}{\assign{\termbang{s}}{\intertype{\sigma}{}}}{\cbeta'}{1+\cexp'}{\csize}
  }{\ruleDBang}
}{
  \sequT{\Gamma}{\assign{\termder{\termbang{s}}}{\sigma}}{\cbeta'}{1+\cexp'}{\csize}}{\ruleDDer}
    \] We conclude since the counters verify the expected property.
    
    \item $\ctxt{L}= \termsubs{y}{r}{\ctxt{L'}}$. Immediate from the \ih
  \end{itemize}
  
  \item All the inductive cases for internal reductions are straightforward by \ih
  \end{itemize}
\end{proof}

\begin{theorem}[Completeness for System $\SysTight$]
\label{t:completeness}
If $t \rewriten{\bangweak}^{(\cbeta,\cexp)} p$ with
$\cfnrml{p}$, then there exists a tight type
derivation
$\derivable{\Phi}{\sequT{\Gamma}{\assign{t}{\sigma}}{\cbeta}{\cexp}{\wsize{p}}}{\SysTight}$.
\end{theorem}

\begin{proof}
We prove the statement for $\loredn$ and then conclude for the general notion
of reduction $\rewrite{\bangweak}$ by Theorem~\ref{t:confluence}. Let
$t \loredn^{(\cbeta,\cexp)} p$. We proceed by induction on $\cbeta+\cexp$.

If $\cbeta+\cexp = 0$, then $\cbeta = \cexp = 0$ and thus $t = p$, which
implies $\cfnrml{t}$. Lemma~\ref{l:typing-normal-clash-free} allows us to
conclude.

If $\cbeta+\cexp > 0$, then there exists $t'$ such that $t \loredn^{(1,0)} t'
\loredn^{(\cbeta-1,\cexp)} p$ or $t \loredn^{(0,1)} t'
\loredn^{(\cbeta,\cexp-1)} p$. By \ih there exists a tight derivation
$\derivable{\Phi'}{\sequT{\Gamma}{\assign{t'}{\sigma}}{\cbeta'}{\cexp'}{\wsize{p}}}{\SysTight}$
such that $\cbeta' + \cexp' = \cbeta + \cexp -1$. Lemma~\ref{l:QSE-tight} gives a
tight derivation
$\derivable{\Phi}{\sequT{\Gamma}{\assign{t}{\sigma}}{\cbeta''}{\cexp''}{\wsize{p}}}{\SysTight}$
such that $\cbeta'' + \cexp'' = \cbeta' + \cexp' + 1$. We then have $\cbeta'' +
\cexp'' = \cbeta + \cexp$. The fact that $\cbeta'' = \cbeta$ and $\cexp'' =
\cexp$ holds by a simple case analysis.
\end{proof}

The main results can be illustrated by the term $t_0 =
\termapp{\termapp{\termder{(\termbang{\Kterm})}}{(\termbang{\id})}}{(\termbang{\Omega})}$
in Sec.~\ref{s:bang}, which normalises in $2$ multiplicative steps and $3$
exponential steps to a $\bangweak$-normal form of $\bangweak$-size $1$. A tight
derivation for $t_0$ with appropriate counters $(2,3,1)$ is given in
Example~\ref{example:t0-tight}.


\section{Conclusion}
\label{s:conclusion}

This paper gives a fresh view of the Bang Calculus, a formalism introduced by
T.~Ehrhard to study the relation between CBPV and Linear Logic.

Our reduction relation integrates permutative conversions inside the logical
original formulation of~\cite{Ehrhard16}, thus recovering soundness, \ie
avoiding mismatches between terms in normal form that are semantically
non-terminating. In contrast to~\cite{EhrhardG16}, which models permutative
conversions as $\sigma$-reduction rules by paying the cost of losing
confluence, our \emph{at a distance} formulation yields a confluent reduction
system.

We then define two non-idempotent intersection type systems
for our calculus. On the one hand, system $\SysBang$
provides upper bounds for the length of normalising sequences plus the
size of normal forms. Moreover, it captures typed CBN and CBV.
On the other hand, the quantitative
system $\SysBang$ is further refined into system $\SysTight$, being
able to provide \emph{exact} bounds for normalising sequences and size
of normal forms, independently. Moreover, our tight system $\SysTight$
is able to \emph{discriminate} between different kind of steps
performed to normalise terms.

Concerning related works, several points should be noticed respect to
the closest~\cite{GuerrieriM18}. First of all, our CBV translation recovers the
expected property of preserving the normal forms, as stated in
Lemma~\ref{l:preservation-cbn-cbv}, whereas normal forms in CBV do not
necessarily translate to normal forms in the Bang calculus
in~\cite{GuerrieriM18}. Moreover the CBV embedding in~\cite{GuerrieriM18} is
not complete w.r.t. their CBV type system, \ie there exists a $\lambda$-term $t$
such that $\sequ{\Gamma}{\assign{\cbv{t}}{\sigma}}$ is derivable in $\SysBang$
but $\sequ{\Gamma}{\assign{t}{\sigma}}$ is not derivable in their CBV system
(see~\cite{GuerrieriM18}, Proposition 16). The completeness property in our
framework is stated as Theorem~\ref{t:soundness-completeness-cbn-cbv}. As a
matter of fact in~\cite{GuerrieriM18} the authors remark that the
incompleteness of the CBV embedding is due in particular to the fact that the
CBV type system they use, which is the canonical one stemming from the
relational model of the CBV $\lambda$-calculus defined in \cite{EhrhardG16},
assigns multiset types to \emph{all} the $\lambda$-terms, and in particular to
all the applications, whereas an application in the target of the CBV
translation may well get a non-multiset type in their Bang calculus type
system. They propose at the end of the paper an alternative type system,
targeting the range of their CBV translation, which is essentialy the one we
have adopted here. Nevertheless, they leave (the model theoretic version of)
this question of incompleteness for future work.

Several topics deserve future attention. One of them is the study of
\emph{strong} reduction for the $\BangRev$-calculus, which allows us to
reduce terms under \emph{all} the constructors, including
$\termbang{}$. Another challenging problem is to relate \emph{tight}
typing in CBN/CBV with \emph{tight} typing in our calculus, thus
providing an exact correspondence between (CBN/CBV) reduction steps
and $\BangRev$-reduction steps.

\section{Acknowledgments}

We are grateful to Giulio Guerrieri, Giulio Manzonetto and Victor
Arrial who provided valuable feedback and insightful discussions
during the course of this study. Also, we would like to thank the
reviewers for their constructive comments and suggestions which helped
to improve the manuscript. This work has been partially supported by
LIA INFINIS/IRP SINFIN, the ECOS-Sud program PA17C01, and the ANR COCA
HOLA.

\bibliographystyle{plain}
\bibliography{biblio} 

\end{document}